\documentclass[conference]{IEEEtran}
\IEEEoverridecommandlockouts
 
\usepackage{amsmath}
\usepackage{amsfonts}
\usepackage{amssymb}
\usepackage{amsthm}
\usepackage{url}
\newtheorem{definition}{Definition}
\newtheorem{theorem}{Theorem}
\newtheorem{lemma}{Lemma}

\usepackage{pifont}% http://ctan.org/pkg/pifont - for \cmark and \xmark
\newcommand{\cmark}{\ding{51}}%
\newcommand{\xmark}{\ding{55}}%

\usepackage{amsmath}
\usepackage{amsfonts}
\usepackage{xcolor}
\usepackage{colortbl} %coloring for tables
\usepackage{xspace} %for spacing after macros
\usepackage{paralist} %for compact enumeration enviornments

\usepackage{booktabs}  %fancy table formatting
\usepackage{csvsimple} %for loading csv data

\usepackage{caption} %subfigure
\usepackage{subcaption} %subfigure

\usepackage{bm}

\usepackage{pgf}
\usepackage{pgfplots}
\usepackage{pgfplotstable}
\usepackage{multirow}
\usepackage{subcaption}

\usepackage{circledsteps}
\tikzset{/csteps/outer color=blue}

\usepackage[normalem]{ulem}

\usepackage{xpatch}
\usepackage{textcase}
\makeatletter
\xpatchcmd{\@sect}{\uppercase}{\MakeTextUppercase}{}{}
\xpatchcmd{\@sect}{\uppercase}{\MakeTextUppercase}{}{}
\makeatother

\usepackage[explicit]{titlesec}

\usepackage[numbers,sort&compress]{natbib}

\titleformat{\paragraph}[runin]
{\normalfont\bfseries}
{}
{0pt}
{#1}

\titlespacing{\paragraph}
{0pt} %the indent (horizontal) space
{1ex plus .1ex minus .2ex}
{1ex} %spacing between paragraph title and text

\DeclareMathOperator{\precision}{\textit{prec}}
\DeclareMathOperator{\recall}{\textit{rcl}}

\DeclareMathOperator{\perfectFD}{\textit{PFD}}
\DeclareMathOperator{\designAFD}{\textit{AFD}}
\DeclareMathOperator{\class}{\textit{DISC}}
\DeclareMathOperator{\rankAtMaxRecall}{\textit{r@mr}}

\newcommand{\copykind}{\textsf{copy}\xspace}
\newcommand{\typokind}{\textsf{typo}\xspace}
\newcommand{\boguskind}{\textsf{bogus}\xspace}

\newcommand{\smallerskip}{\smallskip}
\newcommand{\simple}{\textsc{Violation}\xspace}  %Fraction of violations, previouly called  "simple"
\newcommand{\shannon}{\textsc{Shannon}\xspace}
\newcommand{\logical}{\textsc{Logical}\xspace}
\newcommand{\ssimple}{V\xspace} % short versions
\newcommand{\sshannon}{S\xspace}
\newcommand{\slogical}{L\xspace}

\newcommand{\BenchNumTables}{10\xspace}

\newcommand{\BenchTotalDesignFDs}{{143}\xspace}
\newcommand{\BenchTotalPerfectFDs}{{126}\xspace}
\newcommand{\BenchTotalAFDs}{{17}\xspace}
\newcommand{\BenchTotalAFDsRWDErrorLevelOne}{{39}\xspace}

\newcommand{\TotalNumberOfInvalidFDs}{{1634}\xspace}

\newcommand{\adultDataNumber}{$R_1$}
\newcommand{\claimsDataNumber}{$R_2$}
\newcommand{\dblpDataNumber}{$R_3$}

\newcommand{\hospitalDataNumber}{$R_4$}
\newcommand{\taxDataNumber}{$R_5$}
\newcommand{\gathAgentDataNumber}{$R_6$}
\newcommand{\gathAreaDataNumber}{$R_7$}
\newcommand{\gathDataNumber}{$R_8$}
\newcommand{\identTaxonDataNumber}{$R_9$}
\newcommand{\identDataNumber}{$R_{10}$}

\newcommand{\LHSuniqueness}{LHS-uniqueness\xspace}

\newcommand{\RHSskew}{RHS-skew\xspace}

\newcommand{\rev}[1]{#1}

\newcommand{\schema}{\ensuremath{\Delta}}
\newcommand{\size}[1]{\ensuremath{|#1|}}
\newcommand{\tuples}[1]{\ensuremath{\mathcal{T}}}
\DeclareMathOperator{\dom}{\textit{dom}}

\newcommand{\seq}[1]{\ensuremath{\bm{#1}}}
\newcommand{\nats}{\ensuremath{\mathbb{N}}}
\newcommand{\prob}{\ensuremath{p}}
\newcommand{\entrop}{\ensuremath{H}}
\newcommand{\logentrop}{\ensuremath{h}}
\newcommand{\mut}{\ensuremath{I}}
\newcommand{\restr}[2]{\ensuremath{{#1}|_{#2}}}
\newcommand{\afd}{AFD\xspace}
\newcommand{\afds}{AFDs\xspace}
\newcommand{\fd}{FD\xspace}
\newcommand{\fds}{FDs\xspace}

\newcommand{\FI}{\ensuremath{\text{FI}}\xspace}
\newcommand{\RFI}{\ensuremath{\text{RFI}}\xspace}
\newcommand{\RFIplus}{\ensuremath{\RFI^+}\xspace}
\newcommand{\RFInorm}{\ensuremath{\RFI^{'+}}\xspace}
\newcommand{\SFI}{\ensuremath{\text{SFI}}\xspace}
\newcommand{\SHANNONGONE}{\ensuremath{g_1^{\sshannon}}}
\newcommand{\muplus}{\ensuremath{\mu^+}}

\DeclareMathOperator{\pdep}{\textit{pdep}}
\DeclareMathOperator{\expect}{\mathbb{E}}

\newcommand{\w}{\ensuremath{\seq w}}
\newcommand{\x}{\ensuremath{\seq x}}
\newcommand{\y}{\ensuremath{\seq y}}

\newcommand{\W}{\ensuremath{\seq W}}
\newcommand{\X}{\ensuremath{\seq X}}
\newcommand{\Y}{\ensuremath{\seq Y}}
\newcommand{\Z}{\ensuremath{\seq Z}}

\newcommand{\rwd}{\textsc{rwd}\xspace}
\newcommand{\nrwd}{\ensuremath{\rwd^{\text{e}}}\xspace} % polluted RWD
\newcommand{\rwdminus}{\ensuremath{\rwd^-}\xspace} % limited dataset
\newcommand{\syn}{\text{SYN}\xspace}
\newcommand{\synzero}{\ensuremath{\textsc{Err}}\xspace}
\newcommand{\synnoise}{\synzero}
\newcommand{\synlhs}{\ensuremath{\textsc{Uniq}}\xspace}
\newcommand{\synunique}{\synlhs}
\newcommand{\synrhs}{\ensuremath{\textsc{Skew}}\xspace}
\newcommand{\synskew}{\synrhs}

\newcommand{\nonfd}{-}
\newcommand{\smallfd}{+}

\newcommand{\benchmarkfd}{\benchmark^{\smallfd}}
\newcommand{\benchmarknonfd}{\benchmark^{\nonfd}}
\newcommand{\synnoisefd}{\synnoise^{\smallfd}}
\newcommand{\synnoisenonfd}{\synnoise^{\nonfd}}
\newcommand{\synuniquefd}{\synunique^{\smallfd}}
\newcommand{\synuniquenonfd}{\synunique^{\nonfd}}
\newcommand{\synskewfd}{\synskew^{\smallfd}}
\newcommand{\synskewnonfd}{\synskew^{\nonfd}}

\newcommand{\alg}{\ensuremath{A}}
\newcommand{\benchmark}{\ensuremath{\mathcal{B}}}

\newcommand{\nap}{\_}

\definecolor{plotColor1}{RGB}{ 55,126,184}
\definecolor{plotColor2}{RGB}{ 77,175, 74}
\definecolor{plotColor3}{RGB}{152, 78,163}
\definecolor{plotColor4}{RGB}{255,127,  0}
\definecolor{plotColor5}{RGB}{166, 86, 40}

\newcommand\availabilityurl{https://anonymous.4open.science/r/measuring_approximate_functional_dependencies}

\newcommand{\inFullVersion}[1]{#1}
\newcommand{\inConfVersion}[1]{}

\newcommand{\altVersion}[2]{\inFullVersion{#1}\inConfVersion{#2}}

\begin{document}

\title{Measuring Approximate Functional Dependencies:\\ a Comparative Study\\
	\thanks{We thank Dan Suciu for helpful discussions. S. Vansummeren was supported by the Bijzonder Onderzoeksfonds (BOF) of Hasselt University under Grant No. BOF20ZAP02. This research received funding from the Flemish Government under the “Onderzoeksprogramma Artificiële Intelligentie (AI) Vlaanderen” programme. This work was supported by Research Foundation—Flanders (FWO) for ELIXIR Belgium (I002819N). The resources and services used in this work were provided by the VSC (Flemish Supercomputer Center), funded by the Research Foundation – Flanders (FWO) and the Flemish Government.}
}
% comment Marcel: I have seen a couple of papers at ICDE last year that used this format for authors, which saves use quite some space
\author{%
	\IEEEauthorblockN{%
		Marcel Parciak\IEEEauthorrefmark{1}\IEEEauthorrefmark{2},
		Sebastiaan Weytjens\IEEEauthorrefmark{2},
		Niel Hens\IEEEauthorrefmark{2},
		Frank Neven\IEEEauthorrefmark{2},
		Liesbet M. Peeters\IEEEauthorrefmark{1}\IEEEauthorrefmark{2} and
		Stijn Vansummeren\IEEEauthorrefmark{2}}
	\IEEEauthorblockA{\IEEEauthorrefmark{1}UHasselt, BIOMED, Agoralaan, 3590 Diepenbeek, Belgium}
	\IEEEauthorblockA{\IEEEauthorrefmark{2}UHasselt, Data Science Institute, Agoralaan, 3590 Diepenbeek, Belgium}
}

% original format (six single blocks)
%
% \author{\IEEEauthorblockN{Marcel Parciak}
% \IEEEauthorblockA{\textit{BIOMED \& Data Science Institute} \\
% \textit{UHasselt --- Hasselt University}\\
% Diepenbeek, Belgium\\
% \and
% \IEEEauthorblockN{Sebastiaan Weytjens}
% \IEEEauthorblockA{\textit{Data Science Institute} \\
% \textit{UHasselt --- Hasselt University}\\
% Diepenbeek, Belgium\\
% \and
% \IEEEauthorblockN{Niel Hens}
% \IEEEauthorblockA{\textit{Data Science Institute} \\
% \textit{UHasselt --- Hasselt University}\\
% Diepenbeek, Belgium\\
% \and
% \IEEEauthorblockN{Frank Neven}
% \IEEEauthorblockA{\textit{Data Science Institute} \\
% \textit{UHasselt --- Hasselt University}\\
% Diepenbeek, Belgium\\
% \and
% \IEEEauthorblockN{Liesbet M. Peeters}
% \IEEEauthorblockA{\textit{BIOMED \& Data Science Institute} \\
% \textit{UHasselt --- Hasselt University}\\
% Diepenbeek, Belgium\\
% \and
% \IEEEauthorblockN{Stijn Vansummeren}
% \IEEEauthorblockA{\textit{Data Science Institute} \\
% \textit{UHasselt --- Hasselt University}\\
% Diepenbeek, Belgium\\

\maketitle
\begin{abstract}
  Approximate functional dependencies (AFDs) are functional dependencies (FDs)
  that ``almost'' hold in a relation. While various measures have been proposed
  to quantify the level to which an \fd holds approximately, they are difficult
  to compare and it is unclear which measure is preferable when one needs to
  discover \fds in real-world data, i.e., data that only approximately satisfies
  the \fd. In response, this paper formally and qualitatively compares AFD
  measures. We obtain a formal comparison through a novel presentation of
  measures in terms of Shannon and logical entropy. 
  % removed to remove this aspect from the "spotlight"
  % Through our comparison, we identify two novel, previously unexplored measures. 
  Qualitatively, we perform a sensitivity analysis w.r.t.\ structural properties
  of input relations and quantitatively study the effectiveness of AFD measures
  for ranking \afds on real world data.  Based on this analysis, we give clear
  recommendations for the AFD measures to use in practice.
%  identify three measures that are very effective in ranking AFDs 
%  We find that the little-known measure $\mu$, the widely-known measure $g_3$
%  and an unbiased variant of the Shannon entropy-based measure $\FI$, called
%  $\RFI$, are the most effective, but the latter two only when properly
%  normalized---which is not always done in the literature. Interestingly, the
%  normalization of $\RFI$ is novel and a direct consequence of our formal
%  comparison of measures. Ultimately, we find that $\mu$ is the preferred
%  measure as it has a computational advantage over $\RFI$ and, unlike $g_3$, it
%  is not sensitive to right-hand side data skew.
  % We also observe that, perhaps
  % contrary to popular belief, it suffices to only inspect a small number of
  % highly ranked approximate FDs to recover the true design FDs that were
  % obscured by errors.
\end{abstract}


%%% Local Variables:
%%% mode: latex
%%% TeX-master: "main"
%%% End:

\section{Introduction}
\label{sec:intro}
There are many benefits to knowing the set of functional dependencies (\fds for
short) that are expected to hold in the instances of a database schema: \fds aid
in ensuring data consistency and help in data
cleaning~\cite{DBLP:conf/icde/ChuIP13,DBLP:journals/pvldb/RekatsinasCIR17};
facilitate data integration~\cite{DBLP:conf/webdb/WangDSFH09}; and can be
exploited for query
optimization~\cite{DBLP:journals/pvldb/LiuXDE16,KossmannPapenbrockEtAl2022a},
among other tasks.  In many data science scenarios, however, the set of design
\fds is unknown or incomplete~\cite{DBLP:conf/sigmod/PapenbrockN16}. As such, a
variety of techniques have been proposed to reverse engineer this set of design
\fds from a given relation instance% : see e.g.,
~\cite{DBLP:journals/pvldb/PapenbrockEMNRZ15,
  DBLP:conf/edbt/SchirmerP0NHMN19,
  DBLP:journals/pvldb/Berti-EquilleHN18,
  DBLP:conf/cikm/BleifussBFRW0PN16,DBLP:conf/sigmod/PapenbrockN16}.

Reverse engineering is particularly challenging in the practical setting
where the relation instance itself does not satisfy the target set of design
\fds
.  This may happen, e.g.,
due to errors during data entry, resulting in a corrupt instance.  To recover
the set of design \fds in such a setting, it clearly does not suffice to simply
enumerate the \fds that are satisfied in the relation instance. Instead, one must also consider \emph{approximate functional dependencies}
(\afds), that is, \fds that ``almost hold'' in the relation.
A key decision to then make is when an \fd ``almost'' holds. This decision is
reflected in the adoption of an \emph{\afd measure}, which formally quantifies
the extent to which an \fd holds approximately in a given relation by
attributing a score in the interval $[0,1]$. Higher values indicate a higher
degree of \fd satisfaction. As such, an \afd measure provides a way of
\emph{ranking} the search space of all possible \fds where higher-scoring \fds
are ranked before lower-scoring ones. A good \afd measure, then, is one that
ranks \fds in the relation's target set of design \fds higher than those that
are not in the target set, and does this consistently for relation instances
that occur in practice. Given a good AFD measure one may discover the set of design FDs by ranking the search space, and returning all FDs smaller than a given threshold.

Many \afd measures have been proposed
in the literature over the past decades~\cite{DBLP:journals/tcs/KivinenM95,
  DBLP:journals/is/GiannellaR04,DBLP:conf/kdd/MandrosBV17,DBLP:journals/kais/MandrosBV20,DBLP:conf/kdd/PennerathMV20,piatetsky1993,DBLP:conf/webdb/WangDSFH09,DBLP:conf/sigmod/IlyasMHBA04,BERZAL2002207}.
Unfortunately, these measures vary widely in nature and there has been
little study so far in comparing them. As such there is no clear guideline for
deciding which \afd measure(s) to adopt.  In response, we aim to address the following two central questions in this
paper:
\begin{itemize}
  \item[(a)] \emph{How do the different \afd measures compare?}
  \item[(b)] \emph{\rev{Which \afd measure(s) perform best and most consistently when discovering \afds in real-world data?}}
\end{itemize}

\smallerskip
\noindent
\textbf{We adopt the following methodology to answer (a):} In
Section~\ref{sec:measures} we present a survey of known \afd measures in a uniform formal framework.
Specifically, we observe that there are three classes of measures: (i) measures that quantify the fraction of violations (we refer to them as \simple); (ii) measures based on Shannon entropy (\shannon)~\cite{DBLP:journals/ipm/Edwards08}; and (iii) measures based on logical entropy (\logical)~\cite{logical-entropy}. We then proceed to evaluate the measures' sensitivity
to structural properties of the input relation, in
Section~\ref{sec:eval-syn}.  Specifically, we study whether a measure can distinguish
between an \fd $\X\to\Y$ in relation instances that were generated to approximately satisfy
the \fd\ versus relation instances where
$\X$ and $\Y$ were randomly generated. We evaluate this on synthetic data,
controlling the following parameters: (1) the error level; (2)
\emph{\LHSuniqueness}: the normalized number of unique values occurring in
\X; and  (3) \emph{\RHSskew}:
the skewness of the distribution of values occurring in $\pi_{\Y}(R)$.\footnote{\rev{We refer to Table~\ref{table:notation} for an overview of the notation used in this paper.}}
This evaluation %hence
sheds light on the measure's distinguishing power as well as on their biases w.r.t. error level, \LHSuniqueness and \RHSskew.
In general, measures are more robust  if they are inversely proportional to the error level and are insensitive to \LHSuniqueness and \RHSskew, as we discuss in Section~\ref{sec:eval-syn}.

\smallerskip
\noindent
\textbf{We adopt the following methodology to answer (b):}
We compare the effectiveness of \afd measures on real-world data.  Because existing benchmarks for exact \fd discovery are designed to gauge algorithmic efficiency and  do not contain the ``ground truth'' set of design \fds to compare against, we  propose a new \afd discovery benchmark. This benchmark, denoted \rwd, is obtained by manually creating design \fds for existing benchmarks. 
Our analysis in Section~\ref{sec:eval-rwd} 
on \rwd shows that well-ranking measures exist within each class of measures
(\simple, \shannon and \logical). Importantly, these measures are only effective
when they are correctly normalized---which is not always done in the
literature. Furthermore, the best-ranking measure for \simple (measure $g'_3$,
which we formally introduce in Section~\ref{sec:measures}) is sensitive to
\RHSskew and therefore performs worse than the best-ranking measures for
\shannon (measure $\RFInorm$) and \logical (a little-known measure called
$\muplus$), which have comparable performance as well as equal structural
sensitivity properties. However, $\RFInorm$ is slow to compute, whereas $\muplus$
is very efficient. Therefore, we recommend $\muplus$ for practical \afd
discovery. More in-depth discussion and recommendations are given in
Section~\ref{sec:conclusion}, where Table~\ref{table:comparison} summarizes our
comparison.

\smallerskip
\noindent
\textbf{In summary, our contributions are as follows.} (1) A survey of \afd measures, using a new and uniform presentation. 
(2) Sensitivity analysis of measure performance w.r.t.\ structural properties of the input relations. (3) Creation of a real-world benchmark for \afd discovery. (4) Analysis of measure ranking power on this benchmark. (5) Clear recommendations for measure adoption in \afd discovery.

The general organisation of the paper has already been outlined. In addition, we discuss related work and motivate the paper's scope in Section~\ref{sec:related}.
We introduce the necessary background in Section~\ref{sec:prelims}. 
\inConfVersion{An extended paper version, containing extra experimental material that is omitted here due to lack of space, is available online~\cite{artifacts}.}

%%% Local Variables:
%%% mode: latex
%%% TeX-master: "../main/main"
%%% End:

\section{Related Work and Paper Scope}
\label{sec:related}
\paragraph{Relaxing \fds.} % We note that i
In the literature there are two distinct
ways to relax the notion of an \fd~\cite{DBLP:journals/tkde/CaruccioDP16}: relax
the constraint that an \fd $\X \to \Y$ needs to be fully satisfied; or replace
the way in which tuples are compared on their $\X$-values by a similarity
function rather than strict equality. We focus on the former.%  in this
% paper.
% The latter setting leads matching
% dependencies~\cite{DBLP:journals/tods/SchirmerPKN20,DBLP:conf/pods/Fan08} and
% relaxed \fds~\cite{DBLP:journals/tkde/CaruccioDNP21}.

\paragraph{\afd measures versus discovery algorithms.} We focus on understanding \afd measures and how they allow us to rank the search space
of possible \fds for AFD discovery. A related but orthogonal issue, that we do not consider in this paper, is the study
of \afd \emph{discovery algorithms}. \afd discovery algorithms usually fix a way
to quantify the approximateness of an \fd---typically through the choice of an
\afd measure that fixes a ranking of the search space---but then combine a
multitude of techniques to do the actual discovery. This includes pruning the
ranked search space for efficiency reasons (e.g.,
\cite{DBLP:journals/pvldb/0001N18,DBLP:conf/kdd/MandrosBV17,DBLP:journals/kais/MandrosBV20,DBLP:conf/kdd/PennerathMV20})
or complementing the measures' ranking with heuristics for application-specific
purposes (e.g., \cite{DBLP:conf/sigmod/IlyasMHBA04,ZhangGuoEtAl2020}).
%
% In contrast, the central object of this study are the \afd measures themselves and how they rank the search space of possible \fds.
% In particular, we are interested in how well measures by themselves can  distinguish between \afds that are in the target set of design \fds of a given relation and those that are not.  
%
Of course, improved knowledge of \afd measures may aid in improving discovery
algorithms in the future. For this reason, when we introduce measures in
Section~\ref{sec:measures} we also indicate in which discovery algorithms they
are used.  Because discovery algorithms do not directly expose the way in which
they rank the search space of \fds, we do not experimentally compare against
discovery algorithms.

\paragraph{Correlation.} % A fundamental task in query optimization is the
% discovery of statistical correlations among attributes for the purpose
% of accurate selectivity estimations during query planning.
When an \fd holds in
a relation, there is clearly a statistical correlation among the \fd's
attributes. Conversely, correlated attributes may (but need not) indicate the
presence of an \fd. The techniques that are typically used to test statistical
correlation, such as the $\chi^2$ test or mean-square
contingency~\cite{DBLP:conf/sigmod/IlyasMHBA04}, however, only measure the
strength of correlation (e.g., $\X$ and $\Y$ are correlated) but do not indicate
the direction in which functional dependence ($\X \to \Y$ or $\Y \to \X$) is likely
to hold. As such, these techniques do not form appropriate \afd measures~\cite{piatetsky1993} and are not further considered here.

\paragraph{Exact FD discovery.} In the context of exact \fd discovery, some works
consider the problem of ranking exact FDs according to relevance, where the
challenge lies in quantifying relevance~\cite{DBLP:conf/icde/WeiL19}. We are not concerned with exact FD discovery, but with
measures for quantifying the extent to which \fds hold approximately.
Discovery of \afds should also not be confused with the approximate discovery of
exact \fds as e.g., done in \cite{DBLP:conf/cikm/BleifussBFRW0PN16}. There,
only a subset of all \fds that are satisfied are computed in return for
performance improvements.
% only a subset of all \fds that satisfy input $R$ are computed in return for
% performance improvements.

\paragraph{Existing comparisons of \afd measures.}
Giannella and Robertson~\cite{DBLP:journals/is/GiannellaR04} compare a limited
number of measures on theoretical examples as
well as on 4 real world datasets. In their experiments, they report on average
differences between pairs of measures and do not compare with a ground truth set
of \fds. They therefore do not empirically compare the effectiveness of the
measures as done here. In their survey concerning \fd relaxations, Caruccio
et al.~\cite{DBLP:journals/tkde/CaruccioDP16} also survey some of the \afd
measures that are considered here, but do not provide a qualitative comparison.

\paragraph{Discovery of Conditional FDs.} Conditional \fds (CFDs for short)
generalize FDs: they are FDs that only hold on a subset of the
data~\cite{DBLP:conf/icde/BohannonFGJK07}. CFDs are widely used in data
cleaning~\cite{DBLP:journals/pvldb/FanGJ08a,DBLP:journals/tods/FanGJK08}.
Discovery of (approximate) CFDs amounts to (i) selecting suitable subsets of the
data and (ii) discovering (A)FDs in these
subsets~\cite{DBLP:conf/pkdd/RammelaereG18}. In particular, Geerts and
Rammelaere~\cite{DBLP:conf/pkdd/RammelaereG18} propose a generic (A)CFD
discovery algorithm in which any (A)FD discovery algorithm can be plugged.
Insights into AFD measures, as is our focus here, can improve AFD discovery which in turn can be useful for the discovery of (A)CFDs.

\rev{\paragraph{Applications of \afds.} Data management tasks may use \afds to
  different extents. For example, only semantically meaningful \afds (i.e., FDs
  that a human database administrator would include in a database schema design)
  are useful for data cleaning tasks~\cite{DBLP:journals/pvldb/0001N18} while
  for query optimisation one could try to exploit all AFDs present in a
  relation, even if they are not semantically meaningful. Our real-world
  benchmark focuses on the former kind of \afds.
  Thus, we present a study of
  \afd measures applicable to discover semantically meaningful \afds.  }

%%% Local Variables:
%%% mode: latex
%%% TeX-master: "../main/main"
%%% End:

\section{Preliminaries}
\label{sec:prelims}

\begin{table}[t]
	\caption{\label{table:notation}\rev{Glossary of notation used in this paper.}}
	\centering
	\begin{tabular}[t]{cl}
		\toprule
		$X, Y$                & an attribute                                                              \\
		$\dom(X)$             & the domain of an attribute                                                \\
		$\X, \Y$              & a finite set of attributes                                                \\
		$\x$                  & a tuple  over set $\X$ of attributes                                      \\
		$\x\colon \X$         & $\x$ is a tuple over $\X$                                                 \\
		$\restr{\x}{\Y}$      & restriction of $\x$ to $\Y$ with $\Y\subseteq\X$                          \\
		$\X\Y$                & union of two sets of attributes, i.e. $\X\cup\Y$                          \\
		$\x\y$                & idem for tuples, i.e. $\restr{\x\y}{\X} = \x$ and $\restr{\x\y}{\Y} = \y$ \\
		$R$                   & a relation                                                                \\
		$R(\X)$               & $R$ is relation over $\X$                                                 \\
		$R(\x) \in \nats$     & frequency of $\x$ in $R$                                                  \\
		$\x \in R$            & $\x$ is a tuple in $R$ with $R(\x) > 0$                                   \\
		$\dom_R(\Y)$          & $\{ \restr{\x}{Y} \mid  \x \in R\}$                                       \\
		$\size{R}$            & total number of tuples in $R$, i.e. $\sum_{\x\colon\X}{R(\x)}$            \\
		$\pi_{\X}(R)$         & bag-based relational projection of $R$ onto $\X$                          \\
		$\sigma_{\X = \x}(R)$ & bag-based relational selection of $R$ on $\X = \x$                        \\
		$\varphi := \X\to\Y$  & a functional dependency                                                   \\
		$R\models\varphi$     & $R$ satisfies a functional dependency $\varphi$                           \\
		$R\not\models\varphi$ & $R$ violates a functional dependency $\varphi$                            \\
		$\schema(R)$          & the (fixed) schema of $R$                                                 \\
		\bottomrule
	\end{tabular}
\end{table}

%%% Local Variables:
%%% mode: latex
%%% TeX-master: "afd_comparison_table"
%%% End:

\rev{We summarize the notation used in this paper in Table~\ref{table:notation}.}
We assume given a fixed set of attributes, where each attribute $X$ has a domain
$\dom(X)$ of possible data values.
We use uppercase letters $X,Y,Z$ to denote attributes and
boldface type like $\seq{X}, \seq{Y}, \seq{Z}$ to denote sets of attributes.
Lowercase $\seq{x}$, $\seq{y}$, $\seq{z}$ denote tuples over these sets.
Formally, as usual, a tuple over $\seq{X}$ is a mapping $\seq{x}$ that assigns
each attribute $X \in \seq{X}$ to a value $\seq{x}(X) \in \dom(X)$.
We write $\seq{x}\colon \seq{X}$ to indicate that $\seq{x}$ is a tuple over $\seq{X}$, and  $\restr{\seq{x}}{\seq Y}$ for the restriction of $\seq x$ to $\seq{Y} \subseteq \seq X$. We use juxtaposition like $\X\Y$ to denote the union $\X \cup \Y$ of two sets of attributes, and also apply this notation to tuples: if $\x \colon \X$ and $\y\colon \Y$ with $\X$ and $\Y$ disjoint, then $\x\y$ is the tuple that equals $\x$ on all attributes in $\X$ and $\y$ on all attributes in $\Y$, i.e. $\restr{\x\y}{\X} = \x$ and $\restr{\x\y}{\Y} = \y$.

We will work with bag-based relations. Formally, a relation over $\seq{X}$ (also
called $\seq{X}$-relation) is a mapping $R$ that assigns a natural number
$R(\seq{x}) \in \nats$ to each tuple $\seq{x}\colon \seq{X}$.
%We also call $R(\seq{x})$ the frequency of $\seq{x}$ in $R$.
We require relations to be finite in the sense that $R(\seq{x})$ can be non-zero
for at most a finite number of $\seq{x}$.  We write $\seq x \in R$ to denote
that $R(\seq x) > 0$ and stress that $R$ is an $\seq X$-relation by means of the
notation $R(\seq X)$.  $\size{R}$ denotes the total number of tuples in $R$,
$\pi_{\Y}(R)$ denotes bag-based projection on $\Y$, and $\sigma_{\X = \x}(R)$
denotes bag-based selection. If $\Y \subseteq \X$ then we denote by $\dom_R(\Y)$
the set $\{ \restr{\x}{Y} \mid  \x \in R\}$.
%We denote by $\size{R}$ the total
                                % number of tuples in $R$, i.e. $\size{R} = \sum_{\seq{x}\colon \seq{X}} R(\seq{x})$.
% We denote bag-based relational projection and selection as usual by $\pi_{\Y}(R)$ and
% $\sigma_{\X = \x}(R)$, respectively.

\paragraph{Functional Dependencies.} A \emph{functional dependency} (\fd for short) is an expression $\varphi$ of the form $\seq{X} \to \seq{Y}$. An \fd is \emph{linear} if $|\X| = |\Y| = 1$ and \emph{non-linear} otherwise.
A relation $R(\W)$ with $\seq{X},\seq{Y} \subseteq \W$ \emph{satisfies} $\varphi$  if for all tuples $\w, \seq{w'} \in R$ we have that $\restr{\w}{\Y} = \restr{\seq{w'}}{\Y}$ whenever $\restr{\seq{w}}{\X} = \restr{\seq{w'}}{\X}$.
% We write $R \models \varphi$ to indicate that $R$ satisfies $\varphi$, and $R \not \models \varphi$ to indicate that it violates $\varphi$. 
In what follows, we always implicitly assume that $\X$ and $\Y$ are disjoint when considering \fds.

\paragraph{Dependency Discovery.} A \emph{schema} is a finite set of
\fds. In the \emph{exact \fd discovery problem} we are given a relation $R$ that
satisfies all \fds in some fixed design schema $\schema(R)$, but have no knowledge
of $\schema(R)$ itself. We are then asked to recover $\schema(R)$ by  deriving the largest set $\Lambda \supseteq \schema(R)$ of
\fds that are satisfied by $R$. In the \emph{approximate \fd discovery problem}, we are given a relation $R$ that does not satisfy $\schema(R)$ and again we are asked to recover $\schema(R)$. Here, we assume that $R$ is obtained by means of a noisy channel process as follows. From a clean relation $R'$ that satisfies $\schema(R)$, $R$ is obtained by modifying certain values in tuples in $R'$. We consider an \emph{error} each cell for which $R$ differs from the clean version $R'$. Note that by running exact \fd discovery algorithms on $R$, we will still be able to recover satisfied \fds in $\schema(R)$. Our interest in this paper is in \emph{approximate \fd discovery}, i.e., deriving the \fds in $\schema(R)$ that, because of errors introduced, are violated in $R$ and therefore cannot be discovered by exact \fd discovery.

In Section~\ref{sec:measures} we survey various measures that have been
proposed to quantify the level to which an \fd holds approximately. As we will
see, many of these measures are based on exploiting notions of
\emph{Shannon} or \emph{logical} entropy.
%However, some employ \emph{Shannon} entropy while others employ
%\emph{logical} entropy. 
We introduce these notions next.

\paragraph{Probabilities.}
Both notions of entropy are defined w.r.t.\ a given joint probability
distribution. In our setting, this probability distribution is defined by the
relation under consideration. Let $R(\W)$ be a relation. The
joint probability distribution $\prob_R(\seq{W})$ over $\W$ induced by $R$ is
defined by $\prob_R(\W = \w) = \frac{ R(\!\w) }{|R|}$. As such, $\prob_R(\W = \w)$ is the probability of observing $\w$ when
randomly drawing a tuple from $R$. We note that this probability distribution is
only well-defined when $R$ is non-empty. Because the empty relation vacuously satisfies all \fds,  we will implicitly assume without loss of generality
in the rest of this paper that relations are non-empty.

To simplify notation in what follows, we simply write $\prob_R(\w)$,
$\prob_R(\y)$ and $\prob_R(\y \mid \x)$ instead of $\prob_R(\W = \w)$,
$\prob_R(\Y = \y)$ and $\prob_R(\Y = \y \mid \X = \x)$, respectively.
The notions of marginal and conditional distributions derived from $\prob_R(\w)$
are defined as follows:
% repetition, we defined that X and Y are disjoint in "III.  Preliminaries / Functional Dependencies"
% For the remainder of the section, let $\X,\Y \subseteq \W$ be disjoint subsets of $\W$.
%Then,
$\prob_R(\y)$ denotes the marginal probability
distribution on $\seq Y$-tuples in $R$, while
$\prob_R(\y \mid \x)$ is the conditional distribution on $\seq{Y}$ given $\seq{X}=\x$. Thus,
\begin{align*}
  \prob_R(\seq y) = \sum_{\w: \W \text{ s.t. }
    \restr{\seq w}{\seq Y} = \seq y} \prob_R(\seq w), \;\;\;\;\;\;
  \prob_R(\seq y \mid \x) = \frac{\prob_R(\seq x\seq y)}{\prob_R(\seq x)}.
\end{align*}
It is readily verified that $\prob_R(\Y)$ equals the distribution induced by $\pi_Y(R)$, while $\prob_R(\Y \mid \X = \x)$ equals the distribution induced by $\pi_Y\sigma_{\X = \x}(R)$.

\paragraph{Shannon Entropy.}
We write $\entrop_R(\seq{X})$ for the \emph{Shannon entropy} of $\X$ in $R$,
defined as usual~\cite{DBLP:journals/ipm/Edwards08} by\footnote{Here and in the
  sequel we use the common convention that $0 \log 0 = 0$ and
  $\frac{0}{0} = 0$.}
\[ \entrop_R(\seq X) = - \sum_{\seq x\colon \seq X} \prob_R(\seq x) \log \prob_R(\seq x). \]
$\entrop_R(\seq X)$ reflects the average level of uncertainty inherent in the possible tuples over $\seq X$ in $\pi_X(R)$.
The \emph{conditional entropy} $\entrop_R(\seq Y \mid \seq X)$ is the uncertainty in $\seq Y$ given $\seq X$, defined as
\begin{align*}
  \entrop_R(\seq Y \mid \seq X) & = - \sum_{\seq x\colon \seq X, \seq y\colon \seq Y} \prob_R(\seq x \seq y) \log \frac{\prob_R(\seq x \seq y)}{\prob_R(\seq x)}.
\end{align*}
Equivalently, denoting by $\entrop_R(\Y \mid \x)$ the Shannon entropy of $\Y$ in the conditional distribution $\prob_R(\Y \mid \X = \x)$,  we see that $\entrop_R(\Y \mid \X)$ is the expected value of $\entrop_R(\Y \mid \x)$, taken over all $\x$, i.e., $\entrop_R(\Y \mid \X) = \expect_{\x}\left[\entrop_R(\Y \mid \x)\right]$.

\paragraph{Logical Entropy.}
The \emph{logical entropy} of $\X$ in $R$ is the probability that two tuples
$\seq{w}$ and $\seq{w'}$, drawn randomly with replacement from $R$ according to $\prob_R$, differ in some attribute in $\seq X$~\cite{logical-entropy}. That is,
\[ \logentrop_R(\seq{X}) := 1 - \sum_{\seq{x}\colon \seq X} p_R(\seq x)^2. \]
Here, $p_R(\seq x)^2$ is the probability that two random tuples are exactly equal to $\seq x$ on $X$.

We denote by $\logentrop_R(\Y \mid \x)$ the logical entropy of $\Y$ in the conditional distribution $\prob_R(\Y \mid \X = \x)$, i.e.,
\[ \logentrop_R(\Y \mid \x) = 1 - \sum_{\y \colon \Y} \prob_R(\y \mid \x)^2.\]

The \emph{logical conditional entropy} of $\seq Y$ given $\seq X$ in $R$, denoted $h_R(\seq Y \mid \seq X)$, is the probability that two tuples, drawn at random with replacement from $R$ according to $p_R$, are equal in all attributes of $\seq X$, but differ in some attribute of $\seq Y$,

\begin{equation*}
  \logentrop_R(\seq{Y} \mid \seq{X}) := \sum_{\x, \y} \prob_R(\x\y)[\prob_R(\x)-\prob_R(\x\y)].
\end{equation*}
Here, the factor $\prob_R(\x\y)$ expresses the probability of observing $\x\y$
in the first tuple and the factor $\prob_R(\x)-\prob_R(\x,\y)$ is the
probability that the second tuple has the same value for $\x$ but differs in
$\y$.

Note that, in contrast to the case of Shannon entropy where $\entrop_R(\Y \mid \X) = \expect_{\x}[ \entrop_R(\Y \mid \x)]$, in logical entropy $\logentrop_R(\Y \mid \X) \not = \expect_{\x}[ \logentrop_R(\Y \mid \x)]$.
\inFullVersion{We discuss the relationship between logical and Shannon entropy further in Appendix~\ref{sec:logical:vs:Shannon_entropy}}
\inConfVersion{We discuss the relationship between logical and Shannon entropy further in \cite{artifacts}.}

%%% Local Variables:
%%% mode: latex
%%% TeX-master: "../main/main"
%%% End:

\section{AFD Measures}
\label{sec:measures}
In this section, we survey the literature on \afd measures. % While our focus in this paper is on measures and how they compare, and not on \afd discovery algorithms, for purpose of illustration we also mention for each measure the \afd discovery algorithms in which this measures is used.

\paragraph{\afd measures.} Formally, an \emph{\afd measure}, short for
\emph{approximate} \fd measure, is a function that maps pairs $(\varphi, R)$,
with $\varphi$ an \fd and $R$ a relation, to a number in the interval $[0,1]$
that indicates the level to which $\varphi$ holds in $R$. Higher values are
intended to indicate that $R$ makes fewer violations to $\varphi$, and we
require that $f(\varphi, R) = 1$ if $R$ perfectly satisfies $\varphi$.
It is important to note that instead of defining \afd measures, some papers in
the literature define \emph{error measures} where a high value indicates
a high number of errors against the \fd. In what follows, we routinely re-define
such error measures $e$ into an \afd measure $f_e$ by setting
$f_e(\varphi, R) := 1 - e(\varphi, R)$.

Every \afd measure $f$ naturally gives rise to an associated \afd discovery
algorithm as follows. From an abstract viewpoint, an \afd discovery algorithm
simply consists of a fixed \afd measure $f$ and a threshold
$\epsilon \in [0,1]$. Given a relation $R(\W)$ the algorithm returns all \fds
over $\W$ whose $f$-value lies in the range $[\epsilon,1[$. (In particular, this excludes the \fds satisfied by $R$.)

\inFullVersion{
  \paragraph{Interpretation and baselines.}
  As with all threshold-based algorithms, a key difficulty for \afd discovery
  algorithms lies in determining the correct threshold $\epsilon$ to use. At its
  core, this question boils down to how we should interpret the significance of
  the values returned by $f$. It is tempting to see the values of $f$ as a
  percentage with $f(\varphi, R) = 1$ indicating that $R$ perfectly satisfies
  $\varphi$ and $f(\varphi, R) = 0$ indicating that $R$ completely fails to
  satisfy $\varphi$.  This interpretation, however, is only valid if the measure
  has a notion of $R$ ``completely failing to satisfy'' $\varphi$. In particular,
  this is only possible when there are relations for which $f(\varphi,R) = 0$.
  In what follows, we call a
  relation $R$ with $f(\varphi, R) = 0$ a \emph{baseline of $f$ for $\varphi$}. If
  $f$ has a baseline for every \fd $\varphi$ then we say that $f$ \emph{has
    baselines}, otherwise we call $f$ \emph{without baselines}. Having baselines
  is a necessary condition for interpreting measure scores as
  percentages.
}

\paragraph{Conventions.} Throughout this section, assume that
$R$ is a $\W$-relation, let $\X,\Y$ be disjoint subsets of $\W$ and let
$\varphi = \X \to \Y$. We convene that for all measures $f$ that we describe, we
trivially set $f(\varphi, R) := 1$ if $R \models \varphi$. So, the definitions
that follow only apply when $R \not \models \varphi$. In that case, observe that
$R$ must be non-empty, that $\size{\!\dom_R(\X)} \not = \size{R}$ and that
$\size{\!\dom_R(\Y)} > 1$ since otherwise $R$ trivially satisfies $\X \to
  \Y$. As a consequence, $\entrop_R(\Y) > 0$ and $\logentrop_R(\Y) > 0$. This ensures that the denominator of fractions in the formulas that follow are never zero.

\subsection{Co-occurrence ratio}
\label{sec:cords-measure}

Ilyas et al.~\cite{DBLP:conf/sigmod/IlyasMHBA04} consider the derivation of
\afds (called \emph{soft} \fds in their paper) as well as general correlations
between attributes. To derive \afds, they consider the
ratio between the number of distinct $\X$-tuples and the number of distinct
$\X\Y$-tuples occurring in $R$. We denote this measure by $\rho$, formally
defined as:
\[ \rho(\X \to \Y, R)  := \frac{|\dom_{\X}(R)|}{|\dom_{\X\Y}(R)|}. \]
This is $1$ if $R$ satisfies $\X \to \Y$ and decreases when more $\y$-tuples
occur with the same $\x$-tuple. Note that $\rho$ is a set-based measure, as it
ignores the multiplicities of the tuples in $R$. \inFullVersion{It is also
without baselines, as $\size{\!\dom_{\X}(R)}>0$ for any non-empty relation $R$
and as, by convention, $\rho(\varphi,R) = 1$ when $R$ is empty.}

\subsection{g-measures}
\label{sec:g-measures}

Kivinen and Mannila~\cite{DBLP:journals/tcs/KivinenM95} introduced three  error measures on set-based relations. Generalized to bag-based relations, and converted to \afd measures, these are the following.

\paragraph{The measure $g_1$.} The measure $g_1$ is based on logical
entropy. Specifically, Kivinen and Manila defined $g_1$ to reflect the
(normalized) number of violating pairs in $R$.  Here, a pair $(\w,\w')$ of
$R$-tuples is a \emph{violating pair} if they are equal on $\X$ but differ on
$\Y$. Formally, if we denote the bag of violating pairs in $R \times R$ by
$G_1(\X \to \Y, R)$ then, converted to an \afd measure instead of an error measure% , $g_1$ is defined as
\begin{align*}
  g_1(\X \to \Y,R) & := \frac{|R|^2  - |G_1(\X \to \Y,R)|}{|R|^2} \\
                   & = 1 - \frac{|G_1(\X \to \Y,R)|}{|R|^2} =  1 - h_R(\Y \mid \X).
\end{align*}
In other words, $g_1$ is maximized when the logical conditional entropy is minimized.

\inFullVersion{The measure $g_1$ is without baselines.}
Because pairs of the form $(\w,\w)$ are never violating, it is straightforward to see that the total number of violating pairs is bounded from above by $\size{R}^2 - \sum_{\w} R(\w)^2$. We denote by $g'_1$ the normalized version of $g_1$,
\begin{align*}
  g'_1(\X \to \Y,R) & := 1 -  \frac{|G_1(\X \to \Y,R)|}{\size{R}^2-\sum_{\w} R(\w)^2}.
\end{align*}
\inFullVersion{The baselines of $g'_1$ are hence those relations for which the set $G_1(\X \to \Y, R)$ consists of all possible violating pairs.}

Both $g_1$ and $g'_1$ have been used as the basis of \afd discovery
algorithms. In particular, $g_1$ is the basis of
\textsc{Fdx}~\cite{ZhangGuoEtAl2020} while $g'_1$ is the basis of
\textsc{Pyro}~\cite{DBLP:journals/pvldb/0001N18}.  Adaptations of $g'_1$ are also used in the context of denial constraints ~\cite{PenaAlmeidaEtAl2019}
and roll-up dependencies~\cite{DBLP:journals/tods/CaldersNW02}.

\paragraph{The measure $g_2$.} Kivinen and Manila defined $g_2$ to reflect the probability that a random tuple participates in a violating pair. We define $G_{2}(\X \to \Y, R)$ to be the set of all tuples in $R$ that participate in a violating pair, \\
\scalebox{.9}{\parbox{1.1\linewidth}{%
	\begin{align*}
		G_{2}(\X {\tiny\to} \Y, R):=\left\{ \w\in R\mid \exists \w' \in R, (\w,\w') \in G_1(\X {\tiny\to} \Y, R) \right\}.
	\end{align*}
}}
Then, $g_2$, converted to an \afd measure instead of
an error measure as originally proposed,
computes the probability that a tuple, drawn randomly from $R$ according to $\prob_R$, is not part of a violating pair,
\[ g_2(\X \to \Y, R) := 1 - \sum_{\w \in G_2(\X \to \Y, R)} \prob_R(\w) . \]
The FD-compliance-ratio, that is used as one of the building blocks in UNI-DETECT  \cite{WangHe2019}, is based on $g_2$.

\paragraph{The measure $g_3$.} The measure $g_3$ computes the relative size of a maximal subrelation of $R$ for which $\X \to \Y$ holds. Specifically, define $R'(\W)$ to be a subrelation of $R(\W)$, denoted $R'\subseteq R$, if $R'(\w) \leq R(\w)$
for all $\w\colon \W$. Let $G_3(\X\to \Y, R)$ denote the set of all subrelations of $R$ that satisfy $\X \to \Y$,
\[ G_3(\X\to\Y, R) := \left\{R' \mid R'\subseteq R, R'\models \X\to \Y \right\}.\] Then
$g_3$ is defined as the maximum relative size of a subrelation satisfying $\X \to \Y$:
\begin{align*}
  g_3(\X \to \Y, R) & := \max_{R' \in G_3(\X \to \Y, R)} \frac{|R'|}{|R|}.
\end{align*}
Note that $1-g_3(\X \to \Y, R)$ can naturally be interpreted as the minimum fraction of tuples that need to be removed for $\X \to \Y$ to hold in $R$.

\altVersion{The measure $g_3$ is without baselines.} For any non-empty $R$ we can
always obtain a subrelation $R' \in G_3(\varphi, R)$ of size
$\size{\!\dom_{\X}(R)}$ by arbitrarily fixing one $\y$-value for each
$\x$-value. As such, $g_3$ is bounded from below by
$\frac{|\dom_{\X}(R)|}{\size{R}} > 0$. Gianella and
Robertson~\cite{DBLP:journals/is/GiannellaR04} proposed a normalized variant
$g'_3$ of $g_3$, defined as follows:
\begin{align*}
  g'_3(\X \to \Y, R) & := \max_{R' \in G_3(\X \to \Y, R)} \frac{\size{R'} - \size{\dom_{R}(\X)}}{\size{R} - \size{\dom_{R}(\X)}}.
\end{align*}
\inFullVersion{This variant has baselines, namely all relations $R$ for which no 
subrelation $R' \in G_3(\varphi,R)$ is larger than $\size{\dom_R(\X)}$.}

The unnormalized measure $g_3$ is used in multiple \afd discovery algorithms~\cite{DBLP:journals/tcs/KivinenM95,DBLP:journals/pvldb/Berti-EquilleHN18, DBLP:journals/cj/HuhtalaKPT99,DBLP:journals/jamds/KingL03}. {Furthermore, the `per-tuple' probability of an \fd, as defined in \cite{DBLP:conf/webdb/WangDSFH09}, is precisely  $g_3$.}
%In addition,
Berzal et al.~\cite{BERZAL2002207} use it as the basis for relational decomposition based on \afds instead of \fds.
{Exact and approximate solutions for the computation of $g_3$ in the context of non-crisp \fds are proposed in~\cite{DBLP:conf/icde/Faure-Giovagnoli22}.}
We note that $g_3$ has been generalized to other dependencies as well: e.g.,
conditional \fds~\cite{DBLP:conf/sigmod/CormodeGKMSZ09,DBLP:journals/pvldb/RammelaereG18},  inclusion dependencies~\cite{MarchiLopesEtAl2009}, and conditional matching dependencies~\cite{DBLP:journals/tkdd/WangSCYC17}.
By contrast, the normalized version $g'_3$ only appears in  \cite{DBLP:journals/is/GiannellaR04}.

\subsection{Fraction of information}
\label{sec:fraction-information-measures}

\paragraph{Fraction of Information.} Cavallo and Pittarelli~\cite{DBLP:conf/vldb/CavalloP87} introduced
\emph{fraction of information} (\FI)  as a way to generalize \fds from
deterministic to probabilistic databases.  Usage of \FI as an \afd measure was later studied by Giannelli and Robertson~\cite{DBLP:journals/is/GiannellaR04}% , and who also compare \FI to $g'_3$ and $\tau$ (see section~\ref{sec:prob-dependency})
.
\FI is based on Shannon entropy and is formally defined as
\[ \FI(\X \to \Y, R) :=
  % \begin{cases}
  %   1                                                           & \text{ if } |\dom_R(\Y)| = 1, \\
  \frac{\entrop_R(\Y) - \entrop_R(\Y \mid \X)}{\entrop_R(\Y)}. % & \text{otherwise.}
  % \end{cases}
\]
% The case analysis is to avoid division by $0$, since $\entrop_R(\Y) = 0$ if, and only if, $\size{\!\dom_R(\Y)} = 1$. In the second case, 
The numerator $\entrop_R(\Y) - \entrop_R(\Y \mid \X)$ is known as \emph{mutual information}~\cite{DBLP:journals/ipm/Edwards08}, which we further denote by $\mut_R(\X; \Y)$.

We can understand \FI as follows.  $\entrop_R(\Y)$ measures the uncertainty of
observing $\Y$, while $\entrop_R(\Y \mid \X)$ measures the uncertainty of
observing $\Y$ after observing $\X$. \FI hence represents the proportional
reduction of uncertainty about $\Y$ that is achieved by knowing $\X$. When $R$
satisfies $\X \to \Y$, there is no uncertainty about $\Y$ after observing $\X$
and hence $\entrop_R(\Y \mid \X) = 0$ and so \FI is $1$. Conversely, when $\X$
and $\Y$ are independent random variables in $\prob_R$, there is no
reduction in uncertainty, and hence $\entrop_R(\Y\mid \X)=\entrop_R(\Y)$ and so
$\FI$ is $0$. \inFullVersion{ Thus, the baselines of $\FI$ for $\X \to \Y$ are
those relations $R$ for which $\X$ and $\Y$ are independent in $\prob_R$.}

\paragraph*{Bias.}
Mandros et al. \cite{DBLP:conf/kdd/MandrosBV17,DBLP:journals/kais/MandrosBV20}
and Pennerath et al. \cite{DBLP:conf/kdd/PennerathMV20} proposed two refinements
to \FI specifically for \afd discovery, called \emph{reliable \FI} (\RFI) and
\emph{smoothed \FI} (\SFI), respectively. They are motivated in proposing these
refinements by the following observation. Consider a relation $S(\W)$ and assume
that we are given relation $R(\W)$ of size $n$ that is obtained by sampling $n$
tuples from $S$ according to distribution $\prob_S$. Further assume that we do
not have access to $S$ and wish to determine $\FI(\X \to \Y, S)$ based on
$R$. Then a result by Roulston~\cite{ROULSTON1999285} states that the expected
value of $\mut_R(\X; \Y)$, taken over all $R$ obtained in this manner, equals
\begin{equation*}
  \mut_S(\X; \Y) + \frac{\size{\dom_S(\X\Y)} - \size{\dom_S(\X)} - \size{\dom_S(\Y)} + 1}{2n}.
\end{equation*}
In other words, we may expect  $\mut_R(\X; \Y)$ to overestimate
$\mut_S(\X; \Y)$ and the magnitude of overestimation depends on the size of the
active domains of $\X\Y$, $\X$, and $\Y$ in $S$, as well as on $n$.  Additionally, because $\entrop_R(\Y)$ underestimates $\entrop_S(\Y)$~\cite{ROULSTON1999285},
we may conclude that  $\FI(\X \to \Y, R)$ is expected to overestimate
$\FI(\X \to \Y, S)$ and the magnitude of overestimation depends on the active
domain sizes and the size of $S$. This overestimation is problematic since $\FI(\X \to \Y, R)$ will be quite large, even if $\X$ and $\Y$ are independent in $\prob_S$, resulting in $\FI(\X \to \Y, S)$ being $0$.

\paragraph*{Reliable \FI.}
Reliable \FI corrects for this bias by subtracting the mutual information value
that is expected under random $(\X;\Y)$-permutations.
\begin{definition}
  \label{def:permuation}
  Relation $R'$ is an \emph{$(\X;\Y)$-permutation} of $R$, denoted
  $R \sim_{\X;\Y} R'$
  if (i) $\size{R} = \size{R'}$; (ii)
  $\pi_{\X}(R) = \pi_{\X}(R')$; (iii) $\pi_{\Y}(R) = \pi_{\Y}(R')$; and (iv) $\pi_{\Z}(R) = \pi_{\Z}(R')$ where $\Z = \W \setminus \X\Y$.
\end{definition}
In particular, $R'$ and $R$ have the same marginal distributions on $\X$ and on $\Y$, $\prob_{R'}(\X) = \prob_R(\X)$ and $\prob_{R'}(\Y) = \prob_R(\Y)$. In what follows, for a measure $f$, we denote by $\expect_{R}[f(\X \to \Y, R))]$ the expected value of $f(\X \to \Y,R)$ where the expectation is taken over all  $(\X;\Y)$-permutations of $R$.

Reliable fraction of information is then defined as
\begin{align*}
  \RFI(\X \to \Y, R) & := \FI(\X \to \Y, R) - \expect_{R}[ \FI(\X \to \Y, R)].
\end{align*}
Because the number of permutations of $R$ is finite, we may compute $\expect_R[\FI(\X \to \Y, R)]$, and hence also $\RFI(\X \to \Y, R)$, by simply computing $\FI(\X \to \Y, R')$ for every permutation $R'$ of $R$ and taking the average. More efficient algorithms are proposed in \cite{DBLP:journals/kais/MandrosBV20,DBLP:conf/kdd/MandrosBV17}. Even with these improved algorithms, computing $\RFI$ remains inefficient, as we show in Section~\ref{sec:eval-rwd}.

% . Then,
% NOTE TO SELF the equivalence holds by linearity of expectation and the fact that $\expect(\entrop_{R'}(\Y)) = \entrop_R(\Y) since the distribution of $\Y$ is fixed.

% \frank{Is RFI be 0 when $\entrop_R(\Y) = 0$? because in that case the expected value is 1 as well? Or is it a typo? } \stijn{Not a typo: the expected value is 1 because all permutations have the same distribution on $\Y$. So, $\entrop_R(\Y) = 0$ iff $\dom_R(\Y)$ is a singleton, which remains the case for all permutations.}
% For this first case, when $\size{\dom_R(\Y)}= 1$, then $\FI(\X \to \Y, R)=\expect_{R}[ \FI(\X \to \Y, R)]=1$, and, hence, \RFI$(\varphi,R)=0$.
% For the second case, the equality holds by linearity of expectation and the observation that, because all permutations of $R$ have the same marginal $\Y$ distribution, we have  $\expect_{R}[\entrop_{R}(\Y)] = \entrop_R(\Y)$.

Strictly speaking, $\RFI$ is not an \afd measure since it can become negative when $\FI(\varphi, R) < \expect_{R}[\FI(\varphi, R)]$. Because such negative $\RFI$ values indicate that there is weak evidence to conclude that $\varphi$ is an \afd, we turn $\RFI$ into an actual AFD measure $\RFIplus$ by setting
\[
  \RFIplus(\X \to \Y, R) := \max(\RFI(\X \to \Y, R), 0).
\]
\inFullVersion{The baselines of $\RFIplus$ for $\X \to \Y$ are hence all relations whose \FI value is smaller or equal than the expected value under random permutations.}

\paragraph*{Smoothed \FI.}
Smoothed \FI uses \emph{laplace smoothing} to reduce bias. Laplace smoothing is
a well-known statistical technique to reduce estimator variance. It is
parameterized by a value $\alpha > 0$. Specifically, for a relation $S(\X\Y)$, let $S^{(\alpha)}$
denote the $\alpha$-smoothed version of $S$, defined by
$S^{(\alpha)}(\x\y) := S(\x\y) + \alpha$ for every $\x \in \dom_{S}(\X)$ and $\y \in \dom_S(\Y)$. Note in particular, that it is possible that  $S(\x\y) = 0$, in which case $S^{(\alpha)}(\x\y) = \alpha$. Then the smoothed \FI of $R$ is simply the normal \FI  of the $\alpha$-smoothed version of $\pi_{\X\Y}(R)$:
\begin{align*}
  \SFI_{\alpha}(\X \to \Y, R) & := \FI(\X \to \Y, \pi_{\X\Y}^{(\alpha)}(R)).
\end{align*}
We note that, because $\pi_{\X\Y}^{(\alpha)}(R)$ contains a tuple $\x\y$ for every possible combination of $\x \in \dom_{\X}(R)$ and $\y \in \dom_{\Y}(R)$, it can be many times larger than $R$. \SFI is therefore also relatively inefficient to compute, as we show in Section~\ref{sec:eval-rwd}.

\afd discovery algorithms based on \RFI and \SFI are presented in \cite{DBLP:conf/kdd/MandrosBV17,DBLP:journals/kais/MandrosBV20}
and \cite{DBLP:conf/kdd/PennerathMV20}, respectively.

\subsection{Probabilistic dependency, $\tau$ and $\mu$}
\label{sec:prob-dependency}

Piatetsky-Shapiro and Matheus~\cite{piatetsky1993} proposed \emph{probabilistic
  dependency} as another probabilistic generalization of a functional
dependency. They also introduced a normalized version of probabilistic
dependency, which is equivalent to the Goodman and Kruskal $\tau$ measure of
association~\cite{10.2307/2281536}. Finally, they also propose a rescaled
version of $\tau$. All three notions are defined as follows. It is worth noting,
that, apart from~\cite{piatetsky1993}, we are not aware of any work that
considers these measures for \afd discovery in the database context, let alone
designs \afd discovery algorithms for them.

\paragraph*{Probabilistic dependency.}
The \emph{probabilistic dependency of\/ $\Y$ on $\X$ in $R$}, denoted by
$\pdep(\X \to \Y, R)$, represents the conditional probability that two tuples
drawn randomly with replacement from $R$ are equal on $\Y$, given that they are
equal on $\X$. Formally, 
\begin{align*}
  \pdep(\X \to \Y, R) & := \sum_{\x} \prob_R(\x) \pdep(\Y \mid \x, R),
\end{align*}
where $\pdep(\Y \mid \x, R)$ is the probability that two random $\Y$-tuples
drawn with replacement from the conditional distribution $\prob_R(\Y \mid \x)$
are equal:
\[ \pdep(\Y \mid \x, R) := \sum_{\y} \prob_R(\y \mid \x)^2 = 1 - \logentrop_R(\Y \mid \x). \]
Probabilistic dependency is hence a measure based on logical entropy.
It can be understood as follows. Suppose that we are given two tuples that equal
$\x$ on $\X$. Then $\pdep(\Y\mid \x, R)$ is the probability that these tuples are
also equal on $\Y$, and $\pdep(\X \to \Y, R)$ is the expected value of
$\pdep(\Y\mid \x, R)$ over all $\x$.

We note that probabilistic dependency can also be seen as a generalization of the measure $g_2$. Whereas $g_2$ computes the probability that a random tuple cannot be extended to a violating pair, probabilistic dependency computes the average conditional probability that a given $\X$-tuple $\x$ cannot be extended to a violating pair, where the average is taken over all values of $\X$.

\paragraph*{The measure $\tau$.} \inFullVersion{It is straightforward to see that $\pdep(\X \to \Y, R) > 0$, always. As such, $\pdep$ is a measure without baselines. In fact, }Piatetsky-Shapiro and Matheus~\cite{piatetsky1993} show that we always have
\[ \pdep(\X \to \Y, R) \geq \pdep(\Y, R) \]
where $\pdep(\Y, R)$, called \emph{probabilistic self-dependency},  is defined
as the probability that two random tuples in $R$ have equal $\Y$ attributes,
\[ \pdep(\Y, R) := \sum_{\y} \prob_R(\y)^2 = 1 - \logentrop_R(\Y). \]
To account for the relationship between $\pdep(\X \to \Y, R)$ and $\pdep(\Y,R)$,
Piatetsky-Shapiro and Matheus propose to normalize $\pdep(\X\to\Y,R)$ w.r.t. $\pdep(\Y,R)$. The
resulting measure is equivalent to the Goodman and Kruskal $\tau$ (tau) measure
of association~\cite{10.2307/2281536}, which is defined as
\begin{align*}
  \tau(\X\to \Y, R) :=
  \frac{\pdep(\X \to \Y, R)-\pdep(\Y, R)}{1-\pdep(\Y, R)}.
\end{align*}
Piatetsky-Shapiro and Matheus
explain $\tau$  in the following way~\cite{piatetsky1993}. Suppose we
are given a tuple drawn randomly from $R$ according to $\prob_R$, and we need to
guess its $\Y$ value. One strategy is to make guesses randomly according to the
marginal distribution of $\Y$, i.e. guess value $\Y=\y$ with probability
$\prob_R(\y)$. Then the probability for a correct guess is $ \pdep(\Y,R)$. If we
also know that item has $\X=\x$, we can improve our guess using conditional
probabilities of $\Y$, given that $\X=\x$. Then our probability for success,
averaged over all values of $\X$, is $\pdep(\X \to \Y,R)$, and
$\tau(\X \to \Y, R)$ is the relative increase in our probability of successfully
guessing $\Y$, given $\X$. \inFullVersion{The baselines of $\tau$ for $\X \to \Y$ are hence those relations where this relative increase is zero.}

\paragraph{The measure $\mu$.} Piatetsky-Shapiro and
Matheus~\cite{piatetsky1993} note that $\pdep$ and $\tau$ have the following undesirable property.

\begin{theorem}[Piatetsky-Rotem-Shapiro~\cite{piatetsky1993}]
  Given a random relation $R$ of size $N\geq 2$ containing attributes $\X$ and $\Y$, where $\X$ has $K = \size{\dom_R(\X)}$ distinct values in its active domain, the expected values of $\pdep$ and $\tau$ under random permutations of $R$ are
\scalebox{0.95}{\parbox{1.05\linewidth}{%
	\begin{align*}
		\expect_R[\pdep(\X {\tiny\to} \Y, R)] & = \pdep(\Y, R) + \frac{K-1}{N-1}(1-\pdep(\Y, R)), \\
		\expect_R[\tau(\X {\tiny\to} \Y, R)] & = \frac{\size{\dom_R(\X)}-1}{\size{R}-1}.
	\end{align*}
}}
\end{theorem}
\vspace{-1.5ex}
Thus, for a fixed distribution of $\Y$ values, $\expect_R[\pdep(\X \to \Y,R)]$
depends only on the number of distinct $\X$ values and not on their relative
frequency. Moreover, the formula for $\expect_R[\tau(\X\to\Y, R)]$ tells us that
if we have two candidate \afds with the same right hand side, $\X \to \Y$ and
$\Z \to \Y$, then if $\size{\dom_R(\Z)} > \size{\dom_R(\X)}$, we may expect
$\tau$ to score $\Z \to \Y$ better than $\X \to \Y$, regardless of any intrinsic
better relationship between $\Z$ and $\Y$ over $\X$ and $\Y$ in $R$. In
response, Piatetsky-Shapiro and Matheus compensate for this effect by
introducing the measure $\mu$ which normalizes $\pdep(\X\to \Y, R)$ with respect
to $\expect_R[\pdep(\X\to \Y, R)]$ instead of
$\pdep(\Y, R)$:\footnote{\label{footnote:muwelldef}This fraction is ill-defined
  if the denominator $1 - \expect_R[\pdep(\varphi,R)] = 0$. This only happens,
  however, when $R \models \varphi$, which we have assumed not to be the case
  throughout this section, since we have already convened to set
  $\mu(\varphi,R) = 1$ whenever $R \models \varphi$. \inConfVersion{See the
    full paper~\cite{artifacts} for a proof.}\inFullVersion{See the Appendix for a
    proof.}}
\begin{align*}
\mu(\X{\tiny\to}\Y, R) & := \frac{\pdep(\X {\tiny\to} \Y, R)-\expect_R[\pdep(\X {\tiny\to} \Y, R)]}{1-\expect_R[\pdep(\X {\tiny\to} \Y, R)]} \\
					   & = 1-\frac{1-\pdep(\X {\tiny\to} \Y, R)}{1-\pdep(\Y, R)} \frac{|R|-1}{|R|-|\dom_R(\X)|}
\end{align*}
Strictly speaking, $\mu$ is not a measure since it returns negative values when $\pdep(\X \to \Y, R) > \expect_R[\pdep(\X \to \Y, R)]$.  Because such negative $\mu$ values indicate that there is weak evidence to conclude that $\varphi$ is an \afd, we turn $\mu$ into an actual AFD measure $\muplus$ by setting
\[ \muplus(\X \to \Y, R) := \max(\mu(\X \to \Y, R), 0). \]
\inFullVersion{The baselines of $\muplus$ for $\X \to \Y$ are hence all relations where the $\pdep(\X \to \Y)$ value is smaller or equal to the expected value under random permuations.}

\subsection{Classes of AFD measures.} 
\rev{Looking at the previous definitions, we observe three different notions that are used to formally define a measure. Consequently,} we discern the following three classes (see also the second row in Table~\ref{table:comparison}):
\begin{compactenum}[(1)]
    \item The class of measures that have a notion of ``violation'' and quantify the number of violations, consisting of $\rho, g_2, g_3$, and $g'_3$. We denote this class by $\simple$ (\ssimple).
    \item The class of measures based on Shannon entropy, consisting of $\FI, \RFIplus$, and $\SFI$. We denote this class by $\shannon$ (\sshannon).

    \item The class of measures based on logical entropy, consisting of $g_1,g'_1, \pdep, \tau$, and $\muplus$ and denoted by $\logical$ (\slogical).
\end{compactenum}

\rev{\inConfVersion{In the full paper version~\cite{artifacts}}\inFullVersion{In appendix~\ref{sec:formal:comparison}} we observe that there are
  striking similarities between measures in each of these classes, allowing
  us to relate and link measures \emph{across classes}. For example, we can
  show that $\FI$ is the $\shannon$-based version of $\tau$, while $\pdep$
  is a $\logical$-based version of $g_2$. This gives a first insight into
  how measures compare formally. By means of this comparison, we have
  identified the following two new measures, which do not appear in the
  literature but which are $\shannon$ versions of existing measures. For
  completeness, we include both of these measures in our study.}
\paragraph{The new measure \SHANNONGONE.} The measure $g_1$ is based on logical entropy. We may hence view $1 - \entrop_R(\Y \mid \X)$ as the
Shannon equivalent of $g_1$, where logical entropy is replaced by Shannon entropy.
Giannella and Robertson~\cite{DBLP:journals/is/GiannellaR04} observed that 
$1 - \entrop_R(\Y \mid \X)$ has the range $[-\infty, 1]$ instead of $[0,1]$
and therefore disregard it as an \afd measure.
Nevertheless, we propose the following Shannon variant
for our comparison, obtained by limiting $1  - \entrop_R(\Y\mid\X)$ to be positive:
\[ \SHANNONGONE(\X \to \Y, R) := \max(1 - \entrop_R(\Y \mid \X), 0). \]

\paragraph{The new measure \RFInorm.}
We also observe a conceptual similarity between $\mu$ and $\RFI$: $\mu$ corrects for the  bias of $\tau$ under random permutations while $\RFI$ corrects for the bias of $\FI$ under random permutations. Despite this conceptual similarity, note that the corrections are done differently: $\mu$ corrects by taking the \emph{normalized} difference between $\pdep$ and $\expect_R[\pdep]$ while $\RFI$ corrects by taking the \emph{absolute} difference between $\FI$ and $\expect_R[\FI]$. As such, $\RFI$ is not a normalized measure. Since it is natural to ask what the normalized variant of $\RFI$ is and how it behaves, we define
\[ \RFInorm(\varphi, R) := \max\left(\frac{\FI(\varphi,R) - \expect_R[\FI(\varphi,R)]}{1 -\expect_R[\FI(\varphi,R)]}, 0\right). \]

%%% Local Variables:
%%% mode: latex
%%% TeX-master: "../main/main"
%%% End:

\section{Sensitivity Analysis}
\label{sec:eval-syn}
In this section, we investigate % in this section
the sensitivity of measures w.r.t.\ structural properties of the input relation. Specifically, we want to get insight into the measures'
ability to distinguish between an \fd $\X \to \Y$ in relation instances that were generated to satisfy the \fd, but subsequently had errors introduced so that the \fd no longer holds exactly, versus relation instances where $\X$ and $\Y$ were randomly generated. A good \fd measure should be able to consistently distinguish between these two cases, giving high scores to the former and low scores to the latter with a clear separation in scores between the two cases.

\subsection{Methodology}

There are various properties of the input relation $R$ that may affect the measures'  power to distinguish between these two aforementioned cases. We study the effect on the distinguishing power for $\X\to\Y$ w.r.t.\ the following.

(1) The \emph{error rate}, i.e., the amount of errors introduced.
A reasonable requirement for a good \afd measure is that it should be inversely proportional to the error rate: an increase in the number of introduced errors should result in a decrease of the measure value. As discussed below, this is not the case for all considered measures.

(2) Statistics of the left-hand-side $\X$ of the \fd and the right-hand-side $\Y$ of
the \fd. Specifically, define \emph{left-hand-side (LHS) uniqueness} of the
input relation $R$ as the ratio $\size{\!\dom_R(\X)} / \size{R}$, quantifying
the uniqueness of $\X$. Define the \emph{right-hand-side (RHS) skew of\/ $R$},
as the skewness of the distribution $\prob_R(\Y)$. It has previously been observed in the
literature
\cite{piatetsky1993,DBLP:conf/kdd/PennerathMV20,DBLP:conf/kdd/MandrosBV17,DBLP:journals/kais/MandrosBV20}
that measures can be biased w.r.t.\ \LHSuniqueness and \RHSskew in the sense
that measures may give very high scores to relations with high \LHSuniqueness or
\RHSskew values. However, these statistics alone do \emph{not} provide a good
signal for quantifying the approximateness of $\X \to \Y$ because they look only
at $\X$ or only at $\Y$ but not at their correlation. We note that this is in
particular true for non-linear FDs: as the number of attributes in $\X$ grows,
LHS uniqueness naturally tends to $1$. Thus, a measure that is biased
w.r.t. \LHSuniqueness risks returning too many non-linear \fds, especially on
tables with many attributes.

We have created three synthetic benchmarks, denoted $\synnoise$, $\synunique$, and $\synskew$ to study the measures' sensitivity to errors,
LHS uniqueness, and \RHSskew, respectively.
Each synthetic benchmark $\benchmark$ consists of relations $R(\X\Y)$ with
$\X = \{X\}$ and $\Y = \{Y\}$ and is partitioned into two subsets: (1)
$\benchmarknonfd$ containing relations $R$ where $\X\to\Y \not \in \schema(R)$;
and (2) $\benchmarkfd$ containing relations where $\X\to\Y \in
  \schema(R)$.
(Recall that $\schema(R)$ denotes the design schema of $R$.) Each subset employs a distinct random process to generate relations. For
relations in $\benchmarknonfd$, $\X$ and $\Y$ values are generated independently
at random, while relations in $\benchmarkfd$ are generated by first constructing
a relation $R$ such that $R \models \X\to\Y$, and then passing $R$ through a controlled noisy error channel.

\paragraph*{Generation process.}
The generation process of a relation $R$ depends on a number of parameters
that are drawn uniformly at random from the following ranges:
$\size{R}\in [100; 10000]$; $\size{\!\dom_R(\X)}\in [\frac{1}{5}\size{R}, \frac{3}{4}\size{R}]$,
$\size{\!\dom_R(\Y)}\in [5, \frac{1}{2}\size{\!\dom_R(\X)}]$; and error rate $\eta\in [0.5\%, 2\%]$.
Values for $\X$ and $\Y$ are drawn according to the Beta
distribution,
$B(\alpha,\beta)$, which is a family of continuous probability distributions
defined on the interval $[0, 1]$ in terms of two positive parameters $\alpha$
and $\beta$ that control the shape of the distribution. We consider the ranges
$\alpha\in (0,1]$ and $\beta\in[1,10]$.  For $\alpha=\beta=1$ the distribution
is uniform and for any other values it is reverse J-shaped with a right tail.
The \emph{skewness} is defined as
$\frac{2(\beta-\alpha) \sqrt{\alpha+\beta+1}}{(\alpha+\beta+2) \sqrt{\alpha
      \beta}}$ and is known to measure the asymmetry of the probability
distribution about its mean.  In particular, the skew is zero for the uniform
distribution and increasing values indicate longer tails with lower mass, that
is, a higher mass near the left end of the interval $[0,1]$.  We sample values
for $\alpha$ and $\beta$ such that the skewness is at most one (except for
\synskew below where we consider skew values up to 10).

So, for every relation $R$ the parameters $\size{R}$, $\size{\!\dom_R(\X)}$, $\size{\!\dom_R(\Y)}$, $\alpha_{\X}$, $\beta_{\X}$, $\alpha_{\Y}$, $\beta_{\Y}$, $\eta$ are chosen uniformly at random under the conditions described above.
To generate a table $R$ in $\benchmarknonfd$, the following procedure is repeated $\size{R}$ times: sample $\x\in\dom_R(\X)$ (resp., $\y\in\dom_{R}(\Y)$) according to $B(\alpha_{\X}, \beta_{\X})$ (resp., $B(\alpha_{\Y}, \beta_{\Y})$) and add $(\x,\y)$ to $R$.
To generate a table $R$ in $\benchmarkfd$,
we first construct a dictionary $D$ by, for each value $\x\in dom_R(\X)$,
assigning a value $D(\x)\in dom_R(\Y)$ drawn at random according to
$B(\alpha_{\Y}, \beta_{\Y})$. Then, we populate $R$ by adding $\size{R}$ tuples
$(\x,D(\x))$ where
$\x\in dom_R(\X)$ is drawn at random according to $B(\alpha_{\X}, \beta_{\X})$.
By construction, $R$ satisfies the \fd $\X \to \Y$. We then pass
$R$ through a controlled error channel such that, denoting by $R'$ the obtained
relation, $R'$ does not satisfy $\X \to \Y$ anymore. Concretely, we modify
$k=\lfloor \eta|R|\rfloor$ tuples $\w = (\x, D(\x))$, where $\eta$ indicates the error
rate, by randomly picking any $\tilde{\w} \in R$ with
$\tilde{\w}|_{\Y} \neq \w|_{\Y}$ and make $\tilde{\w}|_{\Y}$ the new value for
$\w|_{\Y}$.  We point out that this does not introduce any new $\Y$-values and keeps $\dom_{R}(\Y)$ stable. We also experimented with other error channels that introduce new $\Y$ values, but the results were similar and are therefore omitted.
Note that $\X$ is not modified, and therefore $\prob_{R'}(\X) = \prob_{R}(\X)$.
\inFullVersion{We note that the generation process is related to the one from
Zhang et al.~\cite{ZhangGuoEtAl2020} but with the addition of value
distributions for both $\X$ and $\Y$ based on the Beta distribution.}

The three synthetic benchmarks are created by controlling one of the parameters in the parameter set as follows. Every benchmark $\benchmark$ consists of 2500 $\benchmarknonfd$ tables and 2500 $\benchmarkfd$ tables.
\paragraph*{Benchmark \synnoise.}
We iteratively increase the error rate $\eta$ from $0\%$ to $10\%$ in 50 steps and generate 50 relations in $\synnoisefd$ per step, varying all other parameters as described above. \synnoise
is then extended with 2500 tables generated in $\synnoisenonfd$.

\paragraph*{Benchmark \synunique.}
We construct \synunique by %generating 5000 relations, 2500 of both types, 
iteratively increasing \LHSuniqueness from $\frac{1}{5}|{R}|$ to $10|{R}|$ in 50 steps and by generating for every step 50 tables in $\synuniquefd$ and $\synuniquenonfd$.

\paragraph*{Benchmark \synskew.}
We construct \synskew by %generating 5000 tables, 2500 of both types, 
iteratively increasing \RHSskew from 0 to 10 in 50 steps and by generating for every step 50 tables in $\synskewfd$ and $\synskewnonfd$.

%%% Local Variables:
%%% mode: latex
%%% TeX-master: "../main/main"
%%% End:

\begin{figure*}
	\small
	\centering
	\pgfplotsset{
		width=5cm,
		every axis plot/.append style={thick},
		every axis y label/.style={
				at={(ticklabel cs:0.5)},rotate=90,anchor=near ticklabel,
			},
		every axis x label/.style={
				at={(ticklabel cs:0.5)},anchor=near ticklabel,
			},
		ymin=-0.05, ymax=1.05,
		ymajorgrids=true,
		xmajorgrids=true,
		legend columns=-1,
		legend style={draw=none},
	}
	\begin{minipage}{0.32\textwidth}
		\centering
		\hspace{2em}\ref{syn_noise_simple_legend}
		\\
		\begin{tikzpicture}
			\begin{axis}[
					legend to name=syn_noise_simple_legend,
					xlabel={Error rate},
					ylabel={Separation on $\synnoise$},
					xmin=-0.0, xmax=0.09,
				]
				\addplot [plotColor1] table [x=error,y=difference] {plotdata/syn_error_copy_rho.dat};
				\addlegendentry{$\rho$}
				\addplot [plotColor2] table [x=error,y=difference] {plotdata/syn_error_copy_g2.dat};
				\addlegendentry{$g_2$}
				\addplot [plotColor3] table [x=error,y=difference] {plotdata/syn_error_copy_g3.dat};
				\addlegendentry{$g_3$}
				\addplot [plotColor4] table [x=error,y=difference] {plotdata/syn_error_copy_g3_prime.dat};
				\addlegendentry{$g'_3$}
			\end{axis}
		\end{tikzpicture}
		\begin{tikzpicture}
			\begin{axis}[
					legend to name=syn_keylike_simple_legend,
					xlabel={LHS Uniqueness},
					ylabel={Separation on $\synunique$},
					xmin=0.1, xmax=0.9,
				]
				\addplot [plotColor1] table [x=lhs_uniq,y=difference] {plotdata/syn_keylike_rho.dat};
				\addlegendentry{$\rho$}
				\addplot [plotColor2] table [x=lhs_uniq,y=difference] {plotdata/syn_keylike_g2.dat};
				\addlegendentry{$g_2$}
				\addplot [plotColor3] table [x=lhs_uniq,y=difference] {plotdata/syn_keylike_g3.dat};
				\addlegendentry{$g_3$}
				\addplot [plotColor4] table [x=lhs_uniq,y=difference] {plotdata/syn_keylike_g3_prime.dat};
				\addlegendentry{$g'_3$}
			\end{axis}
		\end{tikzpicture}
		\begin{tikzpicture}
			\begin{axis}[
					legend to name=syn_rhsskew_simple_legend,
					xlabel={RHS Skew},
					ylabel={Separation on $\synskew$},
					xmin=0.0, xmax=9.0,
				]
				\addplot [plotColor1] table [x=rhs_skew,y=difference] {plotdata/syn_rhsskew_rho.dat};
				\addlegendentry{$\rho$}
				\addplot [plotColor2] table [x=rhs_skew,y=difference] {plotdata/syn_rhsskew_g2.dat};
				\addlegendentry{$g_2$}
				\addplot [plotColor3] table [x=rhs_skew,y=difference] {plotdata/syn_rhsskew_g3.dat};
				\addlegendentry{$g_3$}
				\addplot [plotColor4] table [x=rhs_skew,y=difference] {plotdata/syn_rhsskew_g3_prime.dat};
				\addlegendentry{$g'_3$}
			\end{axis}
		\end{tikzpicture}
	\end{minipage}
	\begin{minipage}{0.34\textwidth}
		\centering
		%%%%% HACK: legative whitespace to shift the legend to the left
		\ref{syn_noise_shannon_legend}
		\\
		\begin{tikzpicture}
			\begin{axis}[
					legend to name=syn_noise_shannon_legend,
					xlabel={Error rate},
					xmin=0.0, xmax=0.09,
				]
				\addplot [plotColor1] table [x=error,y=difference] {plotdata/syn_error_copy_shannon_g1_prime.dat};
				\addlegendentry{$\SHANNONGONE$}
				\addplot [plotColor2] table [x=error,y=difference] {plotdata/syn_error_copy_fraction_of_information.dat};
				\addlegendentry{$\FI$}
				\addplot [plotColor3] table [x=error,y=difference] {plotdata/syn_error_copy_reliable_fraction_of_information_prime.dat};
				\addlegendentry{$\RFIplus$}
				\addplot [plotColor4] table [x=error,y=difference] {plotdata/syn_error_copy_fraction_of_information_prime.dat};
				\addlegendentry{$\RFInorm$}
				\addplot [plotColor5] table [x=error,y=difference] {plotdata/syn_error_copy_smoothed_fraction_of_information.dat};
				\addlegendentry{$\SFI$}
			\end{axis}
		\end{tikzpicture}
		\begin{tikzpicture}
			\begin{axis}[
					legend to name=syn_keylike_shannon_legend,
					xlabel={LHS Uniqueness},
					xmin=0.1, xmax=0.9,
				]
				\addplot [plotColor1] table [x=lhs_uniq,y=difference] {plotdata/syn_keylike_shannon_g1_prime.dat};
				\addlegendentry{$\SHANNONGONE$}
				\addplot [plotColor2] table [x=lhs_uniq,y=difference] {plotdata/syn_keylike_fraction_of_information.dat};
				\addlegendentry{$\FI$}
				\addplot [plotColor3] table [x=lhs_uniq,y=difference] {plotdata/syn_keylike_reliable_fraction_of_information_prime.dat};
				\addlegendentry{$\RFIplus$}
				\addplot [plotColor4] table [x=lhs_uniq,y=difference] {plotdata/syn_keylike_fraction_of_information_prime.dat};
				\addlegendentry{$\RFInorm$}
				\addplot [plotColor5] table [x=lhs_uniq,y=difference] {plotdata/syn_keylike_smoothed_fraction_of_information.dat};
				\addlegendentry{$\SFI$}
			\end{axis}
		\end{tikzpicture}
		\begin{tikzpicture}
			\begin{axis}[
					legend to name=syn_rhsskew_shannon_legend,
					xlabel={RHS Skew},
					xmin=0.0, xmax=9.0,
				]
				\addplot [plotColor1] table [x=rhs_skew,y=difference] {plotdata/syn_rhsskew_shannon_g1_prime.dat};
				\addlegendentry{$\SHANNONGONE$}
				\addplot [plotColor2] table [x=rhs_skew,y=difference] {plotdata/syn_rhsskew_fraction_of_information.dat};
				\addlegendentry{$\FI$}
				\addplot [plotColor3] table [x=rhs_skew,y=difference] {plotdata/syn_rhsskew_reliable_fraction_of_information_prime.dat};
				\addlegendentry{$\RFIplus$}
				\addplot [plotColor4] table [x=rhs_skew,y=difference] {plotdata/syn_rhsskew_fraction_of_information_prime.dat};
				\addlegendentry{$\RFInorm$}
				\addplot [plotColor5] table [x=rhs_skew,y=difference] {plotdata/syn_rhsskew_smoothed_fraction_of_information.dat};
				\addlegendentry{$\SFI$}
			\end{axis}
		\end{tikzpicture}
	\end{minipage}
	\begin{minipage}{0.32\textwidth}
		\centering
		\ref{syn_noise_logical_legend}
		\\
		\begin{tikzpicture}
			\begin{axis}[
					legend to name=syn_noise_logical_legend,
					xlabel={Error rate},
					xmin=-0.0, xmax=0.09,
				]
				\addplot [plotColor1] table [x=error,y=difference] {plotdata/syn_error_copy_g1_prime.dat};
				\addlegendentry{$g_1,g'_1$}
				\addplot [plotColor2] table [x=error,y=difference] {plotdata/syn_error_copy_pdep.dat};
				\addlegendentry{$\pdep$}
				\addplot [plotColor3] table [x=error,y=difference] {plotdata/syn_error_copy_tau.dat};
				\addlegendentry{$\tau$}
				\addplot [plotColor4] table [x=error,y=difference] {plotdata/syn_error_copy_mu_prime.dat};
				\addlegendentry{$\muplus$}
			\end{axis}
		\end{tikzpicture}
		\begin{tikzpicture}
			\begin{axis}[
					legend to name=syn_keylike_logical_legend,
					xlabel={LHS Uniqueness},
					xmin=0.1, xmax=0.9,
				]
				\addplot [plotColor1] table [x=lhs_uniq,y=difference] {plotdata/syn_keylike_g1_prime.dat};
				\addlegendentry{$g_1,g'_1$}
				\addplot [plotColor2] table [x=lhs_uniq,y=difference] {plotdata/syn_keylike_pdep.dat};
				\addlegendentry{$\pdep$}
				\addplot [plotColor3] table [x=lhs_uniq,y=difference] {plotdata/syn_keylike_tau.dat};
				\addlegendentry{$\tau$}
				\addplot [plotColor4] table [x=lhs_uniq,y=difference] {plotdata/syn_keylike_mu_prime.dat};
				\addlegendentry{$\muplus$}
			\end{axis}
		\end{tikzpicture}
		\begin{tikzpicture}
			\begin{axis}[
					legend to name=syn_rhsskew_logical_legend,
					xlabel={RHS Skew},
					xmin=0.0, xmax=9.0,
				]
				\addplot [plotColor1] table [x=rhs_skew,y=difference] {plotdata/syn_rhsskew_g1_prime.dat};
				\addlegendentry{$g_1,g'_1$}
				\addplot [plotColor2] table [x=rhs_skew,y=difference] {plotdata/syn_rhsskew_pdep.dat};
				\addlegendentry{$\pdep$}
				\addplot [plotColor3] table [x=rhs_skew,y=difference] {plotdata/syn_rhsskew_tau.dat};
				\addlegendentry{$\tau$}
				\addplot [plotColor4] table [x=rhs_skew,y=difference] {plotdata/syn_rhsskew_mu_prime.dat};
				\addlegendentry{$\muplus$}
			\end{axis}
		\end{tikzpicture}
	\end{minipage}
	\vspace{-1ex}
	%	\caption{\label{fig:synexperiments} Increasing error rates, \LHSuniqueness levels or \RHSskew levels impacts the ability to differentiate between synthetically generated \fds and non-\fds for most measures. The plots show the separation  in measure values between generated \fds versus relations generated at random on the \synnoise (top), \synunique (middle) and \synskew (bottom) benchmarks.
	\caption{\label{fig:synexperiments} Increasing error rates, \LHSuniqueness levels or \RHSskew levels impacts most measures' ability to seperate between $\benchmarkfd$ and $\benchmarknonfd$. The plots show the separation on  \synnoise (top), \synunique (middle) and \synskew (bottom) benchmarks.
	}
\end{figure*}
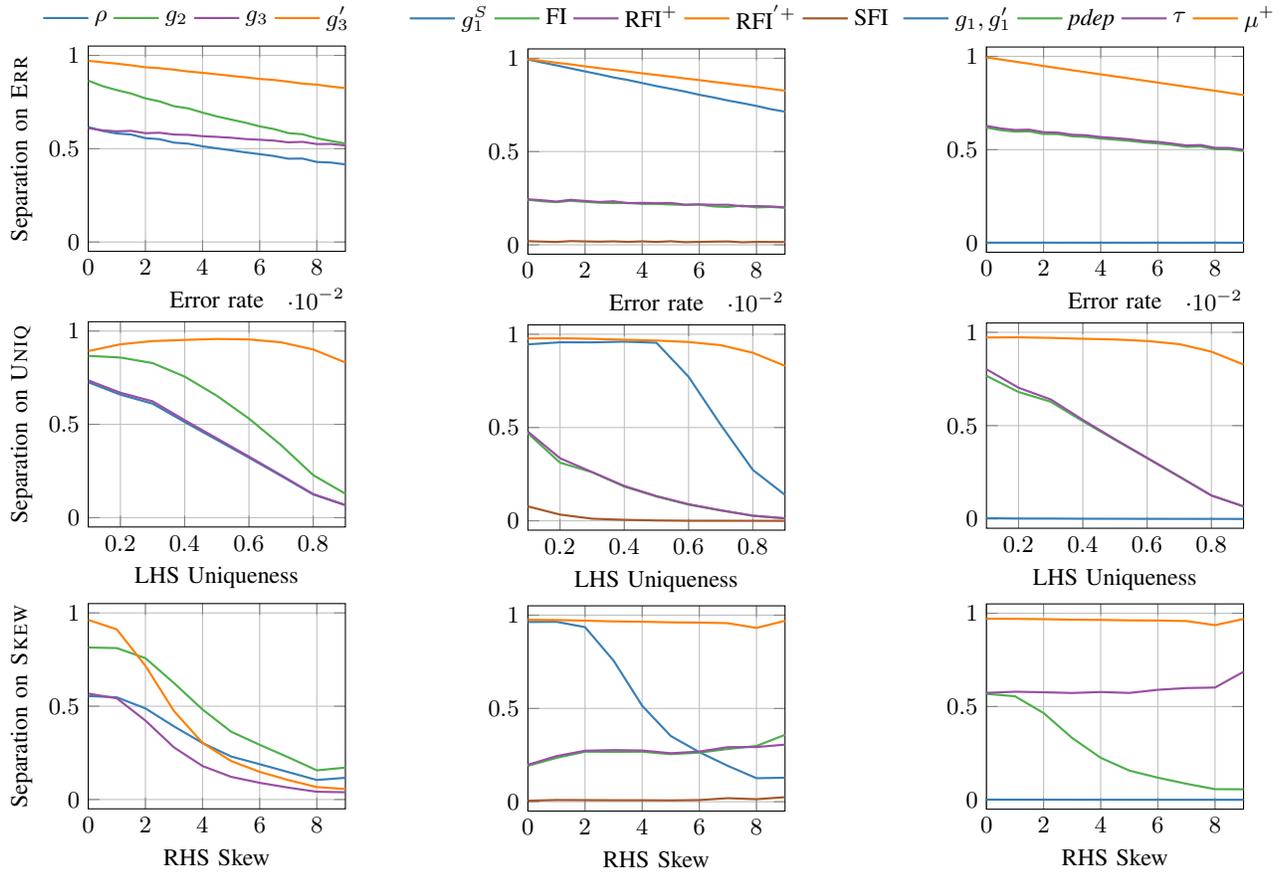

\subsection{Results}
\label{sec:syn:res}

We describe the results on the basis of Figure~\ref{fig:synexperiments}, which, for each benchmark $\benchmark$ and measure $f$ plots $\delta(f,\benchmark)$, the difference between average measure values on $\benchmarkfd$ and average measure values on $\benchmarknonfd$,
\[ \delta(f,\benchmark) := \text{avg}_{R \in \benchmarkfd} f(\X {\scriptsize\to} \Y, R) \ - \
	\text{avg}_{R \in \benchmarknonfd} f(\X {\scriptsize\to} \Y, R). \] We also call
$\delta(f,\benchmark)$ the \emph{separation} of $f$ on
$\benchmark$. When it is small, $f$ cannot distinguish between cases
where $\X$ and $\Y$ are sampled independently at random and where
data is generated according to our generation process for $\benchmarkfd$.
Values for $g_1$ and $g'_1$ are grouped in Figure~\ref{fig:synexperiments} as
their separations are identical. Our conclusions are summarized in
Table~\ref{table:comparison}.

\paragraph*{Error rate.}
The top row of Figure~\ref{fig:synexperiments} plots the separation on \synnoise as a function of error rate
$\eta$. For $g_1$ and $g'_1$, the separation is zero, while for $\SFI$ it is nearly zero. This means that these measures have limited distinguishing power and are not well-suited as a yardstick
for assessing the amount of errors w.r.t. an \fd.
For all other measures, there is a
clear separation, albeit less pronounced for $\FI$ and $\RFIplus$.
As expected, when the error level increases the separation decreases, save for $g_1$, $g'_1$, and $\SFI$ where it remains constant. The measures hence become less certain of having found an \afd as the error rate increases.
While $\FI$ and $\RFIplus$ also decrease as $\eta$ increases, this decrease is less steep than for the other measures.

\paragraph*{\LHSuniqueness.}
The middle row of Figure~\ref{fig:synexperiments} shows the separation on \synunique as a function of \LHSuniqueness.
For $g_1$, $g'_1$ and $\SFI$, we see the same behavior as on \synnoise: their separation is (nearly) zero; they hence lack distinguishing power.
Because it would be misleading to label $g_1$, $g'_1$ and $\SFI$ as being insensitive to \LHSuniqueness, we indicate in Table~\ref{table:comparison} that \LHSuniqueness is inapplicable with the symbol \nap.
The distinguishing power of $g'_3$, $\RFInorm$, and $\muplus$
is not affected by \LHSuniqueness as the separation remains large for all values of \LHSuniqueness. We do observe that the separation decreases slightly for very large \LHSuniqueness levels, indicating that these measures become less confident to have found an \fd $\X\to\Y$ in a relation $R$ when $\pi_{\X}(R)$ contains fewer duplicates.
%\frank{$g'_3$ initially increases and then decreases, is this correct?}

For %all 
other measures the separation drops as \LHSuniqueness increases, tending to zero
at maximum \LHSuniqueness levels. % There are two underlying causes.
% for this decline in separation.
% Specifically, for $\rho$, $g_2$, $\SHANNONGONE$, $g_3$, $\FI$, $\pdep$, and $\tau$ the average measure values on $\synuniquenonfd$ increase,  % with the \LHSuniqueness level
% eventually approaching the measure values on $\synuniquefd$. For $\RFIplus$ and $\SFI$, by contrast, the average measure values over $\synuniquefd$ decrease towards zero for increasing \LHSuniqueness, eventually reaching the value on $\synuniquenonfd$. Note that this decrease is already observable for  small \LHSuniqueness values.
%
These measures are hence % This means that all mentioned measures
% mentioned in this paragraph 
biased w.r.t.\ \LHSuniqueness% (either by inflating the scores of non-FDs or deflating the scores of FDs)
.
As discussed before,
it will therefore prove problematic to discover non-linear \afds by means of these measures.

\paragraph*{\RHSskew.}
The bottom row of Figure~\ref{fig:synexperiments} shows the separation on $\synskew$ as a function of
% average measure value for $\synskewfd$ and $\synskewnonfd$ for different levels of
\RHSskew.
The measures $g_1$, $g'_1$, and $\SFI$ exhibit the same behavior as before, with (nearly) zero separation. % we see the same behavior as before:
% the solid and dotted lines coincide or are very close together
% which severely hampers their distinguishing power.
We indicate the corresponding cells in Table~\ref{table:comparison} with the symbol \nap.
The distinguishing power of all \simple measures, as well as $\SHANNONGONE$, and $\pdep$, drops when \RHSskew increases.
%\inFullVersion{As illustrated by Figure~\ref{fig:synexperiments-additional} available in Appendix~\ref{sec:syn-extra},}\inConfVersion{Through a more detailed analysis,} we find that this is because the average measure values on $\synskewnonfd$ increase for large \RHSskew levels,  approaching  the  average values over $\synskewfd$ which remain constant.
% Indeed, over $\synskewnonfd$, the average measure values approach the  values over $\synskewfd$ as  \RHSskew increases.
These measures are thus biased w.r.t.\ \RHSskew. %: the score for $\X\to\Y$
%increases solely on the basis of $\Y$ and independent of $\X$ even if relations are generated by a process that sampled $\X$ and $\Y$ independently at random.
By contrast, $\FI$, $\RFIplus$, $\RFInorm$, $\tau$, and $\muplus$ correct for this behavior and are insensitive to % . Furthermore, $\FI$ and $\tau$ slightly decrease at higher levels of
\RHSskew.

\paragraph*{Conclusion.}
The measures $g_1$, $g'_1$ and $\SFI$ are the least suitable \afd measures since, by contrast to the other measures, they do not clearly separate relations in $\benchmarkfd$ from relations in  $\benchmarknonfd$ for any of the three considered sensitivity parameters. % In particular, they are invariant under increasing error levels, which is undesirable.
The measures $g'_3$, \RFInorm, and $\muplus$ have a built-in mechanism that corrects for \LHSuniqueness, which is a most desirable property when discovering non-linear \fds. The  \shannon measures (save $\SHANNONGONE$ and $\SFI$), $\tau$, and $\muplus$ correct for \RHSskew. The most desirable measures are therefore
$\RFInorm$ and $\muplus$ as they both are insensitive to \LHSuniqueness and \RHSskew, and are inversely proportional to the error level.

% Using the \syn data generation process, we present three experiments that focus on different properties: \texttt{error level}, \texttt{LHS uniqueness} and \texttt{RHS skewness}.
% \begin{table*}
% 	\caption{AFD measure differences between synthetically generated FDs and non-FDs for noise.}
% 	\csvautobooktabular{../tables/SYN_noise_copy_differences.csv}
% \end{table*}

% \begin{table*}
% 	\caption{AFD measure differences between synthetically generated FDs and non-FDs for LHS uniqueness.}
% 	\csvautobooktabular{../tables/SYN_lhs_uniqueness_differences.csv}
% \end{table*}

% \begin{table*}
% 	\caption{AFD measure differences between synthetically generated FDs and non-FDs for RHS skewness.}
% 	\csvautobooktabular{../tables/SYN_rhs_skewness_differences.csv}
% \end{table*}

% The purpose of this section is to show the following:
% \begin{itemize}
% 	\item
% 	      Noise: expect measures to decrease when noise increases
% 	\item LHS uniqueness: g3 and g3\_N behave differently, some measures are not affected (mu, g3\_N)
% 	\item RHS skewness: only mu separates between FD and non-FDs
% \end{itemize}

%%% Local Variables:
%%% mode: latex
%%% TeX-master: "../main/main"
%%% End:

%%% Local Variables:
%%% mode: latex
%%% TeX-master: "../main/main"
%%% End:

\section{Evaluation on Real-World Data}
\label{sec:eval-rwd}
\label{sec:rwd}
In this section, we compare the effectiveness of the described \afd measures on
% \rwd: a new benchmark consisting of real-world tables which hence exhibits data
% distributions as well as data errors that occur in practice. Because of space limitations, we focus on our most important findings. \inConfVersion{Complementary information and additional experiments may be found in the extended version of this paper~\cite{artifacts}.} \inFullVersion{Complementary information and additional experiments may be found in the Appendix.}
real-world tables which exhibit data distributions as well as data errors that
occur in practice. Because of space limitations, we focus on our most important
findings and refer to
\inConfVersion{the extended version of this paper~\cite{artifacts}.}
\inFullVersion{complementary information and additional experiments in the Appendix.}

\subsection{Overview}
\paragraph*{\fds in relations with \textsf{NULLs}.} The relations that we consider in this section come from practical domains and often also contain \textsf{NULL} values. Because it is unclear whether two distinct occurrences of a \textsf{NULL} should be considered the same value, or distinct values,
there is no clear semantics of \fds in the presence of \textsf{NULL} values. We therefore ignore \textsf{NULL} values when checking \fd satisfaction and calculating measure scores. That is, if $R(\W)$ is a relation with \textsf{NULL}s and $\varphi = \X \to \Y$ an \fd, then we consider $\varphi$ to be satisfied if it is satisfied in the subrelation $R'$ of $R$ consisting of all tuples $\w \in R$ for which $\w(A) \not = \textsf{NULL}$ for all $A \in \X\Y$. Similarly, the score of measure $f$ on $(\varphi,R)$ is computed by computing $f(\varphi, R')$ instead.

\paragraph*{Real world data benchmark (\rwd).}
We created the \rwd benchmark as follows. We started by considering
% all relations mentioned in \cite{birnickBlasiusEtAl2020},\footnote{\url{https://owncloud.hpi.de/s/j6Z0yvXC0qhtGCk/download}} which collects the real-world relations most commonly used in the dependency discovery literature. This base set was extended with 
% the relation Adult\footnote{\url{https://archive.ics.uci.edu/ml/datasets/Adult}} used, e.g., in \cite{ChiangMiller2008,DBLP:journals/pvldb/0001N18,DBLP:conf/icde/WeiL19}.
all relations mentioned in \cite{birnickBlasiusEtAl2020}, which collects the
real-world relations most commonly used in the dependency discovery
literature. This base set was extended with the relation Adult used, e.g., in
\cite{ChiangMiller2008,DBLP:journals/pvldb/0001N18,DBLP:conf/icde/WeiL19}.
Since design schemas for these relations are unavailable, we manually created
them as follows. First, in order to ensure semantically sound design schemas, we
restricted our attention to the subset of relations that have a generally
interpretable domain. Further, to keep the manual annotation endeavor
manageable, we restricted ourselves to relations that have no more than 50
columns and to linear FDs. This results in \BenchNumTables relations, listed in
Table~\ref{tab:rwd}. For each relation $R$, we enumerate all candidate linear
FDs (i.e., pairs
$(X, Y): \exists \w \in R, \w(X) \not = \textsf{NULL} \land \w(Y) \not =
  \textsf{NULL}$).
\rev{We manually validate whether a candidate FD is semantically
  meaningful, and is hence part of the design schema or not, if its
 $g_3$-score is $\geq 0.5$. While we risk %limiting validation to candidates with $g_3 \geq 0$ does risk
    missing semantically meaningful \fds this way, note that
  a $g_3$-score $< 0.5$ means that we need to remove more
  than $50\%$ of the tuples to obtain a subrelation that satisfies the candidate
  \fd, making it an improbable candidate for the design schema.
  We identified 1170 candidate \fds to inspect. Two individuals manually inspected
  each candidate.  Non-matching decisions (i.e. one saw a candidate as valid
  whereas the other did not) were discussed until a consensus was reached.  We
  observe that each validated semantically meaningful candidate \fd has a
  $g_3$-score $\geq 0.99$, strengthening our impressions that it is unlikely that we have missed semantically meaningful FDs.}

% marcel: I removed the following part, as it is repeated at the end of the "methodology" section
%\footnote{Full benchmark information is available at \cite{artifacts}, where inspected and excluded candidates \fds are listed as CSV files under the \textsf{data} folder.}

In this manner, we derive for
each benchmark relation
$R$ its design schema $\schema(R)$. This set of \fds is
partitioned into two sets: $\perfectFD(R) := \{ \varphi \in \schema(R) \mid R \models \varphi \}$ the set of \emph{perfect} (design) \fds and $\designAFD(R) := \{ \varphi \in \schema(R) \mid R \not \models \varphi \}$ the set of \emph{approximate} (design) \fds. In particular, $\designAFD{R}$ forms the ground truth of \fds to discover during \afd discovery on $R$.
% \begin{align*}
%   \perfectFD(R) := \{ \varphi \in \schema(R) \mid R \models \varphi \},
%   \designAFD(R) := \{ \varphi \in \schema(R) \mid R \not \models \varphi \}.
% \end{align*}
% We will refer to these sets as the \emph{perfect} (design) \fds
% and \emph{approximate} (design) \fds, respectively. In particular,
% $\designAFD(R)$ forms the ground truth of \fds to discover during \afd discovery
% on $R$.

Table~\ref{tab:rwd} shows statistics of the obtained benchmark. In total, we obtain \BenchTotalDesignFDs design \fds across all relations in \rwd, of which \BenchTotalPerfectFDs are perfect design \fds and \BenchTotalAFDs are approximate design \fds.
To appreciate the difficulty of the \afd discovery task, it is worth pointing
out that the search space during \afd discovery consists of \TotalNumberOfInvalidFDs
candidate \fds across all relations in \rwd. Out of these, only a
small number (\BenchTotalAFDs) are \afds, which
emphasizes the intrinsic difficulty of \afd discovery and illustrates the need for good measures to distinguish \afds from the rest of the search space.

\begin{table}[t]
  \small
  \caption{\label{tab:rwd} Overview of relations in \rwd benchmark. % The \char"0023 insp column indicates the number of manually inspected candidates when determining the design schema.
  }
  \setlength\tabcolsep{4.0pt}
  \begin{tabular}[t]{@{}llllll@{}}
    \toprule
    Relation $R$ & \char"0023 rows & \char"0023 attrs & \char"0023 insp & \char"0023$\perfectFD(R)$ & \char"0023$\designAFD(R)$ \\\midrule
    \begin{filecontents*}{tables/dataset_descriptions.csv}
name,rows,attributes,candFDs,FDs,AFDs
{\adultDataNumber}\quad{\adultData},32561,15,111,2,0
{\claimsDataNumber}\quad{\claimsData},97231,13,42,2,2
{\dblpDataNumber}\quad{\dblpData},10000,34,368,75,2
{\hospitalDataNumber}\quad{\hospitalData},114919,15,74,22,7
{\taxDataNumber}\quad{\taxData},1000000,15,95,3,0
{\gathAgentDataNumber}\quad{\gathAgentData},72737,18,55,5,2
{\gathAreaDataNumber}\quad{\gathAreaData},137710,11,43,3,2
{\gathDataNumber}\quad{\gathData},90991,35,64,0,1
{\identTaxonDataNumber}\quad{\identTaxonData},562958,3,2,0,1
{\identDataNumber}\ \ {\identData},91799,38,85,14,0
\end{filecontents*}
\csvreader[
      head to column names,
    late after line=                                                                                                            \\,
    late after last line=                                                                                                       \\\bottomrule,
    ]
    {tables/dataset_descriptions.csv}{}%
    {\name       & \rows           & \attributes      & \candFDs        & \FDs                      & \AFDs}                    %
  \end{tabular}
\end{table}

%%% Local Variables:
%%% mode: latex
%%% TeX-master: "../main/main"
%%% End:

\paragraph*{Methodology.} We are interested in comparing the suitability of \afds measures for the purpose
of \afd discovery.

Therefore, we compare \afd measures as follows. Remember from
Section~\ref{sec:measures} that every AFD measure $f$ and every threshold
$\epsilon \in [0,1]$ naturally induces a discovery algorithm $\alg_f^\epsilon$
which, on input relation $R(\W)$, returns all \fds $\varphi$ over $\W$ with $R \not \models \varphi$ and $f(\varphi,R) \in [\epsilon,1[$. In this respect, every
measure hence defines a class $\class_f$ of discovery algorithms, namely
$\class_f = \{ \alg_f^\epsilon \mid 0 \leq \epsilon < 1\}$. Given a subset $\benchmark$ of benchmark relations, we compare the
effectiveness of measures on $\benchmark$ by computing the area under the precision-recall
%\footnote{\url{https://en.wikipedia.org/wiki/Precision_and_recall}}
curve
(AUC-PR) of $\class_f$ for each measure $f$, where the PR-curve is the set
$\{ (\recall(\alg, \benchmark), \precision(\alg,\benchmark)) \mid \alg \in \class_f\}$. Here, $\recall(\alg, \benchmark)$ and $\precision(\alg,\benchmark)$ denote recall and precision of $A$ on $\benchmark$, respectively.
It is known that PR curves are well-suited to visualize the tradeoff between
precision and recall at various values of $\epsilon$ when the prediction classes
are very imbalanced, which is the case here. So, the measure with the highest
AUC-PR score is the measure providing the best such tradeoff.

Furthermore, to obtain a more fine-grained view of measure performance on the
level of each Relation $R$ individually we also report the \emph{rank at max recall}:
%\begin{align*}
$	\rankAtMaxRecall(f,R) := \size{A^\epsilon_f}$, with $\epsilon = min(f(\designAFD(R))$.
%\end{align*}
Intuitively, $\rankAtMaxRecall(f,R)$ indicates how many candidate \afds need to be examined
when processing them in decreasing order of $f$-score to find all of $\designAFD(R)$.

Since $\SFI$ is parameterized by a parameter $\alpha$ it is not one
measure but a collection of measures. We performed experiments with the same
values of $\alpha$ as in the original \SFI
paper~\cite{DBLP:conf/kdd/PennerathMV20}, namely $\alpha \in \{0.5, 1,2\}$. Because the performance of $\alpha = 0.5$ consistently dominates the
performance of $\alpha \in \{1,2\}$, we only report the performance of $\SFI$
for $\alpha = 0.5$ in what follows.

We implemented all measures in a Python library. This library, together with the benchmark datasets is publicly available~\cite{artifacts}.
Given a candidate \fd $\X \to \Y$, computing the score for most measures is straightforward, requiring only the evaluation of the given formula. For \RFI, \RFInorm and \SFI, which are the most complex to compute, we use the currently best known algorithms, for \RFI and \RFInorm the one of \cite{DBLP:conf/kdd/MandrosBV17}, for \SFI the one of \cite{DBLP:conf/kdd/PennerathMV20}.

%%% Local Variables:
%%% mode: latex
%%% TeX-master: "../main/main"
%%% End:

\begin{figure*}[!bht]
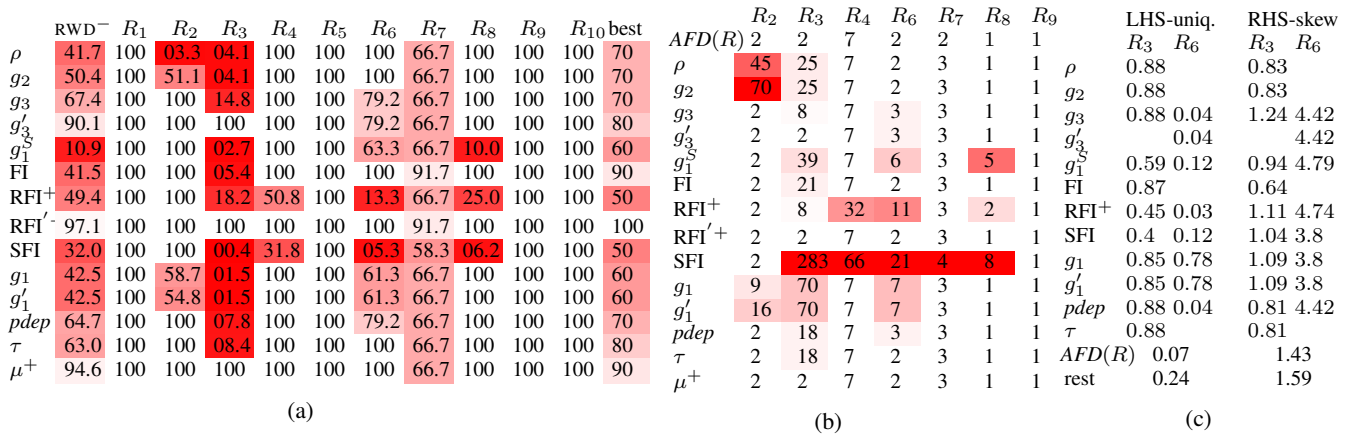

	\centering
	\begin{subfigure}[t]{0.45\textwidth}
		%%%% PR-AUC heatmaps
		\footnotesize
		\pgfplotstableset{
			heatmap/.style={
					col sep=comma,
					column type={p{0.24cm}},
					postproc cell content/.code={\pgfkeysalso{@cell content=\!\!\cellcolor{white!##1!red}##1}},
				}
		}
		\begin{filecontents*}{tables/table_aucs.csv}
,rwdminus,\adultDataNumber,\claimsDataNumber,\dblpDataNumber,\hospitalDataNumber,\taxDataNumber,\gathAgentDataNumber,\gathAreaDataNumber,\gathDataNumber,\identTaxonDataNumber,\identDataNumber,best
$\rho$,41.7,100,03.3,04.1,100,100,100,66.7,100,100,100,70
$g_2$,50.4,100,51.1,04.1,100,100,100,66.7,100,100,100,70
$g_3$,67.4,100,100,14.8,100,100,79.2,66.7,100,100,100,70
$g'_3$,90.1,100,100,100,100,100,79.2,66.7,100,100,100,80
$\SHANNONGONE$,10.9,100,100,02.7,100,100,63.3,66.7,10.0,100,100,60
$\FI$,41.5,100,100,05.4,100,100,100,91.7,100,100,100,90
$\RFIplus$,49.4,100,100,18.2,50.8,100,13.3,66.7,25.0,100,100,50
$\RFInorm$,97.1,100,100,100,100,100,100,91.7,100,100,100,100
$\SFI$,32.0,100,100,00.4,31.8,100,05.3,58.3,06.2,100,100,50
$g_1$,42.5,100,58.7,01.5,100,100,61.3,66.7,100,100,100,60
$g'_1$,42.5,100,54.8,01.5,100,100,61.3,66.7,100,100,100,60
$\pdep$,64.7,100,100,07.8,100,100,79.2,66.7,100,100,100,70
$\tau$,63.0,100,100,08.4,100,100,100,66.7,100,100,100,80
$\muplus$,94.6,100,100,100,100,100,100,66.7,100,100,100,90
\end{filecontents*}
\pgfplotstabletypeset[
			heatmap,
			columns/0/.style={
					column name={},
					string type,
					column type={p{0.4cm}},
					postproc cell content/.code={},
				},
			columns/rwdminus/.style={
					column name={\!\!\!\!\!\rwdminus},
				},
			columns/best/.style={
					column name={\!\!\!best},
					column type={p{0.2cm}},
				}
		]
		{tables/table_aucs.csv}
		\caption{\label{tab:auc}}% Heatmap of PR-AUC scores. % at the benchmark and relation level, expressed as percentage. The last column expresses the fraction of relations on which a measure yields the highest AUC score.
	\end{subfigure}
	\hfill
	\begin{subfigure}[t]{0.25\textwidth}
		%%%% rank @ max recall table
		\footnotesize
		\begin{filecontents*}{tables/recall_results_10_rwd_pessimistic_t.csv}
,\claimsDataNumber,\dblpDataNumber,\hospitalDataNumber,\gathAgentDataNumber,\gathAreaDataNumber,\gathDataNumber,\identTaxonDataNumber
AFD,2,2,7,2,2,1,1
$\rho$,45,25,7,2,3,1,1
$g_2$,70,25,7,2,3,1,1
$g_3$,2,8,7,3,3,1,1
$g'_3$,2,2,7,3,3,1,1
$\SHANNONGONE$,2,39,7,6,3,5,1
$\FI$,2,21,7,2,3,1,1
$\RFIplus$,2,8,32,11,3,2,1
$\RFInorm$,2,2,7,2,3,1,1
$\SFI$,2,283,66,21,4,8,1
$g_1$,9,70,7,7,3,1,1
$g'_1$,16,70,7,7,3,1,1
$\pdep$,2,18,7,3,3,1,1
$\tau$,2,18,7,2,3,1,1
$\muplus$,2,2,7,2,3,1,1
\end{filecontents*}
\pgfplotstabletypeset[
			col sep=comma,
			column type={p{0.2cm}},
			columns/0/.style={
					column name={},
					string type,
					column type={p{0.6cm}},
				},
			every row 0 column 0/.style={
					postproc cell content/.code={\pgfkeysalso{@cell content=\!$\designAFD(R)$}},
				},
			every row 1 column 1/.style={
					postproc cell content/.code={\pgfkeysalso{@cell content=\cellcolor{red!63!white}##1}},
				},
			every row 1 column 2/.style={
					postproc cell content/.code={\pgfkeysalso{@cell content=\cellcolor{red!8!white}##1}},
				},
			every row 1 column 3/.style={
					postproc cell content/.code={\pgfkeysalso{@cell content=\cellcolor{red!0!white}##1}},
				},
			every row 1 column 4/.style={
					postproc cell content/.code={\pgfkeysalso{@cell content=\cellcolor{red!0!white}##1}},
				},
			every row 1 column 5/.style={
					postproc cell content/.code={\pgfkeysalso{@cell content=\cellcolor{red!0!white}##1}},
				},
			every row 1 column 6/.style={
					postproc cell content/.code={\pgfkeysalso{@cell content=\cellcolor{red!0!white}##1}},
				},
			every row 1 column 7/.style={
					postproc cell content/.code={\pgfkeysalso{@cell content=\cellcolor{red!0!white}##1}},
				},
			every row 2 column 1/.style={
					postproc cell content/.code={\pgfkeysalso{@cell content=\cellcolor{red!100!white}##1}},
				},
			every row 2 column 2/.style={
					postproc cell content/.code={\pgfkeysalso{@cell content=\cellcolor{red!8!white}##1}},
				},
			every row 2 column 3/.style={
					postproc cell content/.code={\pgfkeysalso{@cell content=\cellcolor{red!0!white}##1}},
				},
			every row 2 column 4/.style={
					postproc cell content/.code={\pgfkeysalso{@cell content=\cellcolor{red!0!white}##1}},
				},
			every row 2 column 5/.style={
					postproc cell content/.code={\pgfkeysalso{@cell content=\cellcolor{red!0!white}##1}},
				},
			every row 2 column 6/.style={
					postproc cell content/.code={\pgfkeysalso{@cell content=\cellcolor{red!0!white}##1}},
				},
			every row 2 column 7/.style={
					postproc cell content/.code={\pgfkeysalso{@cell content=\cellcolor{red!0!white}##1}},
				},
			every row 3 column 1/.style={
					postproc cell content/.code={\pgfkeysalso{@cell content=\cellcolor{red!0!white}##1}},
				},
			every row 3 column 2/.style={
					postproc cell content/.code={\pgfkeysalso{@cell content=\cellcolor{red!2!white}##1}},
				},
			every row 3 column 3/.style={
					postproc cell content/.code={\pgfkeysalso{@cell content=\cellcolor{red!0!white}##1}},
				},
			every row 3 column 4/.style={
					postproc cell content/.code={\pgfkeysalso{@cell content=\cellcolor{red!5!white}##1}},
				},
			every row 3 column 5/.style={
					postproc cell content/.code={\pgfkeysalso{@cell content=\cellcolor{red!0!white}##1}},
				},
			every row 3 column 6/.style={
					postproc cell content/.code={\pgfkeysalso{@cell content=\cellcolor{red!0!white}##1}},
				},
			every row 3 column 7/.style={
					postproc cell content/.code={\pgfkeysalso{@cell content=\cellcolor{red!0!white}##1}},
				},
			every row 4 column 1/.style={
					postproc cell content/.code={\pgfkeysalso{@cell content=\cellcolor{red!0!white}##1}},
				},
			every row 4 column 2/.style={
					postproc cell content/.code={\pgfkeysalso{@cell content=\cellcolor{red!0!white}##1}},
				},
			every row 4 column 3/.style={
					postproc cell content/.code={\pgfkeysalso{@cell content=\cellcolor{red!0!white}##1}},
				},
			every row 4 column 4/.style={
					postproc cell content/.code={\pgfkeysalso{@cell content=\cellcolor{red!5!white}##1}},
				},
			every row 4 column 5/.style={
					postproc cell content/.code={\pgfkeysalso{@cell content=\cellcolor{red!0!white}##1}},
				},
			every row 4 column 6/.style={
					postproc cell content/.code={\pgfkeysalso{@cell content=\cellcolor{red!0!white}##1}},
				},
			every row 4 column 7/.style={
					postproc cell content/.code={\pgfkeysalso{@cell content=\cellcolor{red!0!white}##1}},
				},
			every row 5 column 1/.style={
					postproc cell content/.code={\pgfkeysalso{@cell content=\cellcolor{red!0!white}##1}},
				},
			every row 5 column 2/.style={
					postproc cell content/.code={\pgfkeysalso{@cell content=\cellcolor{red!13!white}##1}},
				},
			every row 5 column 3/.style={
					postproc cell content/.code={\pgfkeysalso{@cell content=\cellcolor{red!0!white}##1}},
				},
			every row 5 column 4/.style={
					postproc cell content/.code={\pgfkeysalso{@cell content=\cellcolor{red!21!white}##1}},
				},
			every row 5 column 5/.style={
					postproc cell content/.code={\pgfkeysalso{@cell content=\cellcolor{red!0!white}##1}},
				},
			every row 5 column 6/.style={
					postproc cell content/.code={\pgfkeysalso{@cell content=\cellcolor{red!57!white}##1}},
				},
			every row 5 column 7/.style={
					postproc cell content/.code={\pgfkeysalso{@cell content=\cellcolor{red!0!white}##1}},
				},
			every row 6 column 1/.style={
					postproc cell content/.code={\pgfkeysalso{@cell content=\cellcolor{red!0!white}##1}},
				},
			every row 6 column 2/.style={
					postproc cell content/.code={\pgfkeysalso{@cell content=\cellcolor{red!6!white}##1}},
				},
			every row 6 column 3/.style={
					postproc cell content/.code={\pgfkeysalso{@cell content=\cellcolor{red!0!white}##1}},
				},
			every row 6 column 4/.style={
					postproc cell content/.code={\pgfkeysalso{@cell content=\cellcolor{red!0!white}##1}},
				},
			every row 6 column 5/.style={
					postproc cell content/.code={\pgfkeysalso{@cell content=\cellcolor{red!0!white}##1}},
				},
			every row 6 column 6/.style={
					postproc cell content/.code={\pgfkeysalso{@cell content=\cellcolor{red!0!white}##1}},
				},
			every row 6 column 7/.style={
					postproc cell content/.code={\pgfkeysalso{@cell content=\cellcolor{red!0!white}##1}},
				},
			every row 7 column 1/.style={
					postproc cell content/.code={\pgfkeysalso{@cell content=\cellcolor{red!0!white}##1}},
				},
			every row 7 column 2/.style={
					postproc cell content/.code={\pgfkeysalso{@cell content=\cellcolor{red!2!white}##1}},
				},
			every row 7 column 3/.style={
					postproc cell content/.code={\pgfkeysalso{@cell content=\cellcolor{red!42!white}##1}},
				},
			every row 7 column 4/.style={
					postproc cell content/.code={\pgfkeysalso{@cell content=\cellcolor{red!47!white}##1}},
				},
			every row 7 column 5/.style={
					postproc cell content/.code={\pgfkeysalso{@cell content=\cellcolor{red!0!white}##1}},
				},
			every row 7 column 6/.style={
					postproc cell content/.code={\pgfkeysalso{@cell content=\cellcolor{red!14!white}##1}},
				},
			every row 7 column 7/.style={
					postproc cell content/.code={\pgfkeysalso{@cell content=\cellcolor{red!0!white}##1}},
				},
			every row 8 column 1/.style={
					postproc cell content/.code={\pgfkeysalso{@cell content=\cellcolor{red!0!white}##1}},
				},
			every row 8 column 2/.style={
					postproc cell content/.code={\pgfkeysalso{@cell content=\cellcolor{red!0!white}##1}},
				},
			every row 8 column 3/.style={
					postproc cell content/.code={\pgfkeysalso{@cell content=\cellcolor{red!0!white}##1}},
				},
			every row 8 column 4/.style={
					postproc cell content/.code={\pgfkeysalso{@cell content=\cellcolor{red!0!white}##1}},
				},
			every row 8 column 5/.style={
					postproc cell content/.code={\pgfkeysalso{@cell content=\cellcolor{red!0!white}##1}},
				},
			every row 8 column 6/.style={
					postproc cell content/.code={\pgfkeysalso{@cell content=\cellcolor{red!0!white}##1}},
				},
			every row 8 column 7/.style={
					postproc cell content/.code={\pgfkeysalso{@cell content=\cellcolor{red!0!white}##1}},
				},
			every row 9 column 1/.style={
					postproc cell content/.code={\pgfkeysalso{@cell content=\cellcolor{red!0!white}##1}},
				},
			every row 9 column 2/.style={
					postproc cell content/.code={\pgfkeysalso{@cell content=\cellcolor{red!100!white}##1}},
				},
			every row 9 column 3/.style={
					postproc cell content/.code={\pgfkeysalso{@cell content=\cellcolor{red!100!white}##1}},
				},
			every row 9 column 4/.style={
					postproc cell content/.code={\pgfkeysalso{@cell content=\cellcolor{red!100!white}##1}},
				},
			every row 9 column 5/.style={
					postproc cell content/.code={\pgfkeysalso{@cell content=\cellcolor{red!100!white}##1}},
				},
			every row 9 column 6/.style={
					postproc cell content/.code={\pgfkeysalso{@cell content=\cellcolor{red!100!white}##1}},
				},
			every row 9 column 7/.style={
					postproc cell content/.code={\pgfkeysalso{@cell content=\cellcolor{red!0!white}##1}},
				},
			every row 10 column 1/.style={
					postproc cell content/.code={\pgfkeysalso{@cell content=\cellcolor{red!10!white}##1}},
				},
			every row 10 column 2/.style={
					postproc cell content/.code={\pgfkeysalso{@cell content=\cellcolor{red!24!white}##1}},
				},
			every row 10 column 3/.style={
					postproc cell content/.code={\pgfkeysalso{@cell content=\cellcolor{red!0!white}##1}},
				},
			every row 10 column 4/.style={
					postproc cell content/.code={\pgfkeysalso{@cell content=\cellcolor{red!26!white}##1}},
				},
			every row 10 column 5/.style={
					postproc cell content/.code={\pgfkeysalso{@cell content=\cellcolor{red!0!white}##1}},
				},
			every row 10 column 6/.style={
					postproc cell content/.code={\pgfkeysalso{@cell content=\cellcolor{red!0!white}##1}},
				},
			every row 10 column 7/.style={
					postproc cell content/.code={\pgfkeysalso{@cell content=\cellcolor{red!0!white}##1}},
				},
			every row 11 column 1/.style={
					postproc cell content/.code={\pgfkeysalso{@cell content=\cellcolor{red!20!white}##1}},
				},
			every row 11 column 2/.style={
					postproc cell content/.code={\pgfkeysalso{@cell content=\cellcolor{red!24!white}##1}},
				},
			every row 11 column 3/.style={
					postproc cell content/.code={\pgfkeysalso{@cell content=\cellcolor{red!0!white}##1}},
				},
			every row 11 column 4/.style={
					postproc cell content/.code={\pgfkeysalso{@cell content=\cellcolor{red!26!white}##1}},
				},
			every row 11 column 5/.style={
					postproc cell content/.code={\pgfkeysalso{@cell content=\cellcolor{red!0!white}##1}},
				},
			every row 11 column 6/.style={
					postproc cell content/.code={\pgfkeysalso{@cell content=\cellcolor{red!0!white}##1}},
				},
			every row 11 column 7/.style={
					postproc cell content/.code={\pgfkeysalso{@cell content=\cellcolor{red!0!white}##1}},
				},
			every row 12 column 1/.style={
					postproc cell content/.code={\pgfkeysalso{@cell content=\cellcolor{red!0!white}##1}},
				},
			every row 12 column 2/.style={
					postproc cell content/.code={\pgfkeysalso{@cell content=\cellcolor{red!5!white}##1}},
				},
			every row 12 column 3/.style={
					postproc cell content/.code={\pgfkeysalso{@cell content=\cellcolor{red!0!white}##1}},
				},
			every row 12 column 4/.style={
					postproc cell content/.code={\pgfkeysalso{@cell content=\cellcolor{red!5!white}##1}},
				},
			every row 12 column 5/.style={
					postproc cell content/.code={\pgfkeysalso{@cell content=\cellcolor{red!0!white}##1}},
				},
			every row 12 column 6/.style={
					postproc cell content/.code={\pgfkeysalso{@cell content=\cellcolor{red!0!white}##1}},
				},
			every row 12 column 7/.style={
					postproc cell content/.code={\pgfkeysalso{@cell content=\cellcolor{red!0!white}##1}},
				},
			every row 13 column 1/.style={
					postproc cell content/.code={\pgfkeysalso{@cell content=\cellcolor{red!0!white}##1}},
				},
			every row 13 column 2/.style={
					postproc cell content/.code={\pgfkeysalso{@cell content=\cellcolor{red!5!white}##1}},
				},
			every row 13 column 3/.style={
					postproc cell content/.code={\pgfkeysalso{@cell content=\cellcolor{red!0!white}##1}},
				},
			every row 13 column 4/.style={
					postproc cell content/.code={\pgfkeysalso{@cell content=\cellcolor{red!0!white}##1}},
				},
			every row 13 column 5/.style={
					postproc cell content/.code={\pgfkeysalso{@cell content=\cellcolor{red!0!white}##1}},
				},
			every row 13 column 6/.style={
					postproc cell content/.code={\pgfkeysalso{@cell content=\cellcolor{red!0!white}##1}},
				},
			every row 13 column 7/.style={
					postproc cell content/.code={\pgfkeysalso{@cell content=\cellcolor{red!0!white}##1}},
				},
			every row 14 column 1/.style={
					postproc cell content/.code={\pgfkeysalso{@cell content=\cellcolor{red!0!white}##1}},
				},
			every row 14 column 2/.style={
					postproc cell content/.code={\pgfkeysalso{@cell content=\cellcolor{red!0!white}##1}},
				},
			every row 14 column 3/.style={
					postproc cell content/.code={\pgfkeysalso{@cell content=\cellcolor{red!0!white}##1}},
				},
			every row 14 column 4/.style={
					postproc cell content/.code={\pgfkeysalso{@cell content=\cellcolor{red!0!white}##1}},
				},
			every row 14 column 5/.style={
					postproc cell content/.code={\pgfkeysalso{@cell content=\cellcolor{red!0!white}##1}},
				},
			every row 14 column 6/.style={
					postproc cell content/.code={\pgfkeysalso{@cell content=\cellcolor{red!0!white}##1}},
				},
			every row 14 column 7/.style={
					postproc cell content/.code={\pgfkeysalso{@cell content=\cellcolor{red!0!white}##1}},
				},
		]
		{tables/recall_results_10_rwd_pessimistic_t.csv}
		\caption{\label{tab:recall_at_k:rwd}}% Rank at max recall}
	\end{subfigure}
	\hfill
	\begin{subfigure}[t]{0.22\textwidth}
		%%%% top-ranked non-FDs 
		\footnotesize
		\begin{filecontents*}{tables/top_ranked_non_fds.csv}
measure,lhsuniquenessdblp,lhsuniquenessgathAgent,rhsskewnessdblp,rhsskewnessgathAgent
$\rho$,0.88,,0.83,
$g_2$,0.88,,0.83,
$g_3$,0.88,0.04,1.24,4.42
$g'_3$,,0.04,,4.42
$\SHANNONGONE$,0.59,0.12,0.94,4.79
$\FI$,0.87,,0.64,
$\RFIplus$,0.45,0.03,1.11,4.74
$\SFI$,0.4,0.12,1.04,3.8
$g_1$,0.85,0.78,1.09,3.8
$g'_1$,0.85,0.78,1.09,3.8
$\pdep$,0.88,0.04,0.81,4.42
$\tau$,0.88,,0.81,
\end{filecontents*}
\pgfplotstabletypeset[
			col sep=comma,
			numeric type,
			fixed,
			column type={p{0.2cm}},
			columns/measure/.style={
					column name={},
					string type,
					column type={p{0.4cm}},
				},
			columns/lhsuniquenessdblp/.style={
					column name={\dblpDataNumber},
				},
			columns/lhsuniquenessgathAgent/.style={
					column name={\gathAgentDataNumber},
				},
			columns/rhsskewnessdblp/.style={
					column name={\dblpDataNumber},
				},
			columns/rhsskewnessgathAgent/.style={
					column name={\gathAgentDataNumber},
				},
			every head row/.style={
					before row={
							& \multicolumn{2}{c}{LHS-uniq.} & \multicolumn{2}{c}{\RHSskew} \\
						},
				},
			every last row/.style={
					after row={
							\!$\designAFD(R)$ & \multicolumn{2}{c}{0.07} & \multicolumn{2}{c}{1.43} \\
							rest & \multicolumn{2}{c}{0.24} & \multicolumn{2}{c}{1.59}
						},
				},
		]
		{tables/top_ranked_non_fds.csv}
		\caption{\label{tab:top_ranked_non_fds_dblp:rwd}}
	\end{subfigure}
	\vspace{-1.5ex}
	\caption{Measure performance on \rwdminus. (a) Heatmap of PR-AUC scores (red indicates low value). (b) Rank at max recall (heatmap per column where red indicates high rank)\rev{, indicating how many FD candidates need to be examined to retrieve $\designAFD(R)$}. (c) \LHSuniqueness and \RHSskew of mislabeled FD candidates.  }
	\label{fig:rwd-results}
\end{figure*}

\paragraph*{Efficient computation and \rwdminus.} We observe massive
differences in \rev{the time  required to compute a score for each \afd
candidate}. In particular, $\SFI$, $\RFIplus$ and $\RFInorm$ require an
unreasonable amount of computation time. In 24 hours using a single CPU core,
\SFI was able to calculate values for 1430 while \RFIplus and \RFInorm
processed 250 of all \TotalNumberOfInvalidFDs \afd candidates. All other
measures finished processing all candidates in roughly two minutes. For
practical applications, the runtime of \SFI, \RFIplus and \RFInorm is hence
problematic. \inFullVersion{An overview of all runtimes is included in
Appendix~\ref{sec:rwd-extra} in Table~\ref{tab:runtimes:rwd}.}

In fact, the computational complexity of \RFIplus and
\RFInorm did not allow us to compute values for \RFIplus and \RFInorm on all
candidate \fds in \rwd in a reasonable amount of time. In approximately 168 hours we obtained values of
\RFIplus and \RFInorm for 1229 candidate \afds, including
all design \afds, out of a total of \TotalNumberOfInvalidFDs. We denote this set of 1229 candidates
\fds by \rwdminus in what follows. To ensure fair comparison among all measures,
we report all comparison metrics (AUC, $\rankAtMaxRecall$, \dots) relative to
\rwdminus.
%%%%%%
\inFullVersion{Although the results are similar, for the sake of completeness, we include in Appendix~\ref{sec:rwd-extra}, the performance metrics relative to the whole \rwd excluding the measures \RFIplus and \RFInorm.}
\inConfVersion{The results on the whole of \rwd excluding the measures \RFIplus and \RFInorm are similar~\cite{artifacts}.}

\subsection{Results}
\paragraph*{AUC.} \
Figure~\ref{tab:auc} lists the AUC scores for \rwdminus at the benchmark and
relation level, where the AUC value is expressed as a percentage. The last
column shows the fraction of relations on which a measure reached maximal AUC
score, allowing us to judge how consistent a measure is.
At the benchmark level, we observe that there are effective measures in each
measure class. $\RFInorm$ (\shannon, AUC = 0.971) is the most effective
measure, closely followed by $\muplus$ (\logical, AUC = 0.946) and somewhat
further followed by $g'_3$ (\simple, AUC = 0.901). All other measures have
significantly lower AUC values. When the correct number of \afds is
not known beforehand and a specific threshold needs to be set uniformly for all
relations, $\RFInorm$, $\muplus$ and $g'_3$ hence provide the best tradeoff
between precision and recall. We find it striking to note that the unnormalized
variants of these measures (i.e., $\FI$, $\pdep$, and $g_3$, respectively)
perform significantly worse, which highlights the importance of normalisation
when designing measures. For $\RFInorm$ and $\muplus$ in particular, we note
that the normalisation w.r.t the expected value of $\FI$ resp. $\pdep$ under
random permutations performs significantly better than computing the absolute
difference w.r.t this absolute value ($\RFIplus$), respectively normalising
w.r.t. $\pdep(\Y)$ (for $\tau$).

The AUC scores at the relation level give a more detailed picture. In
particular, the last column in Figure~\ref{tab:auc} shows that
$\RFIplus$ yields the highest AUC score on each relation, while $\muplus$ does
so for $90\%$ of the relations, and $g'_3$ for ``only'' $80\%$ of the
relations. $\muplus$ performs worse than $\RFInorm$ only on
relation \gathAreaDataNumber% \ (gath. area)
, where its AUC score still outperforms
the other measures. $g'_3$ also performs worse than $\RFInorm$ on
\gathAreaDataNumber\ and additionally performs worse than both $\RFInorm$ and
$\muplus$ on \gathAgentDataNumber.

Surprisingly, $\FI$, which has a low AUC score $= 0.415$ at the benchmark level
has a highest AUC score on $90\%$ of the relations, like $\muplus$. It does
particularly poor on relation \dblpDataNumber\ (dblp, AUC=$5.4\%$), which
explains its AUC score on the benchmark level. Similarly to $g'_3$, $\tau$ has
a highest AUC on $80\%$ of the relations, but it also performs very poor on
$R_3$, explaning its lower AUC score at the benchmark level. We note that our
observation from Section~\ref{sec:eval-syn} on synthetic data, namely that
$g_1, g'_1$ and $\SFI$ have poor distinguishing power, is confirmed on
$\rwdminus$: these measures perform the poorest, attaining maximal AUC score in
only $60\%$ resp. $50\%$ of the relations. Measures $\SHANNONGONE$ and
$\RFIplus$ perform equally poor.

In summary, $\RFInorm$, $\muplus$ and $g'_3$ provide the best tradeoff between
precision and recall with $\RFInorm$ performing better than $\muplus$
(marginally, on one relation), and $\muplus$ performing better than $g'_3$
(again marginally, on one relation).

\begin{table*}[bht]
	\caption{\label{table:comparison} Properties of considered \afd measures. The symbol {\cmark} stands for \emph{applies}, \rev{the symbol {\xmark}} denotes \emph{does not apply}. The symbol \nap\ stands for \emph{not applicable} (cf.\ Section~\ref{sec:eval-syn}).  Further, $f'$ refers to the normalization of the measure $f$ as discussed in Section~\ref{sec:measures} and $f^+$ to the adaptation of $f$ that maps all negative values to zero.}
  \footnotesize
  \setlength\tabcolsep{0.8pt}
  \centering
  \begin{tabular}[t]{@{}lcccccccccccccc@{}}
    \toprule
             & $\rho$  & $g_2$ & $g_3$   & $g'_3$       & $\SHANNONGONE$ & $\FI$  & $\RFIplus$ & $\RFInorm$  & $\SFI$  & $g_1$ & $g'_1$     & $\pdep$ & $\tau$  & $\muplus$ \\\midrule
    \begin{filecontents*}{tables/comp_measures.csv}
coltwo;colrho;gtwo;gthree;gthreeprime;shgone;colFI;colRFI;colRFInorm;colSFI;gone;goneprime;colpdp;coltau;colmu
Considered in;{\cite{DBLP:conf/sigmod/IlyasMHBA04}};{\cite{DBLP:journals/tcs/KivinenM95,WangHe2019}};{\cite{DBLP:journals/pvldb/Berti-EquilleHN18,BERZAL2002207, DBLP:journals/cj/HuhtalaKPT99}};{\cite{DBLP:journals/is/GiannellaR04}};new;{\cite{DBLP:conf/vldb/CavalloP87,DBLP:journals/is/GiannellaR04}};{\cite{DBLP:conf/kdd/MandrosBV17,DBLP:journals/kais/MandrosBV20}};new;{\cite{DBLP:conf/kdd/PennerathMV20}};{\cite{DBLP:journals/tcs/KivinenM95,ZhangGuoEtAl2020}};{\cite{DBLP:journals/pvldb/0001N18}};{\cite{piatetsky1993}};{\cite{10.2307/2281536,piatetsky1993}};{\cite{piatetsky1993}}
;;;{\cite{DBLP:journals/jamds/KingL03,DBLP:conf/webdb/WangDSFH09,DBLP:journals/tcs/KivinenM95}};;;{\cite{DBLP:conf/kdd/MandrosBV17,DBLP:journals/kais/MandrosBV20,DBLP:conf/kdd/PennerathMV20}};;;;;;;;
Class ({\ssimple}/{\sshannon}/{\slogical});\ssimple;\ssimple;\ssimple;\ssimple;\sshannon;\sshannon;\sshannon;\sshannon;\sshannon;\slogical;\slogical;\slogical;\slogical;\slogical
Has baselines;;{\cmark};;{\cmark};{\cmark};{\cmark};{\cmark};{\cmark};{\cmark};;{\cmark};;{\cmark};{\cmark}
Efficiently computable;{\cmark};{\cmark};{\cmark};{\cmark};{\cmark};{\cmark};;;{\cmark};{\cmark};{\cmark};{\cmark};{\cmark};{\cmark}
Inversely proportional to error level;{\cmark};{\cmark};{\cmark};{\cmark};;{\cmark};{\cmark};{\cmark};;;;{\cmark};{\cmark};{\cmark}
Insensitive to LHS-unique;;;;{\cmark};;;;{\cmark};{\cmSFILHS};{\cmgLHS};{\cmgLHS};;;{\cmark}
Insensitive to RHS-skew;;;;;;{\cmark};{\cmark};{\cmark};{\cmSFIRHS};{\cmgRHS};{\cmgRHS};;{\cmark};{\cmark}
AUC on RWD;0.417;0.504;0.674;0.901;0.109;0.415;0.494;0.971;0.320;0.425;0.425;0.647;0.630;0.946
\end{filecontents*}
\csvreader[
      head to column names,
      /csv/separator=semicolon,
    late after line=                                                                                                                                                        \\,
    late after last line=                                                                                                                                                   \\\bottomrule,
    ]
    {tables/comp_measures.csv}{}%
    {\coltwo & \colrho & \gtwo & \gthree & \gthreeprime & \shgone        & \colFI & \colRFI    & \colRFInorm & \colSFI & \gone & \goneprime & \colpdp & \coltau & \colmu}   %
  \end{tabular}
\end{table*}

%%% Local Variables:
%%% mode: latex
%%% TeX-master: "../main/main"
%%% End:

\paragraph*{Rank at max recall} is shown in Figure~\ref{tab:recall_at_k:rwd}. The
first row indicates the total number of design \afds to discover (the smallest
attainable $\rankAtMaxRecall$ value).  We observe that the best measures,
$g'_3$, $\RFInorm$ and $\muplus$ have optimal $\rankAtMaxRecall$,
save on relation \gathAreaDataNumber\xspace where they differ by $1$ from the
optimum and still have minimal $\rankAtMaxRecall$ among all measures. In
addition, $g'_3$ has non-optimal $\rankAtMaxRecall$ on \gathAgentDataNumber,
where it is off by $1$. \rev{
	At maximum recall, these measures retain a precision of $100\%$,
    save on relation \gathAreaDataNumber\xspace ($66\%$), and $g'_3$ also on relation \gathAgentDataNumber\xspace ($66\%$).
    As such, these measures allow us only inspect a small
    number of highly ranked \AFDs to recover the true design \fds that were obscured by errors.}
By contrast, all the other measures have relations where the $\rankAtMaxRecall$ is an order of magnitude larger than the optimum, or more, yielding low precision at maximum recall.

\paragraph*{\LHSuniqueness and \RHSskew.}
% R3: \dblpDataNumber
% R6: \gathAgentDataNumber
% R7: \gathAreaDataNumber
From Figures~\ref{tab:auc} and~\ref{tab:recall_at_k:rwd} we observe that there
are four kinds of relations in \rwd: ``trivial'' relations for which every
measure attains optimal AUC and $\rankAtMaxRecall$ (relations \adultDataNumber,
\taxDataNumber,\identTaxonDataNumber,\identDataNumber), ``easy'' relations for
which nearly all measures do so (\claimsDataNumber\xspace and \gathDataNumber),
``challenging'' relations where only a minority of measures reach optimal
scores (\dblpDataNumber\xspace and \gathAgentDataNumber), and ``out-of-reach''
relations where no measure attains the optimum (\gathAreaDataNumber).

Next, we investigate what properties of the input data makes a relation
challenging by analyzing `mislabeled' candidate \fds in \dblpDataNumber\xspace
and \gathAgentDataNumber. We refer to a candidate as \emph{mislabeled}
analogous to our definition of $\rankAtMaxRecall$: from the candidates counted
for $\rankAtMaxRecall(f,R)$ we exclude all $\afd(R)$ to obtain our mislabeled
candidate \fds. In other words, the mislabeled candidate \fds are the highest
ranked mistakes made by a measure. Figure~\ref{tab:top_ranked_non_fds_dblp:rwd}
shows the average \LHSuniqueness and \RHSskew values of all mislabeled
candidate \fds per measure.
For comparison, the bottom two rows show the average \LHSuniqueness and
\RHSskew over the set of all design \afds and the set of all candidate \fds not
in the design set.

We start with analyzing \dblpDataNumber. From Figures~\ref{tab:auc} and
\ref{tab:recall_at_k:rwd} we recall that measures $g'_3$, $\RFInorm$, and $\muplus$
attain optimal AUC and $\rankAtMaxRecall$, while the AUC scores of all other measures
are extremely low and their $\rankAtMaxRecall$ is very high. In
Figure~\ref{tab:top_ranked_non_fds_dblp:rwd}, we observe that these other measures have
much higher \LHSuniqueness values for mislabeled candidate \fds than the
average for design \afds (0.07) or non-\fds (0.2). We postulate that this makes %is what makes
\dblpDataNumber\ challenging for these measures. %The group of measures $\rho$, $g_2$,
$\rho$, $g_2$,
$g_3$, $\FI$, $g_1$, $g'_1$, $\pdep$ and $\tau$ all have mislabeled \LHSuniqueness $>
0.8$ and we recall from Section~\ref{sec:eval-syn} that the distinguishing power of
these measures is small at high \LHSuniqueness values. In addition, from
Figure~\ref{fig:synexperiments} (middle row) we observe that $\RFIplus$ (0.45) and
$\SFI$ (0.4) have small separation (and hence limited distinguishing power) already at
modest values of \LHSuniqueness. For $\SHANNONGONE$, the situation is less clear. We
note, however that its \LHSuniqueness value (0.59) is much larger than the average for 
design \afds in \dblpDataNumber (0.074).

On \gathAgentDataNumber, $\rho$, $g_2$, $\FI$, $\RFInorm$, $\tau$ and $\muplus$
attain optimal AUC and $\rankAtMaxRecall$. In
Figure~\ref{tab:top_ranked_non_fds_dblp:rwd} we observe  high  \RHSskew values
($>$3.7) for all other measures, compared to the values of design \fds and
non-fds (both 0.4). We postulate this is what makes \gathAgentDataNumber
challenging for $g_3$, $g'_3$, $\SHANNONGONE$, and $\pdep$: recall from 
Section~\ref{sec:eval-syn} that these measures are sensitive to \RHSskew. In
contrast, we know from Section~\ref{sec:eval-syn} that $\SFI$, $g_1$ and $g'_1$
consistently have (almost) zero separation, independent of \RHSskew. 
Similarly, $\RFIplus$ is insensitive to \RHSskew, but its separation is
limited, as shown in Figure~\ref{fig:synexperiments} (middle row).% , which indicates that 

We conclude that high \LHSuniqueness and \RHSskew values are observable in
practice and sensitivity to these structural properties may
explain the lower performance of some measures.  Insensitivity to \LHSuniqueness
and \RHSskew are  therefore desirable properties to aim for when
designing measures.

%%% Local Variables:
%%% mode: latex
%%% TeX-master: "../main/main"
%%% End:

%%% Local Variables:
%%% mode: latex
%%% TeX-master: "../main/main"
%%% End:

\section{Conclusions}
\label{sec:conclusion}

An overview of our results is given in Table~\ref{table:comparison}. We find that well-ranking measures exist within each class: $g'_3$ in \simple, \RFInorm in \shannon, and $\muplus$ in \logical.
We further observe that measures are only effective when correctly normalized---which is not always done in the literature.
Indeed,
$g_3$ is widely known and cited~\cite{DBLP:journals/tcs/KivinenM95,DBLP:journals/pvldb/Berti-EquilleHN18, DBLP:journals/cj/HuhtalaKPT99,DBLP:journals/is/GiannellaR04,DBLP:journals/jamds/KingL03,DBLP:conf/icde/Faure-Giovagnoli22,BERZAL2002207}
but to the best of our knowledge only \cite{DBLP:journals/is/GiannellaR04} considers the correctly normalized version $g'_3$. The sensitivity of $g'_3$ to \RHSskew remains a structural weakness,  hampering its effectiveness
as illustrated by its lower performance in practice on \gathAgentDataNumber (Fig~\ref{fig:rwd-results}).

\FI is the defining measure of \shannon and suffers from sensitivity to \LHSuniqueness as illustrated by its behavior on \dblpDataNumber. Its corrections
\SFI~\cite{DBLP:conf/kdd/PennerathMV20} and \RFI~\cite{DBLP:conf/kdd/MandrosBV17,DBLP:journals/kais/MandrosBV20} were aimed at removing bias from \FI, but our experiments on \syn reveal that their distinguishing power is greatly diminished; they especially overcompensate their correction of \FI w.r.t.\ \LHSuniqueness. This is reflected by their behavior on \rwd, where they are  among the worst performing measures. Our novel correction \RFInorm of \RFI is the best performing measure on \rwd and is insensitive to both \LHSuniqueness and \RHSskew. Its main drawback is the slow computation by current algorithms rendering it essentially useless in practice.

Our recommendation for most suitable \afd measure is therefore the little-known measure $\mu$. It has comparable performance to \RFInorm as well as equal structural sensitivity properties, but can be efficiently computed.

When the number of attributes in the LHS increases, LHS-unique\-ness is expected to increase naturally to $1$. Since $g'_3$, \RFInorm, and \muplus are the only measures among the considered ones that are insensitive to \LHSuniqueness, they are the most promising  measures for discovering non-linear \afds.

Another finding worth noting is that we illustrated on \rwd, perhaps contrary to popular belief, that by only inspecting a small number of top-ranked candidate \fds (according to $g'_3$, $\RFInorm$, $\mu$), one already succeeds in finding a large number of true design FDs that were obscured by errors. This means in particular, that a domain expert does not need to wade through hundreds of high-ranked candidate \fds but can restrict attention to a handful.

\inFullVersion{Inspired by Arocena et al.~\cite{DBLP:journals/pvldb/ArocenaGMMPS15},
we study in Appendix~\ref{sec:nrwd_data} the measures' sensitivity to
different kinds and different levels of errors by passing the relations in $\rwd$
through a controlled error channel. We again find that $\RFInorm$ and $\muplus$
perform best, closely followed by $g'_3$. We also find that the effectiveness
of all measures quickly deteriorates for increasing error levels making them
essentially useless for error levels above $5\%$.}

%%% Local Variables:
%%% mode: latex
%%% TeX-master: "../main/main"
%%% End:

% ------------------------------------------------------------------
% Bibliography
\bibliographystyle{ieeetr}
\interlinepenalty=10000 % avoid splitting hyperlinks
\bibliography{biblio}

\inFullVersion{
  \appendices
  
  \section{Discussion on logical vs Shannon entropy}
  \label{sec:logical:vs:Shannon_entropy}
  The notion of logical entropy arises in mathematical philosophy~\cite{logical-entropy}, where it is observed to provide a theory of information based on logic. Importantly, formulas and equalities concerning logical entropy can be converted into corresponding formulas and equalities concerning Shannon entropy by the so-called dit-bit transform (see ~\cite{logical-entropy}). % For example, the following equalities are known
  Logical and Shannon entropy are hence highly similar, but measure different things: logical entropy measures the probability of two random tuples to be distinguished, while Shannon entropy measures average  uncertainty.
  
  \section{Proof of well-definedness of $\mu$}
  \label{sec:mu-welldef}
  In this section we prove the following result mentioned in Footnote~\pageref{footnote:muwelldef}.

\begin{lemma}
  If\/ $\expect_R[\pdep(\varphi,R)] = 1$ then $R \models \varphi$.
\end{lemma}
\begin{proof}
  Assume that $\expect_R[\pdep(\varphi,R)] = 1$. Let $R_1, \dots, R_n$ be an enumeration of all permuations of $R$. Then
  \begin{align*}
    \expect_R[\pdep(\varphi,R)]  = \frac{\sum_{i=1}^N \pdep(\varphi,R_i)}{N}
  \end{align*}
  Hence, $\expect_R[\pdep(\varphi,R)] = 1$ iff $\sum_{i=1}^N \pdep(\varphi,R_i) = N$. Because the range of $\pdep$ is the interval $[0,1]$ this sum can equal $N$ if, and only if, $\pdep(\varphi,R_i) = 1$ for every $R_i$, including $R$ itself. Suppose, for the purpose of contradiction, that $R \not \models \varphi$. Then, the value of $\pdep$ is given by the formula in Section~\ref{sec:prob-dependency}, i.e.
  \begin{align*}
    \pdep(\X \to \Y, R) & = \sum_{\x} \prob_R(\x) [1 - \pdep(\Y \mid \x, R) \\
                        & = 1 - \expect_{\x}[\logentrop_R(\Y \mid \x)]
  \end{align*}
  where the second equality is due to Lemma~\ref{lem:pdep-equiv}. Since $\pdep(\X \to \Y, R) = 1$, this means in particular that $\expect_{\x}[\logentrop_R(\Y \mid \x)] = 0$, which by reasoning similar as above can only happen if $\logentrop_R(\Y \mid \x) = 0$ for every $\x \in \pi_{\X}(R)$. This means, that for every  $\x \in \pi_{\X}(R)$, the probability to draw two distinct $\Y$-tuples in $\pi_{\Y}(\sigma_{\X=\x}(R))$ is zero. But that can only happen if there is only one $\Y$-value $\pi_{\Y}(\sigma_{\X=\x}(R))$, in which case $R \models \varphi$ and we obtain our desired contradiction.
\end{proof}

%%% Local Variables:
%%% mode: latex
%%% TeX-master: "../../main/main"
%%% End:

  \section{Formal comparison of the AFD measures}
  \label{sec:formal:comparison}
  \begin{table}[t]
    \caption{\label{tab:similarities} Overview of similarities between \logical measures and measures in \simple / \shannon. }
    \scriptsize
    %more space between rows
    \renewcommand{\arraystretch}{1.4}
	\setlength\tabcolsep{0.5pt}
    \centering
    \begin{tabular}{@{}l@{\hspace{1ex}}r@{\,}l@{\quad}r@{\,}l@{}}
        \toprule
                                                                                                   &
        \multicolumn{2}{@{}l}{\logical measure}                                                    &
        \multicolumn{2}{@{}l}{\simple/\shannon}                                                                                                                                \\
        \midrule
        \cstep\label{c1}                                                                           &
        $g_1$                                                                                      & $= 1 - \logentrop_R(\Y \mid \X)$                                        &
                                                                                                   & $ 1 - \entrop_R(\Y\mid \X)$                                               \\
        \midrule
        \cstep\label{c2}                                                                           &
        $\pdep$                                                                                    & $= \sum_{\x} \prob_R(\x) \left(1 - \logentrop_R(\Y \mid \x)\right)$     &
        $g_3$                                                                                      & = $\sum_{\x} \prob_R(\x) \max_{\y} \prob_R(\y \mid \x)$
        \\
                                                                                                   &
                                                                                                   & $= 1 - \sum_{\x} \prob_R(\x) \logentrop_R(\Y \mid \x)$                  &
        $g_2$                                                                                      & $= 1 - \sum_{\w \in G_2(\X \to \Y, R)} \prob_R(\w)$                       \\
        \midrule
        \cstep\label{c3}                                                                           &
        $\tau$                                                                                     & $= 1 - \frac{\expect_{\x}[\logentrop_R(\Y \mid \x)]}{\logentrop_R(\Y)}$ &
        $\FI$                                                                                      & $= 1 - \frac{\entrop_R(\Y \mid \X)}{\entrop_R(\Y)}$
        \\
                                                                                                   &
                                                                                                   &                                                                         &
                                                                                                   & $= 1 - \frac{\expect_{\x}[\entrop_R(\Y \mid \x)]}{\entrop_R(\Y)}$         \\
        \midrule
        \cstep\label{c4}                                                                           &
        $\mu$                                                                                      &
        $ = \frac{\pdep(\varphi, R)-\expect_R[\pdep(\varphi, R)]}{1-\expect_R[\pdep(\varphi, R)]}$ &
        $\RFI$                                                                                     & $= \FI(\varphi, R) - \expect_R[\FI(\varphi,R)]$
        \\
        \bottomrule
    \end{tabular}
\end{table}

This section presents a formal comparison of the measures introduced in Section~\ref{sec:measures}, % .  We categorize them into three classes and note that there are significant similarities in their design which
which motivates the definition of the new measures \SHANNONGONE and \RFInorm.

Among the measures introduced in Section~\ref{sec:measures}, we discern the following three classes (see also the second row in Table~\ref{table:comparison}):
\begin{compactenum}[(1)]
    \item The class of measures that have a notion of ``violation'' and quantify the number of violations, consisting of $\rho, g_2, g_3$, and $g'_3$. We denote this class by $\simple$.
    \item The class of measures based on Shannon entropy, consisting of $\FI, \RFIplus$, and $\SFI$. We denote this class by $\shannon$.

    \item The class of measures based on logical entropy, consisting of $g_1,g'_1, \pdep, \tau$, and $\muplus$ and denoted by $\logical$.
\end{compactenum}
We discuss the similarities in the design of \logical
measures and those in the \simple and \shannon class
by means of Table~\ref{tab:similarities}, which clusters
measures into groups that we find similar and where we rewrite measures
into equivalent form when this is necessary to stress the similarities.
\inConfVersion{We prove in the full version of this paper~\cite{artifacts}:}
\begin{theorem}
    \label{thm:equivalenes-correct}
    % Assuming $R \not \models \varphi$, 
    The alternate formulas given in
    Table~\ref{tab:similarities} are equivalent to their definition given in
    Sections \ref{sec:cords-measure}--\ref{sec:prob-dependency}.
\end{theorem}
\inFullVersion{The proof may be found in the Appendix~\ref{sec:equiv-proofs}.}

\newcommand{\circref}[1]{\CircledText[outer color=blue]{\small{\ref{#1}}}}

\circref{c1} We have already observed that $g_1$ is a measure based on
logical entropy, $g_1(\X \to \Y,R) = 1 - \logentrop_R(\Y \mid \X)$.
We find it interesting to observe that Giannella and
Robertson~\cite{DBLP:journals/is/GiannellaR04} considered an axiomatisation of
\fd error measures, and showed that Shannon entropy $\entrop_R(\Y \mid \X)$ is,
up to a multiplicative constant, the unique unnormalized error measure that
satisfies their axioms. As such, we may view $1 - \entrop_R(\Y \mid \X)$ as the
Shannon equivalent of $g_1$. Unfortunately, however, $1 - \entrop_R(\Y \mid \X)$
is not an \afd measure: the value of $\entrop_R(\Y \mid \X)$ is unbounded and
$1 - \entrop_R(\Y \mid \X)$ hence has range $[-\infty, 1]$ instead of $[0,1]$.
Giannella and Robertson~\cite{DBLP:journals/is/GiannellaR04} therefore turn
$1 - \entrop_R(\Y \mid \X)$ into an \afd measure by moving to $\FI$, which
normalizes $\entrop_R(\Y \mid \X)$ w.r.t. $\entrop_R(\Y)$. This is no longer the
conceptual Shannon counterpart of $g_1$. However, as further discussed below, it is nevertheless natural to ask what the conceptual Shannon counterpart of $g_1$ is and how it behaves. We thus propose the following Shannon variant $\SHANNONGONE$ of $g_1$, obtained by limiting $1  - \entrop_R(\Y\mid\X)$ to be positive:
\[ \SHANNONGONE(\X \to \Y, R) := \max(1 - \entrop_R(\Y \mid \X), 0). \]

\circref{c2} We have already observed in Section~\ref{sec:prob-dependency} that we may view $\pdep$ as a generalisation of $g_2$. We may also view it as an alternate to $g_3$. Indeed, $\pdep$ equals the expected value of $1 - \logentrop_R(\Y \mid \x)$---expressing the probability of $\x$ not participating in a violating pair---while $g_3$ equals the expected value of $\max_{\y} \prob_R(\y \mid \x)$---expressing the largest subgroup of non-violating tuples in $\pi_{\Y}\sigma_{\X=\x}(R)$, where in both cases expectation is taken over all $\x$.

\circref{c3} As shown by the rewritten formulas in line 3 of Table~\ref{tab:similarities}, $\FI$ is simply the Shannon entropy-based version of $\tau$.

\circref{c4}
The similarity between $\tau$ and $\FI$ extends to a conceptual similarity between $\mu$ and $\RFI$: $\mu$ corrects for the  bias of $\tau$ under random permutations while $\RFI$ corrects for the bias of $\FI$ under random permutations. Despite this conceptual similarity, note that the corrections are done differently: $\mu$ corrects by taking the \emph{normalized} difference between $\pdep$ and $\expect_R[\pdep]$ while $\RFI$ corrects by taking the \emph{absolute} difference between $\FI$ and $\expect_R[\FI]$.  As such $\RFI$ is not a normalized measure. Since it is natural to ask what the normalized variant of $\RFI$ is and how it behaves, we define
\[ \RFInorm(\varphi, R) := \max\left(\frac{\FI(\varphi,R) - \expect_R[\FI(\varphi,R)]}{1 -\expect_R[\FI(\varphi,R)]}, 0\right). \]

We attribute both new measures $\SHANNONGONE$ and $\RFInorm$ to $\shannon$, 
and compare their behavior to that of the other measures in the following sections.

%%% Local Variables:
%%% mode: latex
%%% TeX-master: "../main/main"
%%% End:

  \section{Proofs of equivalence}
  \label{sec:equiv-proofs}
  In this section we prove that the alternate measure formulations shown in
Table~\ref{tab:similarities} are correct, hence proving
Theorem~\ref{thm:equivalences-correct}. The theorem is proved as a sequence of lemmas. Throughout this section, assume that $R \not \models \X \to \Y$.

\begin{lemma}
  \label{lem:g3-equiv}
$g_3(\X \to \Y, R) = \sum_{\x} \prob_R(\x) \max_{\y} \prob_R(\y \mid \x)$.
\end{lemma}
\begin{proof}
We reason as follows.
\begin{align*}
  g_3(\X \to \Y, R)
  & = \max_{R' \in G_3(\X \to \Y, R)} \frac{|R'|}{|R|}          \\
  & = \max_{R' \in G_3(\X \to \Y,R)} \sum_{\w \in R'} \prob_R(\w) \\
  & = \sum_{\x} \max_{\y} \prob_R(\x\y)                        \\
  & = \sum_{\x} \prob_R(\x) \max_{\y} \prob_R(\y \mid \x).
\end{align*}
Here, the first equality is the definition of $g_3$.  The second equality follows by definition of $\prob_R$. The third equality follows from the following observation: a relation $R'\subseteq R$ can only be maximal if $R'(\w)=R(\w)$ whenever $R'(\w)>0$ for all $\w\in R$. That is, either we keep all occurrences of $\w$ or we remove all of them.  So, maximizing $\sum_{\w \in R'} \prob_R(\w)$ corresponds to, for every $\x$, keeping that $\y$ that maximizes $\prob_R(\x\y)$. Thereby, effectively removing all other tuples $\x\y'$ with $\y\neq\y'$. The last equality then follows from the definition of conditional probability. 
\end{proof}

\begin{lemma}
  \label{lem:pdep-equiv}
  $\pdep(\Y \mid \x, R) = 1 - \logentrop_R(\Y \mid \x)$ and therefore
  \begin{align*}
  \pdep(\X \to \Y, R) & = \sum_{\x} \prob_R(\x) (1 - \logentrop_R(\Y \mid x)) \\                
  & = 1 - \sum_{\x} \prob_R(\x) \logentrop_R(\Y \mid \x)\\
  & = 1 - \expect_{\x}[\logentrop_R(\Y \mid \x)]
  \end{align*}
\end{lemma}
\begin{proof}
  We first observe
  \begin{align*}
    \pdep(\Y \mid \x, R)
     & = \sum_{\y} \prob_R(\y \mid \x)^2           \\
     & = 1 - (1 - \sum_{\y} \prob_R(\y \mid \x)^2) \\
     & = 1 - \logentrop_R(Y \mid \x).
  \end{align*}
  Hence,
  \begin{align*}
    \pdep(\X \to \Y, R) & = \sum_{\x} \prob_R(\x) \pdep(\Y \mid \x, R)                            \\
                        & = \sum_{\x} \prob_R(\x) (1 - \logentrop_R(\Y \mid x))                   \\
                        & = \sum_{\x} \prob_R(\x) - \sum_{\x} \prob_R(\x)\logentrop_R(\Y \mid \x) \\
                        & = 1 - \sum_{\x} \prob_R(\x)\logentrop_R(\Y \mid \x). \qedhere
  \end{align*}
  
\end{proof}

\begin{lemma}
  $\tau(X \to Y, R) = 1 - \frac{\expect_{\x}[\logentrop_R(\Y \mid \x)]}{\logentrop_R(\Y)}$
\end{lemma}
\begin{proof}
  We reason as follows.
  \begin{align*}
    \tau(X\to Y , R) & = \frac{\pdep(\X \to \Y, R)-\pdep(\Y, R)}{1-\pdep(\Y, R)}                                                 \\
                     & = \frac{(1 - \expect_{\x}[\logentrop_R(\Y \mid \x)]) - (1 - \logentrop_R(\Y))}{1 - (1 - \logentrop_R(\Y))} \\
                     & = \frac{\logentrop_R(\Y) - \expect_{\x}[\logentrop_R(\Y \mid \x)]}{\logentrop_R(\Y)}                      \\
                     & = 1 - \frac{\expect_{\x}[\logentrop_R(\Y \mid \x)]}{\logentrop_R(\Y)} \qedhere
  \end{align*}
  Here, the second equality is by Lemma~\ref{lem:pdep-equiv} and the fact that $\pdep(\Y, R) = 1 - \logentrop_R(\Y)$ by definition.
\end{proof}

Relating this to logical entropy we observe
\begin{lemma}
  $ \mu(\X \to \Y,R) = 1 - \frac{\expect_{\x}[\logentrop_R(\Y\mid\x)]}{\logentrop_R(\Y)} \frac{|R|-1}{|R|-|\dom(\X,R)|}$
\end{lemma}
\begin{proof}
  We reason as follows.
  \begin{align*}
  \mu(&\X \to \Y, R) \\
   & :=\frac{\pdep(\X \to \Y, R)-\expect_R[\pdep(\X \to \Y, R)]}{1-\expect_R[\pdep(\X \to \Y, R)]} \\
    & =1-\frac{1-\pdep(\X \to \Y, R)}{1-\pdep(\Y, R)} \frac{|R|-1}{|R|-|\dom(\X,R)|} \\
    & = 1 - \frac{1- (1 - \expect_{\x}[\logentrop_R(\Y \mid \x)])}{1-(1-\logentrop_R(\Y))} \frac{|R|-1}{|R|-|\dom(\X,R)|} \\
    & = 1 - \frac{\expect_{\x}[\logentrop_R(\Y \mid \x)]}{\logentrop_R(\Y)} \frac{|R|-1}{|R|-|\dom(\X,R)|}.\qedhere
\end{align*}

\end{proof}

\begin{lemma}
  \[ \FI(\X \to \Y, R) = 1 - \frac{\entrop_R(\Y \mid \X)}{\entrop_R(\Y)}.\]
\end{lemma}
\begin{proof}
  To show the claimed equality, we reason as follows. Recall that we implicitly
  assume throughout the paper that $R$ is non-empty. By definition
  \[ \FI(\X \to \Y, R) :=
  \begin{cases}
    1                                                           & \text{ if } |\dom_R(\Y)| = 1, \\
    \frac{\entrop_R(\Y) - \entrop_R(\Y \mid \X)}{\entrop_R(\Y)} & \text{otherwise.}
  \end{cases}
\]
 We now make a case analysis.
  \begin{itemize}
    \item If $|\dom_R(\Y)| = 1$ then $\entrop_R(\Y) = 0$. Moreover, if
          $\entrop_R(\Y) = 0$, also $\entrop_R(\Y \mid \X) = 0$. As such,
          \[ 1 - \frac{\entrop_R(\Y \mid \X)}{\entrop_R(\X)} = 1- \frac{0}{0} = 1 - 0 = 1 = \FI(\X \to \Y, R),\]
          as desired.
    \item If $|\dom_R(\Y)| > 1$ then
          \begin{align*}
            \FI(\X \to \Y, R)
             & = \frac{\entrop_R(Y) - \entrop_R(\Y \mid \X)}{\entrop_R(\Y)} \\
             & = 1 - \frac{\entrop_R(\Y\mid \X)}{\entrop(\Y)} \qedhere
          \end{align*}
  \end{itemize}

\end{proof}

%%% Local Variables:
%%% mode: latex
%%% TeX-master: "../../main/main"
%%% End:

  \section{Additional information on the sensitivity analysis}
  \label{sec:syn-extra}
  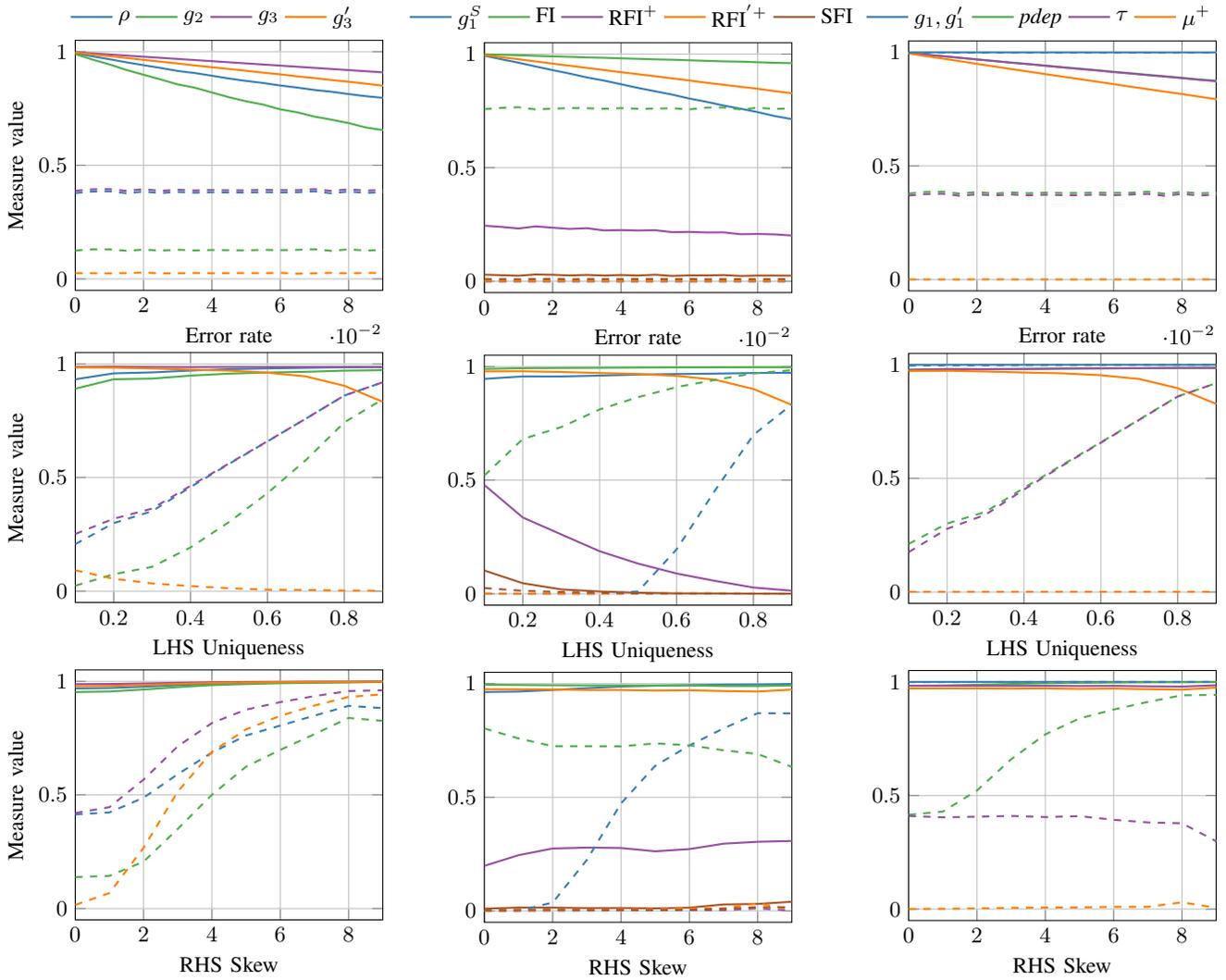
\begin{figure*}
	\small
	\centering
	\pgfplotsset{
		width=6cm,
		every axis plot/.append style={thick},
		every axis y label/.style={
				at={(ticklabel cs:0.5)},rotate=90,anchor=near ticklabel,
			},
		every axis x label/.style={
				at={(ticklabel cs:0.5)},anchor=near ticklabel,
			},
		ymin=-0.05, ymax=1.05,
		ymajorgrids=true,
		xmajorgrids=true,
		legend columns=-1,
		legend style={draw=none},
	}
	\begin{minipage}{0.32\textwidth}
		\centering
		\hspace{2em}\ref{appendix_syn_noise_simple_legend}
		\\
		\begin{tikzpicture}
			\begin{axis}[
					legend to name=appendix_syn_noise_simple_legend,
					xlabel={Error rate},
					ylabel={Measure value},
					xmin=-0.0, xmax=0.09,
				]
				\addplot [plotColor1] table [x=error,y=fd] {plotdata/syn_error_copy_rho.dat};
				\addlegendentry{$\rho$}
				\addplot [plotColor2] table [x=error,y=fd] {plotdata/syn_error_copy_g2.dat};
				\addlegendentry{$g_2$}
				\addplot [plotColor3] table [x=error,y=fd] {plotdata/syn_error_copy_g3.dat};
				\addlegendentry{$g_3$}
				\addplot [plotColor4] table [x=error,y=fd] {plotdata/syn_error_copy_g3_prime.dat};
				\addlegendentry{$g'_3$}
				\addplot [plotColor1,style=dashed] table [x=error,y=random] {plotdata/syn_error_copy_rho.dat};
				\addplot [plotColor2,style=dashed] table [x=error,y=random] {plotdata/syn_error_copy_g2.dat};
				\addplot [plotColor3,style=dashed] table [x=error,y=random] {plotdata/syn_error_copy_g3.dat};
				\addplot [plotColor4,style=dashed] table [x=error,y=random] {plotdata/syn_error_copy_g3_prime.dat};
			\end{axis}
		\end{tikzpicture}
		\begin{tikzpicture}
			\begin{axis}[
					legend to name=appendix_syn_keylike_simple_legend,
					xlabel={LHS Uniqueness},
					ylabel={Measure value},
					xmin=0.1, xmax=0.9,
				]
				\addplot [plotColor1] table [x=lhs_uniq,y=fd] {plotdata/syn_keylike_rho.dat};
				\addlegendentry{$\rho$}
				\addplot [plotColor2] table [x=lhs_uniq,y=fd] {plotdata/syn_keylike_g2.dat};
				\addlegendentry{$g_2$}
				\addplot [plotColor3] table [x=lhs_uniq,y=fd] {plotdata/syn_keylike_g3.dat};
				\addlegendentry{$g_3$}
				\addplot [plotColor4] table [x=lhs_uniq,y=fd] {plotdata/syn_keylike_g3_prime.dat};
				\addlegendentry{$g'_3$}
				\addplot [plotColor1,style=dashed] table [x=lhs_uniq,y=random] {plotdata/syn_keylike_rho.dat};
				\addplot [plotColor2,style=dashed] table [x=lhs_uniq,y=random] {plotdata/syn_keylike_g2.dat};
				\addplot [plotColor3,style=dashed] table [x=lhs_uniq,y=random] {plotdata/syn_keylike_g3.dat};
				\addplot [plotColor4,style=dashed] table [x=lhs_uniq,y=random] {plotdata/syn_keylike_g3_prime.dat};
			\end{axis}
		\end{tikzpicture}
		\begin{tikzpicture}
			\begin{axis}[
					legend to name=appendix_syn_rhsskew_simple_legend,
					xlabel={RHS Skew},
					ylabel={Measure value},
					xmin=0.0, xmax=9.0,
				]
				\addplot [plotColor1] table [x=rhs_skew,y=fd] {plotdata/syn_rhsskew_rho.dat};
				\addlegendentry{$\rho$}
				\addplot [plotColor2] table [x=rhs_skew,y=fd] {plotdata/syn_rhsskew_g2.dat};
				\addlegendentry{$g_2$}
				\addplot [plotColor3] table [x=rhs_skew,y=fd] {plotdata/syn_rhsskew_g3.dat};
				\addlegendentry{$g_3$}
				\addplot [plotColor4] table [x=rhs_skew,y=fd] {plotdata/syn_rhsskew_g3_prime.dat};
				\addlegendentry{$g'_3$}
				\addplot [plotColor1,style=dashed] table [x=rhs_skew,y=random] {plotdata/syn_rhsskew_rho.dat};
				\addplot [plotColor2,style=dashed] table [x=rhs_skew,y=random] {plotdata/syn_rhsskew_g2.dat};
				\addplot [plotColor3,style=dashed] table [x=rhs_skew,y=random] {plotdata/syn_rhsskew_g3.dat};
				\addplot [plotColor4,style=dashed] table [x=rhs_skew,y=random] {plotdata/syn_rhsskew_g3_prime.dat};
			\end{axis}
		\end{tikzpicture}
	\end{minipage}
	\begin{minipage}{0.34\textwidth}
		\centering
		%%%%% HACK: legative whitespace to shift the legend to the left
		\hspace{-0em}\ref{appendix_syn_noise_shannon_legend}
		\\
		\begin{tikzpicture}
			\begin{axis}[
					legend to name=appendix_syn_noise_shannon_legend,
					xlabel={Error rate},
					xmin=-0.0, xmax=0.09,
				]
				\addplot [plotColor1] table [x=error,y=fd] {plotdata/syn_error_copy_shannon_g1_prime.dat};
				\addlegendentry{$\SHANNONGONE$}
				\addplot [plotColor2] table [x=error,y=fd] {plotdata/syn_error_copy_fraction_of_information.dat};
				\addlegendentry{$\FI$}
				\addplot [plotColor3] table [x=error,y=fd] {plotdata/syn_error_copy_reliable_fraction_of_information_prime.dat};
				\addlegendentry{$\RFIplus$}
				\addplot [plotColor4] table [x=error,y=fd] {plotdata/syn_error_copy_fraction_of_information_prime.dat};
				\addlegendentry{$\RFInorm$}
				\addplot [plotColor5] table [x=error,y=fd] {plotdata/syn_error_copy_smoothed_fraction_of_information.dat};
				\addlegendentry{$\SFI$}
				\addplot [plotColor1,style=dashed] table [x=error,y=random] {plotdata/syn_error_copy_shannon_g1_prime.dat};
				\addplot [plotColor2,style=dashed] table [x=error,y=random] {plotdata/syn_error_copy_fraction_of_information.dat};
				\addplot [plotColor3,style=dashed] table [x=error,y=random] {plotdata/syn_error_copy_reliable_fraction_of_information_prime.dat};
				\addplot [plotColor4,style=dashed] table [x=error,y=random] {plotdata/syn_error_copy_fraction_of_information_prime.dat};
				\addplot [plotColor5,style=dashed] table [x=error,y=random] {plotdata/syn_error_copy_smoothed_fraction_of_information.dat};
			\end{axis}
		\end{tikzpicture}
		\begin{tikzpicture}
			\begin{axis}[
					legend to name=appendix_syn_keylike_shannon_legend,
					xlabel={LHS Uniqueness},
					xmin=0.1, xmax=0.9,
				]
				\addplot [plotColor1] table [x=lhs_uniq,y=fd] {plotdata/syn_keylike_shannon_g1_prime.dat};
				\addlegendentry{$\SHANNONGONE$}
				\addplot [plotColor2] table [x=lhs_uniq,y=fd] {plotdata/syn_keylike_fraction_of_information.dat};
				\addlegendentry{$\FI$}
				\addplot [plotColor3] table [x=lhs_uniq,y=fd] {plotdata/syn_keylike_reliable_fraction_of_information_prime.dat};
				\addlegendentry{$\RFIplus$}
				\addplot [plotColor4] table [x=lhs_uniq,y=fd] {plotdata/syn_keylike_fraction_of_information_prime.dat};
				\addlegendentry{$\RFInorm$}
				\addplot [plotColor5] table [x=lhs_uniq,y=fd] {plotdata/syn_keylike_smoothed_fraction_of_information.dat};
				\addlegendentry{$\SFI$}
				\addplot [plotColor1,style=dashed] table [x=lhs_uniq,y=random] {plotdata/syn_keylike_shannon_g1_prime.dat};
				\addplot [plotColor2,style=dashed] table [x=lhs_uniq,y=random] {plotdata/syn_keylike_fraction_of_information.dat};
				\addplot [plotColor3,style=dashed] table [x=lhs_uniq,y=random] {plotdata/syn_keylike_reliable_fraction_of_information_prime.dat};
				\addplot [plotColor4,style=dashed] table [x=lhs_uniq,y=random] {plotdata/syn_keylike_fraction_of_information_prime.dat};
				\addplot [plotColor5,style=dashed] table [x=lhs_uniq,y=random] {plotdata/syn_keylike_smoothed_fraction_of_information.dat};
			\end{axis}
		\end{tikzpicture}
		\begin{tikzpicture}
			\begin{axis}[
					legend to name=appendix_syn_rhsskew_shannon_legend,
					xlabel={RHS Skew},
					xmin=0.0, xmax=9.0,
				]
				\addplot [plotColor1] table [x=rhs_skew,y=fd] {plotdata/syn_rhsskew_shannon_g1_prime.dat};
				\addlegendentry{$\SHANNONGONE$}
				\addplot [plotColor2] table [x=rhs_skew,y=fd] {plotdata/syn_rhsskew_fraction_of_information.dat};
				\addlegendentry{$\FI$}
				\addplot [plotColor3] table [x=rhs_skew,y=fd] {plotdata/syn_rhsskew_reliable_fraction_of_information_prime.dat};
				\addlegendentry{$\RFIplus$}
				\addplot [plotColor4] table [x=rhs_skew,y=fd] {plotdata/syn_rhsskew_fraction_of_information_prime.dat};
				\addlegendentry{$\RFInorm$}
				\addplot [plotColor5] table [x=rhs_skew,y=fd] {plotdata/syn_rhsskew_smoothed_fraction_of_information.dat};
				\addlegendentry{$\SFI$}
				\addplot [plotColor1,style=dashed] table [x=rhs_skew,y=random] {plotdata/syn_rhsskew_shannon_g1_prime.dat};
				\addplot [plotColor2,style=dashed] table [x=rhs_skew,y=random] {plotdata/syn_rhsskew_fraction_of_information.dat};
				\addplot [plotColor3,style=dashed] table [x=rhs_skew,y=random] {plotdata/syn_rhsskew_reliable_fraction_of_information_prime.dat};
				\addplot [plotColor4,style=dashed] table [x=rhs_skew,y=random] {plotdata/syn_rhsskew_fraction_of_information_prime.dat};
				\addplot [plotColor5,style=dashed] table [x=rhs_skew,y=random] {plotdata/syn_rhsskew_smoothed_fraction_of_information.dat};
			\end{axis}
		\end{tikzpicture}
	\end{minipage}
	\begin{minipage}{0.32\textwidth}
		\centering
		\hspace{2em}\ref{appendix_syn_noise_logical_legend}
		\\
		\begin{tikzpicture}
			\begin{axis}[
					legend to name=appendix_syn_noise_logical_legend,
					xlabel={Error rate},
					xmin=-0.0, xmax=0.09,
				]
				%\addplot [plotColor1] table [x=error,y=fd] {plotdata/syn_error_copy_g1.dat};
				%\addlegendentry{$g_1$}
				\addplot [plotColor1] table [x=error,y=fd] {plotdata/syn_error_copy_g1_prime.dat};
				\addlegendentry{$g_1,g'_1$}
				\addplot [plotColor2] table [x=error,y=fd] {plotdata/syn_error_copy_pdep.dat};
				\addlegendentry{$\pdep$}
				\addplot [plotColor3] table [x=error,y=fd] {plotdata/syn_error_copy_tau.dat};
				\addlegendentry{$\tau$}
				\addplot [plotColor4] table [x=error,y=fd] {plotdata/syn_error_copy_mu_prime.dat};
				\addlegendentry{$\muplus$}
				% \addplot [plotColor1,style=dashed] table [x=error,y=random] {plotdata/syn_error_copy_g1.dat};
				\addplot [plotColor1,style=dashed] table [x=error,y=random] {plotdata/syn_error_copy_g1_prime.dat};
				\addplot [plotColor2,style=dashed] table [x=error,y=random] {plotdata/syn_error_copy_pdep.dat};
				\addplot [plotColor3,style=dashed] table [x=error,y=random] {plotdata/syn_error_copy_tau.dat};
				\addplot [plotColor4,style=dashed] table [x=error,y=random] {plotdata/syn_error_copy_mu_prime.dat};
			\end{axis}
		\end{tikzpicture}
		\begin{tikzpicture}
			\begin{axis}[
					legend to name=appendix_syn_keylike_logical_legend,
					xlabel={LHS Uniqueness},
					xmin=0.1, xmax=0.9,
				]
				% \addplot [plotColor1] table [x=lhs_uniq,y=fd] {plotdata/syn_keylike_g1.dat};
				% \addlegendentry{$g_1$}
				\addplot [plotColor1] table [x=lhs_uniq,y=fd] {plotdata/syn_keylike_g1_prime.dat};
				\addlegendentry{$g_1,g'_1$}
				\addplot [plotColor2] table [x=lhs_uniq,y=fd] {plotdata/syn_keylike_pdep.dat};
				\addlegendentry{$\pdep$}
				\addplot [plotColor3] table [x=lhs_uniq,y=fd] {plotdata/syn_keylike_tau.dat};
				\addlegendentry{$\tau$}
				\addplot [plotColor4] table [x=lhs_uniq,y=fd] {plotdata/syn_keylike_mu_prime.dat};
				\addlegendentry{$\muplus$}
				% \addplot [plotColor1,style=dashed] table [x=lhs_uniq,y=random] {plotdata/syn_keylike_g1.dat};
				\addplot [plotColor1,style=dashed] table [x=lhs_uniq,y=random] {plotdata/syn_keylike_g1_prime.dat};
				\addplot [plotColor2,style=dashed] table [x=lhs_uniq,y=random] {plotdata/syn_keylike_pdep.dat};
				\addplot [plotColor3,style=dashed] table [x=lhs_uniq,y=random] {plotdata/syn_keylike_tau.dat};
				\addplot [plotColor4,style=dashed] table [x=lhs_uniq,y=random] {plotdata/syn_keylike_mu_prime.dat};
			\end{axis}
		\end{tikzpicture}
		\begin{tikzpicture}
			\begin{axis}[
					legend to name=appendix_syn_rhsskew_logical_legend,
					xlabel={RHS Skew},
					xmin=0.0, xmax=9.0,
				]
				% \addplot [plotColor1] table [x=rhs_skew,y=fd] {plotdata/syn_rhsskew_g1.dat};
				% \addlegendentry{$g_1$}
				\addplot [plotColor1] table [x=rhs_skew,y=fd] {plotdata/syn_rhsskew_g1_prime.dat};
				\addlegendentry{$g_1,g'_1$}
				\addplot [plotColor2] table [x=rhs_skew,y=fd] {plotdata/syn_rhsskew_pdep.dat};
				\addlegendentry{$\pdep$}
				\addplot [plotColor3] table [x=rhs_skew,y=fd] {plotdata/syn_rhsskew_tau.dat};
				\addlegendentry{$\tau$}
				\addplot [plotColor4] table [x=rhs_skew,y=fd] {plotdata/syn_rhsskew_mu_prime.dat};
				\addlegendentry{$\muplus$}
				% \addplot [plotColor1,style=dashed] table [x=rhs_skew,y=random] {plotdata/syn_rhsskew_g1.dat};
				\addplot [plotColor1,style=dashed] table [x=rhs_skew,y=random] {plotdata/syn_rhsskew_g1_prime.dat};
				\addplot [plotColor2,style=dashed] table [x=rhs_skew,y=random] {plotdata/syn_rhsskew_pdep.dat};
				\addplot [plotColor3,style=dashed] table [x=rhs_skew,y=random] {plotdata/syn_rhsskew_tau.dat};
				\addplot [plotColor4,style=dashed] table [x=rhs_skew,y=random] {plotdata/syn_rhsskew_mu_prime.dat};
			\end{axis}
		\end{tikzpicture}
	\end{minipage}
	\caption{\label{fig:synexperiments-additional} Average measures values of the three experiments on \syn data. The top row shows $\synnoisefd$ (solid lines) and $\synnoisenonfd$ (dashed lines) for different error levels. Middle row shows $\synuniquefd$ (solid lines) and $\synuniquenonfd$ (dashed lines) for different \LHSuniqueness levels. Bottom row shows $\synskewfd$ (solid lines) and $\synskewnonfd$ (dashed lines) for different \RHSskew levels. Some of the lines are hidden behind each other. The $\synnoisenonfd$ lines of $\SHANNONGONE$ and $\RFIplus$ are behind $\RFInorm$. The $\synnoisefd$ line of $\pdep$ is behind $\tau$. The $\synnoisenonfd$ line of $g_1,g'_1$ is behind its own $\synnoisefd$ line. The $\synuniquenonfd$ line of $\RFIplus$ is behind $\RFInorm$. The $\synuniquefd$ line of $\pdep$ is behind $\tau$. The $\synskewnonfd$ lines of $\RFIplus$ and $\RFInorm$ are behind $\SFI$. The $\synskewnonfd$ line of $g_1,g'_1$ is behind the $\synskewfd$ line of itself. }
	\vspace{1.5ex}
	%\caption{\label{fig:synzero:copy} Average measure value for $\synnoisefd$ (top row, solid line) and $\synnoisenonfd$  (top row, dashed line) %in \synzero
	%	for different error levels. }
	%\caption{\label{fig:lhs:copy} Average measure value for $\synuniquenonfd$  (solid line) and $\synuniquefd$ (dashed line) %in \synlhs 
	%	for different \LHSuniqueness levels.}
	%\caption{\label{fig:rhs:copy} Average measure value for $\synskewfd$ (solid line) and $\synskewnonfd$ (dashed line) in \synrhs for different \RHSskew levels.}

      \end{figure*}

We provide more detail on the material presented in Section~\ref{sec:syn:res}. Specifically, for each synthetic benchmark $\benchmark$ Figure~\ref{fig:synexperiments-additional} shows separate plots of the average measure values on $\benchmarkfd$ (in solid lines) and $\benchmarknonfd$ (in dashed lines).
Values for $g_1$ and $g'_1$
are grouped together as their average measure values are the same on both
$\benchmarkfd$ and $\benchmarknonfd$.  From Figure~\ref{fig:synexperiments-additional} we make the following additional observations w.r.t. \LHSuniqueness and \RHSskew.

\paragraph*{\LHSuniqueness.}
For $\rho$, $g_2$, $\SHANNONGONE$, $g_3$, $\FI$, $\pdep$, and $\tau$ the average measure values on $\synuniquenonfd$ increase,  % with the \LHSuniqueness level
eventually approaching the measure values on $\synuniquefd$. For $\RFIplus$ and $\SFI$, by contrast, the average measure values over $\synuniquefd$ decrease towards zero for increasing \LHSuniqueness, eventually reaching the value on $\synuniquenonfd$. Note that this decrease is already observable for  small \LHSuniqueness values.
This means that all mentioned measures %mentioned in this paragraph 
are biased w.r.t.\ \LHSuniqueness (either by inflating the scores of non-FDs or deflating the scores of FDs). % It will therefore prove problematic to discover non-linear \afds by means of these measures,  especially when the number of attributes in the LHS increases and \LHSuniqueness is expected to increase naturally to $1$.

\paragraph*{\RHSskew.}
% The bottom row of Figure~\ref{fig:synexperiments} shows the average measure value for $\synskewfd$ and $\synskewnonfd$ for different levels of \RHSskew.
% %
% For $g_1$, $g'_1$ and \SFI we see the same behavior as before:
% the solid and dotted lines coincide or are very close together
% which severely hampers their distinguishing power. As before we indicate the corresponding cells in Table~\ref{table:comparison} with the symbol \nap.
%
The distinguishing power of all \simple measures, as well as $\SHANNONGONE$, and
$\pdep$ drops when \RHSskew increases. Indeed, over $\synskewfd$ the average
measure values remains relatively constant as \RHSskew increases, while the
average measure values increases and approaches the values over $\synskewfd$.
These measures are thus biased w.r.t.\ \RHSskew: their score for $\X\to\Y$
increases solely on the basis of $\Y$ and independent of $\X$ even if relations
are generated by a process that sampled $\X$ and $\Y$ independently at random.
We observe that the \shannon measures (save $\SHANNONGONE$), $\tau$, and
$\muplus$ correct for this behavior. Furthermore, $\FI$ and $\tau$ slightly
decrease at higher levels of \RHSskew.

%%% Local Variables:
%%% mode: latex
%%% TeX-master: "../main/main"
%%% End:

  \section{Additional information on the evaluation on Real-World Data}
  \label{sec:rwd-extra}
  We provide more detail on the material presented in Section~\ref{sec:eval-rwd}.

\paragraph*{Precision-Recall Curves.}

The precision-recall curves shown in Figure~\ref{fig:pr-curve:auc} supplement our statements based on the area under the curve values. They constitute the graphical counterpart for Table~\ref{tab:auc}. It is clear that for each class one of the measures stands out; $g'_3$ for \simple, $\RFInorm$ for \shannon and $\muplus$ for \logical.

%%%%%%%%%%%% HACK: negative whitespace
%%%%%%%%%%%%%%%%%%%%%%%%%%%%%%%%%%%%%%%
\addtolength{\abovecaptionskip}{-3mm}
%%%%%%%%%%%%%%%%%%%%%%%%%%%%%%%%%%%%%%%
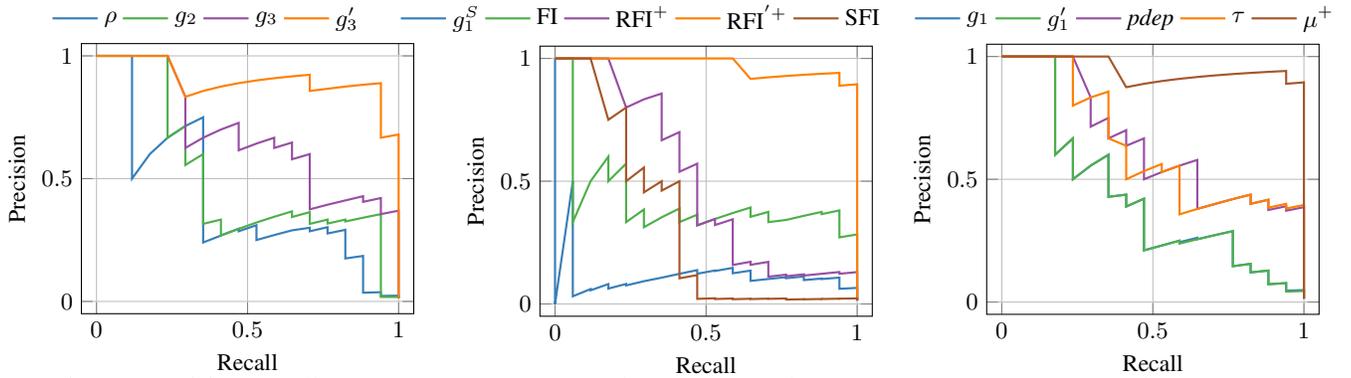
\begin{figure*}
	\small
	\centering
	\pgfplotsset{
		width=6cm,
		every axis plot/.append style={thick},
		xlabel=Recall,
		ylabel=Precision,
		every axis y label/.style={
				at={(ticklabel cs:0.5)},rotate=90,anchor=near ticklabel,
			},
		every axis x label/.style={
				at={(ticklabel cs:0.5)},anchor=near ticklabel,
			},
		xmin=-0.05, xmax=1.05,
		ymin=-0.05, ymax=1.05,
		ymajorgrids=true,
		xmajorgrids=true,
		legend columns=-1,
		legend style={draw=none},
	}
	\begin{minipage}{0.32\textwidth}
		\centering
		\ref{rwd_prcurves_simple_legend}
		\\
		\begin{tikzpicture}
			\begin{axis}[
					legend to name=rwd_prcurves_simple_legend,
				]
				\addplot [plotColor1] table [x=recall,y=precision] {plotdata/rwd_rho.dat};
				\addlegendentry{$\rho$}
				\addplot [plotColor2] table [x=recall,y=precision] {plotdata/rwd_g2.dat};
				\addlegendentry{$g_2$}
				\addplot [plotColor3] table [x=recall,y=precision] {plotdata/rwd_g3.dat};
				\addlegendentry{$g_3$}
				\addplot [plotColor4] table [x=recall,y=precision] {plotdata/rwd_g3_prime.dat};
				\addlegendentry{$g'_3$}
			\end{axis}
		\end{tikzpicture}
	\end{minipage}
	\begin{minipage}{0.34\textwidth}
		\centering
		%%%%% HACK: legative whitespace to shift the legend to the left
		\hspace{-2em}\ref{rwd_prcurves_shannon_legend}
		\\
		\begin{tikzpicture}
			\begin{axis}[
					legend to name=rwd_prcurves_shannon_legend,
				]
				\addplot [plotColor1] table [x=recall,y=precision] {plotdata/rwd_shannon_g1_prime.dat};
				\addlegendentry{$\SHANNONGONE$}
				\addplot [plotColor2] table [x=recall,y=precision] {plotdata/rwd_fraction_of_information.dat};
				\addlegendentry{$\FI$}
				\addplot [plotColor3] table [x=recall,y=precision] {plotdata/rwd_reliable_fraction_of_information_prime.dat};
				\addlegendentry{$\RFIplus$}
				\addplot [plotColor4] table [x=recall,y=precision] {plotdata/rwd_fraction_of_information_prime.dat};
				\addlegendentry{$\RFInorm$}
				\addplot [plotColor5] table [x=recall,y=precision] {plotdata/rwd_smoothed_fraction_of_information.dat};
				\addlegendentry{$\SFI$}
			\end{axis}
		\end{tikzpicture}
	\end{minipage}
	\begin{minipage}{0.32\textwidth}
		\raggedright
		\hspace{1em}\ref{rwd_prcurves_logical_legend}
		\\
		\begin{tikzpicture}
			\begin{axis}[
					legend to name=rwd_prcurves_logical_legend,
				]
				\addplot [plotColor1] table [x=recall,y=precision] {plotdata/rwd_g1.dat};
				\addlegendentry{$g_1$}
				\addplot [plotColor2] table [x=recall,y=precision] {plotdata/rwd_g1_prime.dat};
				\addlegendentry{$g'_1$}
				\addplot [plotColor3] table [x=recall,y=precision] {plotdata/rwd_pdep.dat};
				\addlegendentry{$\pdep$}
				\addplot [plotColor4] table [x=recall,y=precision] {plotdata/rwd_tau.dat};
				\addlegendentry{$\tau$}
				\addplot [plotColor5] table [x=recall,y=precision] {plotdata/rwd_mu_prime.dat};
				\addlegendentry{$\muplus$}
			\end{axis}
		\end{tikzpicture}
	\end{minipage}
	\caption{\label{fig:pr-curve:auc} Precision-Recall curves over \rwdminus grouped per measure class: \simple, \logical, and \shannon.}
\end{figure*}
%%%%%%%%%%%% HACK: negative whitespace
%%%%%%%%%%%%%%%%%%%%%%%%%%%%%%%%%%%%%%%
\addtolength{\abovecaptionskip}{+3mm}
%%%%%%%%%%%%%%%%%%%%%%%%%%%%%%%%%%%%%%%

\paragraph*{Measure Runtimes.} Table~\ref{tab:runtimes:rwd} shows the runtimes of each measure in more detail. We observe that $\rho$ is the fastest measure with $110$ seconds to calculate a value for all \TotalNumberOfInvalidFDs candidate \fds. In general, the measures of the \simple class are faster compared to the others. Measures of the \logical class take on average roughly $21$ seconds longer to compute values for all candidate \fds. In the \shannon class, the differences are much larger. While $\SHANNONGONE$ and $\FI$ achive runtimes comparable to the measures from the \logical class, $\RFIplus$, $\RFInorm$ and $\SFI$ were not able to calculate values for all candidates. In that, $\SFI$ is able to calculate a value for $1430$ (roughly $90\%$) candidates while $\RFIplus$ and $\RFInorm$ finish only $250$ (roughly $15\%$) candidates.

\begin{table}
	\centering
	\small
	\caption{\label{tab:runtimes:rwd} $\RFIplus$, $\RFInorm$ and $\SFI$ are significantly slower than the other measures. The table shows the runtimes, capped at 24 hours, and the number of measures \afd candidates within the runtime.}
	\begin{filecontents*}{tables/rwd_runtimes.csv}
,runtime,candidates
$\rho$,110s,1634
$g_2$,130s,1634
$g_3$,118s,1634
$g'_3$,128s,1634
$\SHANNONGONE$,137s,1634
$\FI$,154s,1634
$\RFIplus$,24h,250
$\RFInorm$,24h,250
$\SFI$,24h,1430
$g_1$,134s,1634
$g'_1$,135s,1634
$\pdep$,135s,1634
$\tau$,151s,1634
$\muplus$,157s,1634
\end{filecontents*}
\csvautobooktabular[]{tables/rwd_runtimes.csv}
\end{table}

\paragraph*{RWD.}
Next we consider the whole of \rwd and disregard the measures $\RFIplus$ and $\RFInorm$ for which the scores of not all candidate \fds could be computed.
A comparison between the AUC values of \rwdminus and \rwd in Table~\ref{tab:aucs:rwd-appendix} show only minor differences. $\SFI$ shows the largest differences ($-0.036$), followed by $g_1$ and $g'_1$ ($-0.026$ and $-0.027$ respectively) and $\FI$ with a difference of $-0.019$. All other measures show differences below $-0.01$. We also note that all differences are negative, which means that each measure achieved better results when the \texttt{NULL} values of $\RFIplus$ and $\RFInorm$ are left out. We do not observe any notable changes for any of the tables of \rwd.

\begin{table*}
	%%%% RWD AUC TABLE
	\centering
	\small
	\caption{\label{tab:aucs:rwd-appendix} The differences between \rwdminus and \rwd without $\RFIplus$ and $\RFInorm$ are minimal. The table shows the same overview as Table~\ref{tab:auc} for \rwd. $\RFIplus$ and $\RFInorm$ show the same values, as the candidates do not change compared to \rwdminus.}
	\begin{tabular}{l l l l l l l l l l l l l}
		\toprule
		 & \rwd & \adultDataNumber & \claimsDataNumber & \dblpDataNumber & \hospitalDataNumber & \taxDataNumber & \gathAgentDataNumber & \gathAreaDataNumber & \gathDataNumber & \identTaxonDataNumber & \identDataNumber & best (pct) \\
		\midrule
		\begin{filecontents*}{tables/appendix_rwd_aucs.csv}
measure,allrwd,adult,claims,dblp,hospital,tax,gathAgent,gathArea,gath,identTaxon,ident,best
$\rho$,0.411,\underline{\textbf{1.000}},0.029,0.038,\underline{\textbf{1.000}},\underline{\textbf{1.000}},\underline{\textbf{1.000}},0.667,\underline{\textbf{1.000}},\underline{\textbf{1.000}},\underline{\textbf{1.000}},70
$g_2$,0.497,\underline{\textbf{1.000}},0.507,0.038,\underline{\textbf{1.000}},\underline{\textbf{1.000}},\underline{\textbf{1.000}},0.667,\underline{\textbf{1.000}},\underline{\textbf{1.000}},\underline{\textbf{1.000}},70
$g_3$,0.669,\underline{\textbf{1.000}},\underline{\textbf{1.000}},0.148,\underline{\textbf{1.000}},\underline{\textbf{1.000}},0.792,0.667,\underline{\textbf{1.000}},\underline{\textbf{1.000}},\underline{\textbf{1.000}},70
$g'_3$,0.901,\underline{\textbf{1.000}},\underline{\textbf{1.000}},\underline{\textbf{1.000}},\underline{\textbf{1.000}},\underline{\textbf{1.000}},0.792,0.667,\underline{\textbf{1.000}},\underline{\textbf{1.000}},\underline{\textbf{1.000}},80
$\SHANNONGONE$,0.108,\underline{\textbf{1.000}},\underline{\textbf{1.000}},0.027,\underline{\textbf{1.000}},\underline{\textbf{1.000}},0.633,0.667,0.100,\underline{\textbf{1.000}},\underline{\textbf{1.000}},60
\FI,0.396,\underline{\textbf{1.000}},\underline{\textbf{1.000}},0.049,\underline{\textbf{1.000}},\underline{\textbf{1.000}},\underline{\textbf{1.000}},\underline{\textbf{0.917}},\underline{\textbf{1.000}},\underline{\textbf{1.000}},\underline{\textbf{1.000}},\textbf{90}
\RFIplus,0.494,\underline{\textbf{1.000}},\underline{\textbf{1.000}},0.182,0.508,\underline{\textbf{1.000}},0.133,0.667,0.250,\underline{\textbf{1.000}},\underline{\textbf{1.000}},50
\RFInorm,\underline{\textbf{0.971}},\underline{\textbf{1.000}},\underline{\textbf{1.000}},\underline{\textbf{1.000}},\underline{\textbf{1.000}},\underline{\textbf{1.000}},\underline{\textbf{1.000}},\underline{\textbf{0.917}},\underline{\textbf{1.000}},\underline{\textbf{1.000}},\underline{\textbf{1.000}},\underline{\textbf{100}}
\SFI,0.284,\underline{\textbf{1.000}},\underline{\textbf{1.000}},0.003,0.296,\underline{\textbf{1.000}},0.050,0.583,0.056,\underline{\textbf{1.000}},\underline{\textbf{1.000}},50
$g_1$,0.399,\underline{\textbf{1.000}},0.587,0.014,\underline{\textbf{1.000}},\underline{\textbf{1.000}},0.598,0.667,\underline{\textbf{1.000}},\underline{\textbf{1.000}},\underline{\textbf{1.000}},60
$g'_1$,0.398,\underline{\textbf{1.000}},0.548,0.014,\underline{\textbf{1.000}},\underline{\textbf{1.000}},0.598,0.667,\underline{\textbf{1.000}},\underline{\textbf{1.000}},\underline{\textbf{1.000}},60
$\pdep$,0.642,\underline{\textbf{1.000}},\underline{\textbf{1.000}},0.076,\underline{\textbf{1.000}},\underline{\textbf{1.000}},0.792,0.667,\underline{\textbf{1.000}},\underline{\textbf{1.000}},\underline{\textbf{1.000}},70
$\tau$,0.623,\underline{\textbf{1.000}},\underline{\textbf{1.000}},0.082,\underline{\textbf{1.000}},\underline{\textbf{1.000}},\underline{\textbf{1.000}},0.667,\underline{\textbf{1.000}},\underline{\textbf{1.000}},\underline{\textbf{1.000}},80
$\muplus$,\textbf{0.946},\underline{\textbf{1.000}},\underline{\textbf{1.000}},\underline{\textbf{1.000}},\underline{\textbf{1.000}},\underline{\textbf{1.000}},\underline{\textbf{1.000}},0.667,\underline{\textbf{1.000}},\underline{\textbf{1.000}},\underline{\textbf{1.000}},\textbf{90}
\end{filecontents*}
\csvreader[head to column names,late after line=\\]{tables/appendix_rwd_aucs.csv}{}{\csvlinetotablerow}
		\bottomrule
	\end{tabular}
\end{table*}

\paragraph*{RWD without \rwdminus.}
To get more insight into the candidate \fds that could not be calculated by $\RFIplus$ and $\RFInorm$, we investigated the measure results for this exact subset: $\rwd \setminus \rwdminus$. This subset contains $405$ candidates that are exclusively non-\afds. For this reason, we cannot assess precision, recall or other metrics. However, we can derive some intuition about how difficult those candidates are by investigating the measure values of the other measures. The upper part of Table~\ref{tab:values:rfinull} shows an overview of summary statistics of $\rwd \setminus \rwdminus$.

We observe that both average and median of most measure values are low. Still, it is noteworthy that judging by very high max values there seems to be a small subset of candidate \fds that could lead to changes the AUCs of $\RFIplus$ and $\RFInorm$. Ultimately, the only way to confirm that suspicion would be to use $\RFIplus$ and $\RFInorm$ to calculate values for all candidates of $\rwd \setminus \rwdminus$.

In addition, we investigated the properties of $\rwd \setminus \rwdminus$, especially with regard to the properties we used in Section~\ref{sec:eval-syn}. The lower part of Table~\ref{tab:values:rfinull} shows an overview of summary statistics of the number of tuples, \LHSuniqueness and \RHSskew. Both \LHSuniqueness and \RHSskew do not show unexpected values. \LHSuniqueness is rather low, \RHSskew is around $1.0$, both of which match properties of \rwd in general. The number of tuples is rather large, indicating that having a lot of values increases the runtimes of $\RFIplus$ and $\RFInorm$.

\begin{table}
	\centering
	\small
	\setlength\tabcolsep{4.0pt}
	\caption{\label{tab:values:rfinull} The candidate \fds that $\RFIplus$ and $\RFInorm$ could not calculate have the potential to decrease their performance. The upper part shows summary statistics of all measures for $\rwd \setminus \rwdminus$. The lower part shows summary statistics of properties of the candidates, indicating that a high tuple count increases the runtime of $\RFIplus$ and $\RFInorm$.}
	\begin{tabular}{l l l l l l}
		\toprule
		 & mean & std & min & median & max \\
		\midrule
		\begin{filecontents*}{tables/appendix_rfinull_stats.csv}
,mean,std,min,median,max
$\rho$,0.163,0.285,0.000,0.010,0.997
$g_2$,0.117,0.248,0.000,0.000,0.995
$g_3$,0.215,0.299,0.000,0.041,0.998
$g'_3$,0.142,0.236,0.000,0.013,0.982
$\SHANNONGONE$,0.069,0.219,0.000,0.000,0.995
$\FI$,0.343,0.339,0.000,0.226,1.000
$\SFI$,0.056,0.142,0.000,0.004,0.871
$g_1$,0.791,0.260,0.071,0.954,1.000
$g'_1$,0.789,0.262,0.071,0.953,1.000
$\pdep$,0.187,0.290,0.000,0.021,0.998
$\tau$,0.164,0.281,0.000,0.012,0.998
$\muplus$,0.089,0.198,0.000,0.000,0.981
\end{filecontents*}
\csvreader[head to column names,late after line=\\]{tables/appendix_rfinull_stats.csv}{}{\csvlinetotablerow}
		\midrule
		\midrule
		\begin{filecontents*}{tables/appendix_rfinull_props.csv}
,mean,std,min,media,max
tuples,377989,432051,10000,114869,1000000
\LHSuniqueness,0.095,0.235,0.000,0.000,0.981
\RHSskew,0.967,1.260,0.000,0.687,8.639
\end{filecontents*}
\csvreader[head to column names,late after line=\\]{tables/appendix_rfinull_props.csv}{}{\csvlinetotablerow}
		\bottomrule
	\end{tabular}
\end{table}

%%% Local Variables:
%%% mode: latex
%%% TeX-master: "abstract"
%%% End:

  \section{Evaluation on \rwd with extra errors}
  \label{sec:nrwd}
  
%\paragraph{Real world data with errors (\nrwd)}
\label{sec:nrwd_data}
To study the measures' sensitivity to different kinds and different levels of
errors on real world data, we created the benchmark \nrwd. We obtain
\nrwd by passing the relations $R \in \rwd$ through a controlled error channel such that, denoting by $R'$ the obtained relation, some \fds in
$\perfectFD(R)$ do not hold anymore in $R'$ and hence become part of
$\designAFD(R')$. Existing AFDs are always maintained, i.e.,
$\designAFD(R) \subseteq \designAFD(R')$.

%
%Let $R$ be a relation in \rwd. 
We actually have multiple error channels, which are parameterized by an error
level $\eta\in [0,1]$ and an error type. When passing $R$ through the channel we
consider all $\X \to \Y \in \perfectFD(R)$ and modify $k=\lfloor \eta|R|\rfloor$
$\Y$-values. To avoid interference, we select at most one \fd $\X\to \Y$ for
every unique $\Y$ per relation, ensuring that $\Y$ does not appear in $\designAFD(R)$, and that no \fd $\Y\to \Z$ has previously been selected.
The procedure to modify the $\Y$ values %of a picked tuple 
is determined by the chosen type of data error for which we consider three categories inspired by Arocena et al.~\cite{DBLP:journals/pvldb/ArocenaGMMPS15}: \copykind error, \typokind and \boguskind value.
%we generate duplicated values (\texttt{copy}), typos (\texttt{typo}) and bogus values (\texttt{bogus}) as noise. 
For a chosen tuple $\w \in R$, only $\restr{\w}{\Y}$ is changed, where the change depends on the data error type:
%, i.e., $\noise{\w}|_{\W\setminus\Y}=\w|_{\W\setminus\Y}$, 
% and a new value for $\w|_{\Y}$ is determined depending on the error type as follows:
\begin{compactenum}[(i)]
  \item \copykind: Randomly pick any $\tilde{\w} \in R$ with $\tilde{\w}|_{\Y} \neq \w|_{\Y}$ and
  make $\tilde{\w}|_{\Y}$ the new value for $\w|_{\Y}$.
  \item  \typokind: To every $\y\in\dom_R(\Y)$, we associate three new values representing three common typos. From these, one is chosen each time at random as the new value for ${\w}|_{\Y}$.
  \item \boguskind: ${\w}|_{\y}$ is assigned a unique newly generated value.
\end{compactenum}

We point out that \copykind does not introduce any new values and keeps $\dom_{R}(\Y)$ stable, while \typokind (resp., \boguskind) introduces a number of new values independent of (resp., dependent on) the error level. $\X$ is not modified, and therefore $\prob_{R'}(\X) = \prob_{R}(\X)$.
To ensure that increasing error levels do not accidentally reduce errors, we
ensure that, for each $\x\colon \X$ we pick at most $\lfloor N_{\x}/2\rfloor$
tuples $\w$ with $\w|_{\X}=\x$ to modify, where $N_{\X}$ is the number of times
that $\x$ occurs in $\pi_{\X}(R)$. $\perfectFD$s for which this cannot be
guaranteed are omitted. The number of new \afds that can be constructed
therefore depends on the error level.

We consider four error levels: 1\%, 2\%, 5\% and 10\%. For each type of data error $t$ and each error level $\eta$, we obtain a new benchmark $\nrwd[t,\eta]$. Consequently, we generate 12 \nrwd tables per \rwd table $R$ for which $|\perfectFD(R)|>0$ (so, tables \gathDataNumber\ and \identTaxonDataNumber\ are excluded). Overall the number of \afds increases from
\BenchTotalAFDs in \rwd to \BenchTotalAFDsRWDErrorLevelOne in
$\nrwd[\copykind, 1\%]$. That number is the same for the other error types but can drop a little for higher noise levels as explained above. Similar to $\rwdminus$, we retain only \afd candidates where each measure calculates a value in a reasonable amount of time. 
A complete overview of the number of additional \afds per relation and per error level is given in \cite{artifacts}.
Per combination of parameters, the ground truth then consists of the thus constructed \afds together with the \afds from $R$.

%%% Local Variables:
%%% mode: latex
%%% TeX-master: "../main/main.tex"
%%% End:

  \begin{table*}[htp]
	%%%% RWDminus polluted AUC TABLE
	\centering
	\small
	\setlength{\tabcolsep}{4pt}
	\caption{\label{tab:auc:rwdn} $\muplus$ and $\RFInorm$ score best $\nrwd$ with one exception. The table shows the AUCs for each noise type and noise level of \nrwd, the dataset derived from \rwd by introducing errors. in the table header, the sizes $n$ of each of the \nrwd datasets is shown. For comparison, the first column repeats the AUCs for \rwdminus.}
	\begin{tabular}{l l l l l l l l l l l l l l}
		\toprule
		    & \rwdminus & $\copykind,1$ & $\copykind,2$ & $\copykind,5$ & $\copykind,10$ & $\boguskind,1$ & $\boguskind,2$ & $\boguskind,5$ & $\boguskind,10$ & $\typokind,1$ & $\typokind,2$ & $\typokind,5$ & $\typokind,10$ \\
		$n$ & 1229      & 1204          & 1206          & 1218          & 1143           & 1130           & 1127           & 1111           & 992             & 1188          & 1194          & 1189          & 915            \\
		\midrule
		\begin{filecontents*}{tables/appendix_rwdn_minus_aucs.csv}
measurename,rwdminuscomparison,copy1,copy2,copy5,copy10,bogus1,bogus2,bogus5,bogus10,typo1,typo2,typo5,typo10
$\rho$,0.417,0.345,0.229,0.161,0.100,0.241,0.172,0.101,0.069,0.288,0.204,0.132,0.073
$g_2$,0.504,0.341,0.258,0.205,0.169,0.235,0.238,0.180,0.215,0.265,0.235,0.169,0.102
$g_3$,0.674,0.595,0.435,0.323,0.251,0.513,0.342,0.250,0.265,0.540,0.362,0.256,0.163
$g'_3$,0.901,0.586,0.474,0.361,0.302,0.550,0.404,0.261,0.271,0.561,0.461,0.283,0.184
$\SHANNONGONE$,0.109,0.468,0.345,0.263,0.223,0.364,0.288,0.220,0.251,0.403,0.302,0.224,0.148
\FI,0.415,0.467,0.367,0.288,0.259,0.367,0.338,0.243,0.303,0.399,0.336,0.238,0.169
\RFIplus,0.494,0.403,0.401,0.378,0.410,0.327,0.289,0.244,0.224,0.392,0.377,0.333,\underline{\textbf{0.387}}
\RFInorm,\underline{\textbf{0.971}},\underline{\textbf{0.775}},\textbf{0.626},\textbf{0.513},\textbf{0.468},\underline{\textbf{0.622}},\textbf{0.497},\textbf{0.363},\underline{\textbf{0.372}},\underline{\textbf{0.689}},\textbf{0.523},\textbf{0.383},\textbf{0.304}
\SFI,0.320,0.238,0.240,0.230,0.256,0.082,0.080,0.074,0.102,0.169,0.170,0.169,0.238
$g_1$,0.425,0.369,0.312,0.252,0.213,0.277,0.271,0.215,0.177,0.316,0.296,0.215,0.119
$g'_1$,0.425,0.368,0.311,0.251,0.212,0.276,0.271,0.215,0.176,0.315,0.296,0.215,0.119
$\pdep$,0.647,0.517,0.391,0.283,0.239,0.412,0.325,0.235,0.253,0.443,0.344,0.237,0.150
$\tau$,0.630,0.630,0.473,0.333,0.275,0.459,0.365,0.260,0.299,0.491,0.379,0.261,0.182
$\muplus$,\textbf{0.946},\textbf{0.773},\underline{\textbf{0.637}},\underline{\textbf{0.550}},\underline{\textbf{0.476}},\textbf{0.618},\underline{\textbf{0.508}},\underline{\textbf{0.374}},\textbf{0.352},\textbf{0.668},\underline{\textbf{0.552}},\underline{\textbf{0.400}},0.294
\end{filecontents*}
\csvreader[late after line=\\]{tables/appendix_rwdn_minus_aucs.csv}{}{\csvlinetotablerow}
		\bottomrule
	\end{tabular}
\end{table*}

\begin{table}
	%%%% Winning Numbers of RWDminus polluted table
	\centering
	\small
	\caption{\label{tab:maxrecall:nrwd} On average, $\muplus$ and $\RFInorm$ score most reliably over \nrwd. The table shows the percentages per error type where a measure has the lowest rank to reach a recall of $1.0$.}
	\begin{tabular}{l l l l}
		\toprule
		 & \copykind & \boguskind & \typokind \\
		\midrule
		\begin{filecontents*}{tables/appendix_rwdn_winning_numbers.csv}
,copywinprc,boguswinprc,typowinprc
$\rho$,3.1,0.0,0.0
$g_2$,9.4,0.0,9.7
$g_3$,21.9,\textbf{42.3},48.4
$g'_3$,25.0,\textbf{42.3},41.9
$\SHANNONGONE$,25.0,19.2,32.3
$\FI$,40.6,38.5,35.5
$\RFIplus$,25.0,26.9,25.8
$\RFInorm$,43.8,\underline{\textbf{50.0}},\underline{\textbf{58.1}}
$\SFI$,12.5,15.4,16.1
$g_1$,12.5,7.7,9.7
$g'_1$,12.5,7.7,9.7
$\pdep$,21.9,26.9,29.0
$\tau$,\textbf{53.1},30.8,41.9
$\muplus$,\underline{\textbf{75.0}},\textbf{42.3},\textbf{54.8}
\end{filecontents*}
\csvreader[late after line=\\]{tables/appendix_rwdn_winning_numbers.csv}{}{\csvlinetotablerow}
		\bottomrule
	\end{tabular}
\end{table}

\paragraph*{AUC.}

Table~\ref{tab:auc:rwdn} lists AUC scores over \nrwd  per error type and for different error levels.
We observe that $\mu'$ has the highest AUC score in 6 out of 12 cases, followed by $\RFInorm$ that scores highest in 4 our of 12 cases. Both $\muplus$ and $\RFInorm$ are the top-two scoring measures consistently. The only exception is $\nrwd[\typokind,10\%]$ where $\RFIplus$ is the highest scoring measure. We also observe that $\muplus$ scores highest on all $\nrwd$ datasets with $\eta=2\%,5\%$ while $\RFInorm$ scores highest on all $\nrwd$ with $\eta=1\%$. There is no clear observation for $\eta=10\%$. Regarding the error types, we observe that in general the AUCs of \copykind are higher on the same error levels. For $\nrwd[\boguskind,10\%]$ $g_3$,$\SHANNONGONE$ and $\FI$ are exceptions. We observe a similar, yet less pronounced, trend for \typokind compared to \boguskind, which does not hold for $\eta=10\%$.
We remark that for some measures the AUC score on $\nrwd$ is larger at the 1\% error level than for \rwdminus. This is not completely unexpected as the ground truth for both is different. A suprising result is presented for $\SHANNONGONE$, where the AUC of each $\nrwd$ dataset is higher than for \rwdminus.

We do see that, as expected, the AUC score for each of the measures deteriorates at increasing error levels to an absolute low at error level 10\%. The exception is $\RFI'$ whose performance increases at higher error levels for error types \copykind but not for the other error types. When error types and error levels are unknown but expected to be small, $\mu'$ therefore remains the best choice of \afd-measure. Furthermore, it is evident from Table~\ref{tab:aucs:rwd-appendix} that \afd-measures are not very effective when error levels are greater than 5\%.

Within \simple $g'_3$ remains the best measure.
Within \shannon we see $\RFInorm$ scoring best while $\SFI$ performs worst for most datasets. The measures $\SHANNONGONE$, $\FI$ and $\RFIplus$ show very similar scores.
Within \logical we observe that $\tau$ is usually an improvement over $\pdep$, and  $\muplus$ is an improvement over $\tau$. Further, $g_1$ and $g'_1$ perform consistently worse than any other measure in \logical.

\paragraph*{Rank at max recall.}
We show in
Table~\ref{tab:maxrecall:nrwd} a qualitative comparison between measures by listing, for each measure $f$ and error type $t$, its \emph{winning number}, which is defined as follows. Consider a particular $(\text{relation}, t, \eta)$ combination in $\nrwd$. A measure  $f$ \emph{wins} this triple if its $\rankAtMaxRecall$ is minimal among all measures on this triple. The winning number of $f$ for error type $t$  is then the number of times $f$ wins, taken over all triples of type $t$.
Here, we see again that both $\RFInorm$ and $\muplus$ score very well. Across all error types, $\muplus$ is the only measure that is in included in the top two values. It scores best on \copykind, where its rank at max recall is minimal for $75\%$ of all relations. $\tau$ is the runner-up with $53.1\%$ of the relations. In \boguskind, $\RFInorm$ achives a minimal rank at max recall for $50.0\%$ of the relations. $g_3$, $g'_3$ and $\muplus$ follow with $42.3\%$. For \typokind, $\RFInorm$ ranks candidates best on $58.1\%$ relations. $\muplus$ is behind that with $54.8\%$.

%%% Local Variables:
%%% mode: latex
%%% TeX-master: "../main/main"
%%% End:

}

\end{document}